\documentclass[journal,draftcls,onecolumn,12pt,twoside]{IEEEtran}
\IEEEoverridecommandlockouts
\usepackage{color}
\usepackage[colorlinks, citecolor=blue, hyperindex]{hyperref}
\usepackage{amsmath}
\usepackage{amsfonts}
\usepackage{amssymb}
\usepackage{amscd}
\usepackage[dvips]{graphicx}
\usepackage{epsfig}
\usepackage{newlfont}
\usepackage{cite}
\usepackage{subfigure}
\newtheorem{theorem}{Theorem}
\newtheorem{lemma}{Lemma}

 \newtheorem{corollary}{Corollary}

\newcommand{\lb}{\left(}
\newcommand{\rb}{\right)}
\newcommand{\lsqb}{\left[}
\newcommand{\rsqb}{\right]}
\newcommand{\lcb}{\left\{}
\newcommand{\rcb}{\right\}}
\newcommand{\xbold}{\mathbf{x}}
\newcommand{\ybold}{\mathbf{y}}
\newcommand{\vonetwob}{\mathbf{v}_{12}}
\newcommand{\vtwooneb}{\mathbf{v}_{21}}
\newcommand{\coneuu}{C_1^{uu}}
\newcommand{\coneul}{C_1^{ul}}
\newcommand{\conelu}{C_1^{lu}}
\newcommand{\conell}{C_1^{ll}}
\newcommand{\ctwouu}{C_2^{uu}}
\newcommand{\ctwoul}{C_2^{ul}}
\newcommand{\ctwolu}{C_2^{lu}}
\newcommand{\ctwoll}{C_2^{ll}}
\newcommand{\mylceil}{\left\lceil}
\newcommand{\myrceil}{\right\rceil}
\newcommand{\pullUp}{\vskip -0.2cm}

\newcommand{\wprvone}{w_{p1}}
\newcommand{\wprvtwo}{w_{p2}}
\newcommand{\wcpone}{w_{cp1}}
\newcommand{\wcptwo}{w_{cp2}}
\newcommand{\tildewcpone}{\widetilde{w}_{cp1}}
\newcommand{\tildewcptwo}{\widetilde{w}_{cp2}}

\newcommand{\dumytwo}{w_{d2}}
\newcommand{\wprvonehat}{\hat{w}_{p1}}

\newcommand{\wcptwohat}{\hat{w}_{cp2}}
\newcommand{\tildewcponehat}{\hat{\widetilde{w}}_{cp1}}
\newcommand{\tildewcptwohat}{\hat{\widetilde{w}}_{cp2}}
\newcommand{\tildercpone}{\widetilde{R}_{cp1}}
\newcommand{\tildercptwo}{\widetilde{R}_{cp2}}

\newcommand{\dumytwohat}{\hat{w}_{d2}}
\newcommand{\mysum}{\displaystyle\sum}
\newcommand{\ubold}{\mathbf{u}}

\newcommand{\sumAEPweak}{\displaystyle\sum_{(\ybold_1^N,\xbold_{p1}^N,\ubold_1^N)\in T_{\epsilon}^{(N)}}}
\newcommand{\sumAEP}{\displaystyle\sum_{(\ybold_1^N,\xbold_{p1}^N,\ubold_1^N,\xbold_{d2}^N)\in T_{\epsilon}^{(N)}}}

\newcommand{\probdislmone}{P_{X_{p1}, X_{p2}, X_{d2}, U_2, Y_2}}
\newcommand{\typseqlmone}{\xbold_{p1}^N, \xbold_{p2}^N, \xbold_{d2}^N, \ubold_2^N, \ybold_2^N}
\newcommand{\sbold}{\mathbf{s}}
\newcommand{\zbold}{\mathbf{z}}
\newcommand{\SNRsum}{\text{SNR} + \text{INR}}
\newcommand{\SNRprod}{\sqrt{\text{SNR}\:\text{INR}}}
\newcommand{\sbar}{\bar{\mathbf{s}}}
\newcommand{\ybar}{\bar{\mathbf{y}}}
\newcommand{\sboldtilde}{\widetilde{\mathbf{s}}}
\newcommand{\coopsignal}{\mathbf{v}_{12}^N,\mathbf{v}_{21}^N}
\newcommand{\coopsignalone}{\mathbf{v}_{12}}
\newcommand{\coopsignaltwo}{\mathbf{v}_{21}}
\newcommand{\sbolddash}{\mathbf{s'}}

\newcommand{\SNRt}{\text{SNR}}
\newcommand{\INRt}{\text{INR}}
\newcommand{\wbold}{\mathbf{w}}
\ifCLASSINFOpdf
\else
\fi

\newcommand{\setxysizeo}{\epsfxsize=3.4in}

\normalsize

\begin{document}
%
\title{On the Capacity of the $2$-User Symmetric Interference Channel
with Transmitter Cooperation and Secrecy Constraints}
\date{}
\author{\authorblockN{Parthajit~Mohapatra$^{*}$ and Chandra R. Murthy$^{\ddag}$}\\
\authorblockA{$^*$Singapore University of Technology and Design, Singapore 487372\\
$^\ddag$Indian Institute of Science, Bangalore, India, 560012\\
Email: parthajit@sutd.edu.sg,  cmurthy@ece.iisc.ernet.in}
\thanks{Major portion of this work was carried out,  when the first author was at the department of ECE, Indian Institute of Science, Bangalore.}
}
%


\maketitle


%
\IEEEpeerreviewmaketitle
\begin{abstract}
This paper studies the value of limited rate cooperation between the transmitters for managing interference
and simultaneously ensuring secrecy, in the $2$-user Gaussian symmetric interference channel (GSIC). First, the
problem is studied in the symmetric linear deterministic IC (SLDIC) setting, and achievable schemes are
proposed, based on interference cancelation, relaying of the other user's data bits, and transmission of random
bits. In the proposed achievable scheme, the limited rate cooperative link is used to share a combination
of data bits and random bits depending on the model parameters.
Outer bounds on the secrecy rate are also derived, using a novel partitioning of the encoded messages
and outputs depending on the relative strength of the signal and the interference. The inner and outer bounds are derived
under all possible parameter settings. It is
found that, for some parameter settings, the inner and outer bounds match, yielding the capacity of the SLDIC
under transmitter cooperation and secrecy constraints. In some other scenarios, the achievable rate matches
with the capacity region of the $2$-user SLDIC without secrecy constraints derived by Wang and
Tse~\cite{wang-TIT-2011}; thus, the proposed scheme offers secrecy for free, in these cases. Inspired by the
achievable schemes and outer bounds in the deterministic case, achievable schemes and outer bounds are
derived in the Gaussian case. The proposed achievable scheme for the Gaussian case
is based on
Marton's coding scheme and stochastic encoding along with dummy message transmission. One of the key
techniques used in the achievable scheme for both the models is interference cancelation, which simultaneously offers two seemingly
 conflicting benefits: it cancels interference and ensures secrecy. Many of the
results derived in this paper extend to the asymmetric case also. The results show that limited transmitter
 cooperation can greatly facilitate secure communications over $2$-user ICs.
\end{abstract}
\section{Introduction}
Interference management and ensuring security of the messages are two important aspects in the design of
multiuser wireless communication systems, owing to the broadcast nature of the physical medium. The interference
channel (IC) is one of the simplest information theoretic models for analyzing the effect of interference on the throughput
and secrecy of a multiuser communication system. One way to
enhance the achievable rate with secrecy constraints at the receivers is through cooperation between the transmitters. In this work,
the role of transmitter cooperation in managing interference and ensuring secrecy is explored by studying the $2$-user IC with
limited-rate cooperation between the transmitters and secrecy constraints at the receivers. In practice, such
scenarios can arise in a cellular network, where different users have subscribed to different data contents, and are
served by different base stations belonging to the same service provider. In
this case, it is important for the service provider to support high throughput, as well
as secure its transmissions, to maximize its own revenue. In these scenarios, the transmitters (e.g., base stations)
are not completely isolated from each other, and cooperation among them is possible. As the base stations
can trust each other, there is no need for secrecy constraints at the transmitters. Such cooperation can
potentially provide significant gains in the achievable throughput in the presence of interference, while simultaneously
guaranteeing security.

To illustrate the value of transmitter cooperation in simultaneously managing interference and ensuring secrecy,
a snapshot of some of the results to come in the sequel is presented in Fig.~\ref{fig:prob-mot}. Here, the
capacity of the symmetric linear deterministic IC (SLDIC) with and without cooperation is plotted against
$\alpha \triangleq \frac{n}{m}$, where $m=(\lfloor 0.5 \log \text{SNR}\rfloor)^{+},\: n=(\lfloor 0.5 \log \text{INR}\rfloor)^{+},$
without any secrecy constraints at the receivers~\cite{wang-TIT-2011}. Also plotted are the outer bound and
the achievable rate for the $2$-user SLDIC \emph{with} secrecy constraints at the receivers, developed in
Secs.~\ref{sec:LDIC-outer} and \ref{sec:sldic-ach}, respectively. Two cases are considered: no transmitter
cooperation ($C=0$) and with cooperation between the transmitters ($C=\frac{m}{4}$ bits per channel use). The
outer bounds plotted for $C=0$ with secrecy constraints at receivers show that, as the value of $\alpha$ increases,
there is a dramatic
loss in the achievable rate compared to the case without the secrecy constraint. The performance significantly
improves with cooperation, and it can be seen that it is possible to achieve a nonzero secrecy rate for all values of
$\alpha$ except $\alpha=1$. This paper presents an in depth study of the interplay between interference, security,
and transmitter cooperation in the $2$-user IC setting. It demonstrates that having a secure cooperative link
in a network can significantly improve the achievable secrecy rate.
\begin{figure*}
\begin{center}
\hspace{4cm} \includegraphics[width=6in, height=3in]{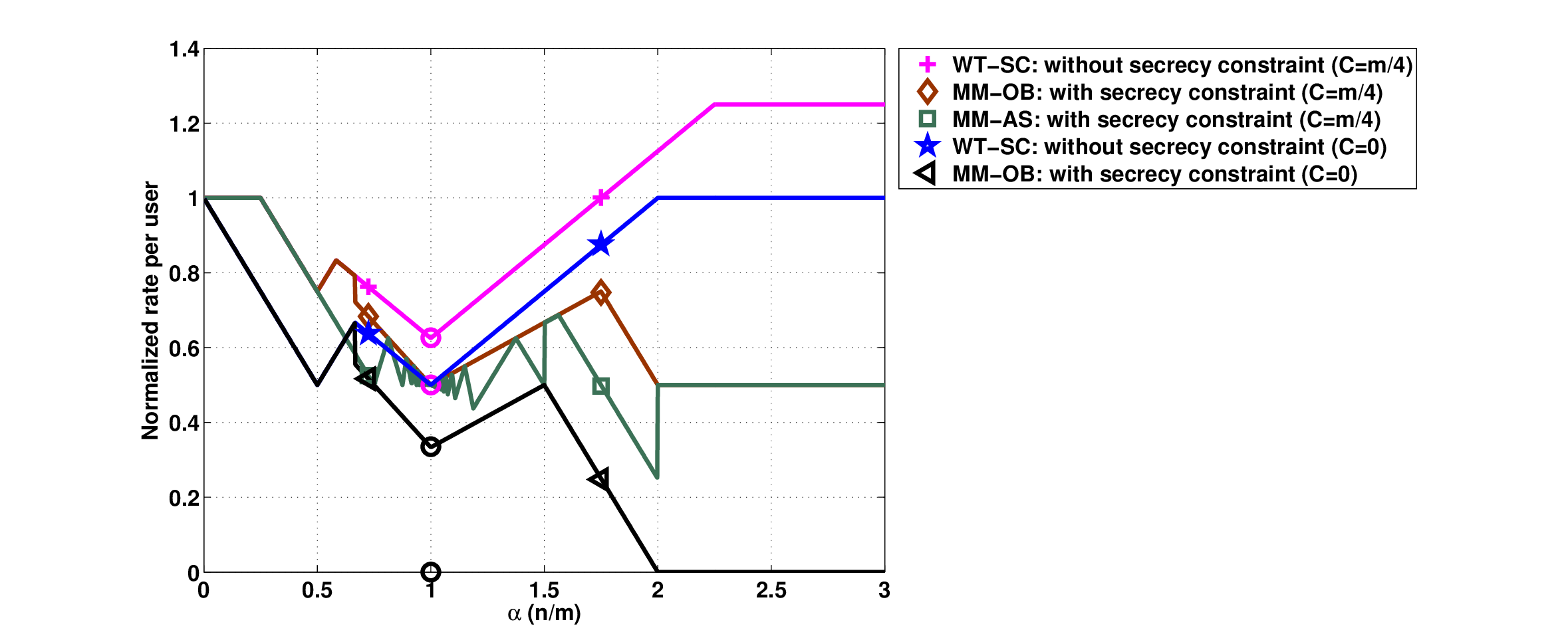}
 \caption{Data rate normalized by $m$ for the $2$-user SLDIC. Here, $C$ is the capacity of the cooperative link between the transmitters, $m=400$~bits, and $n$ is set based on the value of $\alpha = n/m$. }\label{fig:prob-mot}
\end{center}
\end{figure*}

\textit{Past work}: The interference channel has been extensively studied over the past few decades, to
understand the effects of interference on the performance limits of multi-user communication systems. The
capacity region of the Gaussian IC (GIC) without  secrecy constraints at receiver remains an open problem, even
in the $K=2$ user case, except for some special cases such as the strong/very strong interference regimes \cite{carleial-TIT-1975}, \cite{sato-TIT-1981}.
In \cite{liu-TIT-2008}, the broadcast and IC with independent confidential messages are considered and the
achievable scheme is based on random binning techniques. The work in \cite{tang-TIT-2011} demonstrates
that with the help of an independent interferer, the secrecy capacity region of the wiretap channel can be
enhanced. Intuitively, although the use of an independent interferer increases the interference at both the
legitimate receiver and the eavesdropper, the benefit from the latter outweighs the rate loss due to the former.
Some more
results on the IC under different eavesdropper settings can be found in \cite{koyluoglu-TIT-2011, koyluoglu2-TIT-2011, he-TIT-2011}.


The effect of cooperation on secrecy has been explored in \cite{ekrem-ciss-2008, ekrem-TIT-2011, awan-arxiv-2012}.
In \cite{ekrem-ciss-2008}, the effect of user cooperation on the secrecy capacity of the multiple access channel with
generalized feedback is analyzed, where the messages of the senders need to be kept secret from each other.
In \cite{ekrem-TIT-2011}, the role of user cooperation on the secrecy of broadcast channel (BC) with relaying, where the receivers
can cooperate with each other, is considered. The achievable scheme uses a combination of Marton's
coding scheme for the BC and a compress and forward scheme for the relay channel. The role of a relay
in ensuring secrecy under different wireless network settings has been studied in \cite{lai-TIT-2008,vaneet-eurasip-2009,perron-isit-2010}.

A linear deterministic model for relay network was introduced in \cite{avesti-TIT-2011}, which led to insights
on the achievable schemes in Gaussian relay networks. The deterministic model has subsequently been
used for studying the achievable rates with the secrecy constraints in \cite{yates-isit-2008, perron-infocom-2009, shamai-TIT-2012}.
In \cite{yates-isit-2008}, secret communication over the IC is analyzed with two types of secrecy constraints:
in the first case, the secrecy constraint is specific to the agreed-upon signaling strategy, and in the second
case, the secrecy constraint takes into account the fact that the other users may deviate from the agreed-upon
strategy. The deterministic model has also been studied under different eavesdropper settings in
\cite{li-isit-2008,perron-infocom-2009,shamai-TIT-2012}.

It is known that limited-rate cooperation between the transmitters or receivers can significantly increase the
rate achievable in the 2-user IC without secrecy constraints \cite{wang-TIT-2011,wang-TIT-2011-two}. In general, 
the Gaussian IC with transmitter cooperation is more
difficult to analyze than Gaussian IC with receiver cooperation, even when there is no secrecy
constraints at the receivers. For example, when the receivers can cooperate through a link of infinite capacity,
the model reduces to a Gaussian MIMO multiple access channel (MAC). When the transmitters cooperate through a
link of infinite capacity,
the model reduces to a MIMO BC. The capacity region of the general MAC was characterized in
1970s \cite{ahlswede-isit-1971, liao-1972-thesis}. In the MAC, the boundary of the rate region can be achieved if the receiver performs MMSE decoding and successive
interference cancelation of the input data streams. However, it took a long time for researchers to find a
precoding strategy which achieves the boundary of the BC rate region \cite{caire-TIT-2003, weingarten-TIT-2006}.  
Similarly, the IC with cooperative receivers 
is easier to analyze than the IC with cooperative transmitters \cite{wang-TIT-2011-two, wang-TIT-2011}. Further, when there are secrecy constraints
at the receivers, the following difficulties arise
in analyzing the system with rate-limited transmitter cooperation.
\begin{enumerate}
\item There are a number of ways in which the transmitters can use the cooperative link
for encoding their transmission. The cooperation can involve
  the exchange of data bits, random bits or any combination of the two.
\item It is difficult to obtain tractable outer bounds, since the encoded messages are no longer independent
  due to the cooperation between the transmitters. In addition to providing
  carefully selected side-information to receivers, the secrecy constraints at
  the receivers need to be exploited in a judicious manner to obtain tighter outer bounds as compared to the
  outer bounds that do not use
  the secrecy constraints at the receivers.
  \end{enumerate}
To the best of the authors' knowledge, the role of limited transmitter cooperation in a $2$-user IC on interference
management and secrecy has not been explored and is therefore focus of this work.

\textit{Contributions}: In order to make headway into this problem, first,
the problem is addressed in the linear deterministic setting. For the SLDIC with cooperating transmitters and
secrecy constraints at the receivers, achievable schemes and outer bounds on the secrecy rate are derived for
all possible parameter settings. This gives useful insights for the achievable schemes and outer bounds in the
Gaussian setting. Next, the schemes are adapted to the Gaussian case. The proposed transmission/coding
strategy in the Gaussian setting uses a superposition of a non-cooperative private codeword and a cooperative
private codeword. For the non-cooperative private part, stochastic encoding is used \cite{wyner-bell-1975}, and
for the cooperative private part, Marton's coding scheme is used \cite{marton-TIT-1979,wang-TIT-2011}. The
auxiliary codewords corresponding to the cooperative private part are chosen such that the interference caused by
the cooperative private auxiliary codeword of the other user is completely canceled out. This approach is
different from the one used in \cite{wang-TIT-2011}, where the interference caused by the unwanted auxiliary
codeword is approximately canceled. Further, one of the users transmits dummy information to enhance the
achievable secrecy rate. The major contributions of this work can be summarized as follows:
\begin{enumerate}
\item One of the key techniques used in the derivation of the outer bounds for the SLDIC is the  proposed partitioning
of the encoded messages and outputs depending
 on the value of $\alpha$. This partitioning of the encoded
 messages/outputs reveals  what side-information needs to be provided to the
 receivers for canceling negative entropy terms. In addition, partitioning helps to  bound or
 simplify entropy terms which are not easy to evaluate due to the dependence between the encoded messages
 at the transmitters. Also, the
 partitioning of the encoded messages/outputs provides a convenient handle for using the secrecy
 constraints at the receivers efficiently in deriving the outer bounds.  The outer bounds are stated as
 Theorems~\ref{th:theoremSLDIC-outer1}-\ref{th:theoremSLDIC-outer4}
 in Sec.~\ref{sec:LDIC-outer}.
\item For the SLDIC, the achievable scheme is based on interference cancelation, transmission
 of jamming signal (random bits) and relaying of the other user's data bits. The novelty in the proposed scheme
 lies in determining how to combine these techniques to achieve rates that are
 far superior to that achievable individually by these methods.  To the best
of authors' knowledge, exchanging a combination of data bits and random bits between the transmitters
for the purpose of precoding has not been used in the literature. The details of the achievable scheme can be
found in Sec.~\ref{sec:sldic-ach}.

\item Outer bounds on the secrecy rate in the Gaussian setting are derived and stated as
Theorems~\ref{th:theorem_GSIC_outer1}-\ref{th:theorem_GSIC_outer3} in Sec.~\ref{sec:outerGaussian}. As the
partitioning used in deriving the outer bounds for the deterministic case cannot be directly used in the Gaussian
case,  either analogous quantities as side-information need to be found to mimic the
partitioning of the encoded messages/outputs or the bounding steps need to be modified
taking cue from the deterministic model. This is one of the key steps in
deriving the outer bounds on the secrecy rate.

\item Using the intuition gained from the SLDIC, achievable schemes for the Gaussian case are proposed,
which use a combination of stochastic encoding and Marton's coding scheme along with dummy message
transmission by one of the users.  However, in the high
interference regime, stochastic encoding alone cannot ensure secrecy of the
non-cooperative private message, as cross links are stronger than the direct links. Hence, in addition to
stochastic encoding, dummy message transmission
is used by one of the users to ensure secrecy of the
non-cooperative private message at the unintended receiver. In the Marton's coding scheme, the codeword
carrying the cooperative private message is precoded such that it is completely canceled at the unintended
receiver. The details of the achievable
scheme can be found in Sec.~\ref{sec:achGaussian}.

\item Many of the results derived in this paper extend to the asymmetric case
also, and these are mentioned as remarks after corresponding theorems, where applicable.
\end{enumerate}

It is shown that with limited-rate transmitter cooperation, it is possible to achieve a nonzero secrecy rate under
all parameter settings except for the $\alpha=1$ case. In particular, for the very high interference regime $(\alpha \geq 2)$,
 it is possible to achieve non-zero secrecy rate for both the model as compared to the non-cooperating case. In
case of SLDIC, it is found, surprisingly, that in some
nontrivial cases, the achievable secrecy rate equals the capacity of the same system without the secrecy
constraints. Thus, the proposed schemes allow one to get secure communications for free, in these cases.
It is also observed that the proposed outer bounds for the SLDIC with cooperation are strictly tighter than the
 best existing outer bound without the secrecy constraint \cite{wang-TIT-2011} in all interference regimes,
 except for the weak interference regime, where the bounds match. The idea of using a common randomness 
 to improve the achievable rates is an important upcoming theme in multiuser information theory, and 
 the proposed schemes based on sharing random bits between the transmitters is in the same flavor. Thus, the results in this paper provide a
 deep and comprehensive understanding of the benefit of transmitter cooperation in achieving high data rate
 in the IC, while also ensuring secrecy. Parts of this work have appeared in~\cite{partha-spawc-2013} and
 \cite{partha-ncc-2014}.

\textit{Notation}: Lower case or upper case letters represent scalars, lower case boldface letters represent
vectors, and upper case boldface letters represent matrices.

\textit{Organization}: Section \ref{sec:sysmod} presents the system model. In Secs. \ref{sec:LDIC-outer} and
\ref{sec:sldic-ach}, the outer bounds and the achievable schemes for the SLDIC are presented, respectively.
The outer bounds  and  achievable results for the GSIC can be found in  Secs.~\ref{sec:outerGaussian} and
 ~\ref{sec:achGaussian}, respectively. In Sec. \ref{sec:results-discussion}, some numerical examples are
 presented to offer a deeper insight into the bounds, to contrast the performance of the various schemes, and
 to benchmark against known results. Concluding remarks are offered in Sec. \ref{sec:conc}. The proofs of the theorems and lemmas are presented in the Appendices.

\section{System Model}\label{sec:sysmod}
\begin{figure*}[t]
\centering
\mbox{\subfigure[]{\includegraphics[width=2.7in, height = 2.3in]{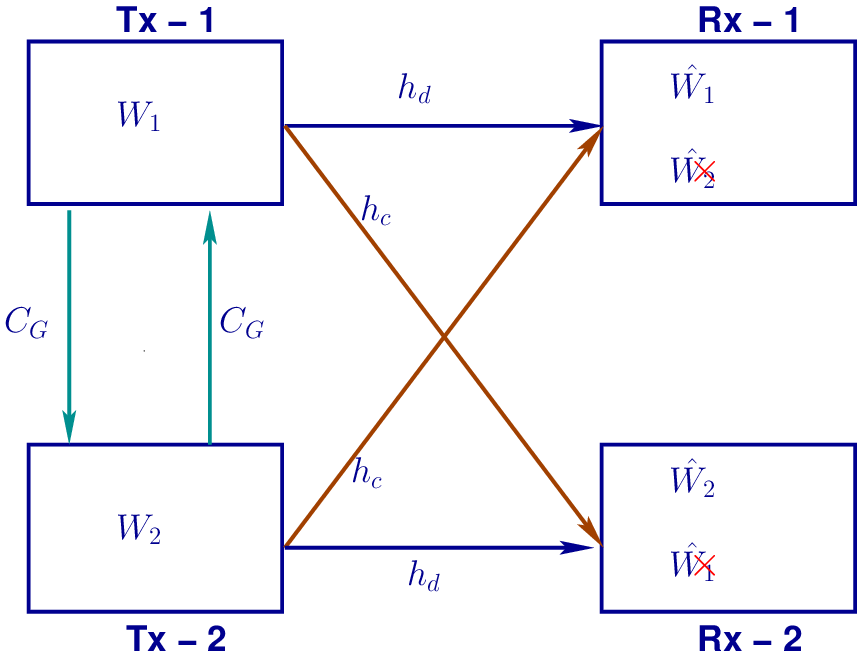}\label{fig:sysmodel1}}\qquad \quad
\subfigure[]{\includegraphics[width=2.7in, height = 2.2in]{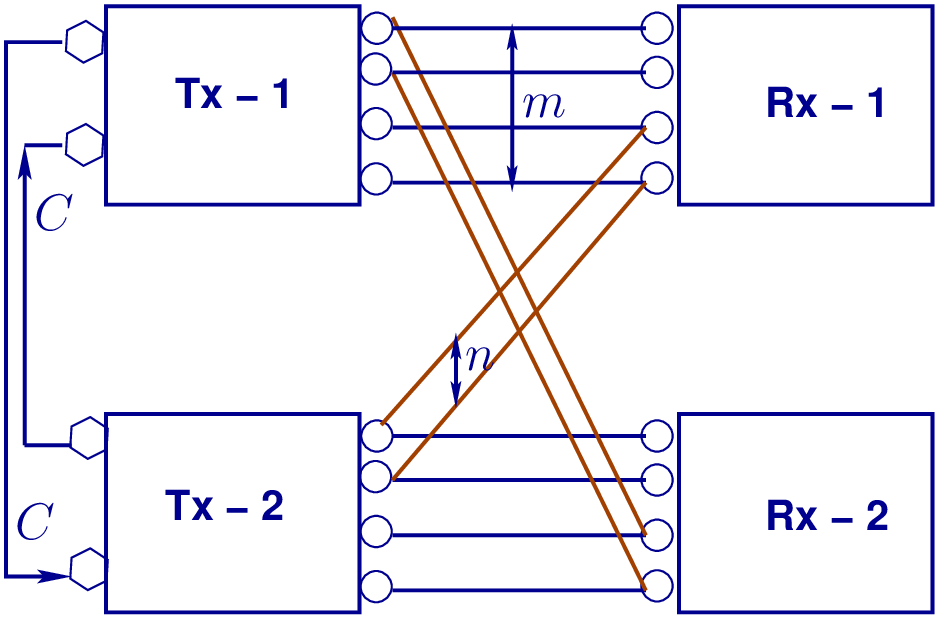}\label{fig:sysmodel2}}}
\caption[]{\subref{fig:sysmodel1} GSIC and \subref{fig:sysmodel2} SLDIC with transmitter cooperation.}\label{fig:sysmodel}
\vspace{-0.2cm}
\end{figure*}
Consider a $2$-user Gaussian symmetric IC (GSIC) with cooperating transmitters. The  signals at the receivers
are modeled as \cite{wang-TIT-2011}:
\begin{align}
 y_1 = h_{d}x_{1} + h_{c}x_2 + z_{1}; \quad
 y_2 = h_{d}x_{2} + h_{c}x_1 + z_{2}, \label{sysmodel1}
\end{align}
where $z_{j}\:(j=1,2)$ is the Gaussian additive noise, distributed as $z_{j} \sim \mathcal{N}(0,1)$. The input
signals are required to satisfy the power constraint: $E[|x_{i}|^{2}] \leq P$. Here, $h_d$ and $h_c$ are
 the channel gains of the direct and cross links, respectively. The transmitters cooperate through a noiseless
 and secure link of finite rate denoted by $C_{G}$. The equivalent deterministic model of (\ref{sysmodel1}) at
 high SNR is as follows~\cite{wang-TIT-2011}:
\begin{align}
\mathbf{y}_{1} = \mathbf{D}^{q-m}\mathbf{x}_{1} \oplus \mathbf{D}^{q-n}\mathbf{x}_{2}; \  \mathbf{y}_{2} = \mathbf{D}^{q-m}\mathbf{x}_{2} \oplus \mathbf{D}^{q-n}\mathbf{x}_{1}, \label{sysmodel2}
\end{align}
where $\mathbf{x}_{i}$ and $\mathbf{y}_{i}$ are binary vectors of length $q \triangleq \max\{m,n\}$, 
$\mathbf{D}$ is a $q \times q$ downshift matrix with elements $d_{j',j''}=1$ if $2 \leq j'=j''+1\leq q$ and 
$d_{j',j''}=0$ otherwise, and $\oplus $ stands for  modulo-$2$ addition (\textsf{XOR} operation).

The parameters $m$ and $n$ are related to the GSIC as $ m = (\lfloor 0.5 \log Ph_d^2\rfloor)^{+},\: n =  (\lfloor 0.5 \log Ph_c^2\rfloor)^{+},$
while the capacity of the cooperative link is $C = \lfloor C_{G} \rfloor$.
The quantity $\alpha \triangleq \frac{n}{m}$ captures the amount of coupling between the signal and the
interference, and is central to characterizing the achievable rates and outer bounds in case of the SLDIC and
GSIC. A  schematic representation of the GSIC and SLDIC with transmitter cooperation is shown in Fig.~\ref{fig:sysmodel}.
The figure also shows the convention followed in this paper for denoting the bits transmitted over the SLDIC,
which  is the same as that in \cite{wang-TIT-2011}. The bits $a_{i}, b_i \in \mathcal{F}_{2}$ denote the  information
bits of transmitters $1$ and $2$, respectively, sent on the $i^{\text{th}}$ level, with the levels numbered starting from
the bottom-most entry.

The transmitter $i$ has a message $W_{i}$, which should be decodable at the intended receiver~$i$, but needs to be kept
secret from the other, unintended receiver $j$, $j \neq i$. In the case of the SLDIC, the encoded message
$(\mathbf{x}_i)$ is a function of its own data bits, the bits received through the cooperative link, and possibly
some random data bits. The encoding at the transmitter should satisfy the causality constraint, i.e., it cannot
depend on future cooperative bits. The decoding is based on solving the linear equations in \eqref{sysmodel2}
at each receiver. For secrecy, it is required to satisfy $I(W_{i}; \mathbf{y}_{j}) = 0, i,j \in \{1,2\} \text{ and } i \neq j$ in the case of the SLDIC~\cite{shannon-bell-1949}. The details of the encoding and decoding scheme for the Gaussian case can be found in Sections \ref{sec:weak-mod-achGaussian} and \ref{sec:high-veryhigh-achGaussian}. In contrast to the SLDIC, the notion of weak secrecy is considered for the Gaussian case \cite{wyner-bell-1975}.
Also, it is assumed that the transmitters trust each other completely and that they do not deviate from the
agreed scheme, for both the models.

The results derived in the paper for the deterministic and Gaussian models under the symmetric assumption can be extended to the asymmetric setting in many cases, and these are indicated as remarks in the following sections. There are two ways in which the model considered in the paper can be asymmetric: (a) when $C_{12} \neq C_{21}$, where $C_{ij}$ is the capacity of the cooperative link from transmitter~$i$ to transmitter~$j$ $(i, j \in \{1, 2\}, i \neq j)$. This is termed as \emph{cooperation asymmetry}. (b) The two direct channel gains and two cross channel gains need not be equal to each other; this is termed as \emph{channel asymmetry}. In this case, the channel
is parameterized by $(m_{1}, n_{1}, m_2, n_2 )$ in the deterministic case and $(h_{11}, h_{12}, h_{22}, h_{21})$ in the Gaussian case. In the sequel, the phrase \emph{asymmetry} is used to account for both channel and cooperation asymmetry.
\section{SLDIC: Outer Bounds}\label{sec:LDIC-outer}
In this section, four outer bounds on the symmetric rate for the $2$-user SLDIC with
cooperation between transmitters and perfect secrecy constraints at the receivers are stated as
Theorems~\ref{th:theoremSLDIC-outer1}-\ref{th:theoremSLDIC-outer4}. Theorem~\ref{th:theoremSLDIC-outer1}
is valid for all $\alpha \geq 0$, while Theorems~\ref{th:theoremSLDIC-outer2}, \ref{th:theoremSLDIC-outer3}, and \ref{th:theoremSLDIC-outer4}
are valid for $\alpha\geq 2$, $1 < \alpha < 2$, and $\alpha=1$, respectively. 

In the derivation of the outer bounds, the following difficulties arise:
 \begin{enumerate}
 \item Due to cooperation between the transmitters, the encoded messages are no
 longer independent. Most existing outer bounding techniques (e.g.: \cite{tang-TIT-2011, liu-TIT-2008}) require the independence of
 the encoded messages to simplify the entropy terms, hence are not applicable in this case.
  \item Determining when and how to use the secrecy constraints at the receivers along with the reliability criteria is
  crucial in deriving a tractable outer bound.
 \end{enumerate}

To meet these challenges, a novel partitioning of the encoded messages and outputs depending on the value of $\alpha$ is proposed. This partitioning of the encoded messages/outputs reveals  what side-information needs to be provided to the receivers and helps to  bound or simplify entropy terms which are not easy to evaluate due to the dependence between the encoded messages at the transmitter. This partitioning also reveals how to judiciously exploit the secrecy constraints at the receivers in deriving the outer bounds.
 
The following relation is repeatedly used in the derivation of these outer bounds: conditioned on the cooperative signals,
denoted by $(\coopsignalone^N,\coopsignaltwo^N)$, the encoded signals and the messages at the two transmitters 
are independent \cite{willems-TIT-1983,wang-TIT-2011}. This is represented as the following Markov chain relationship:
\begin{align}
& (W_1, \xbold_1^N) - (\coopsignalone^N,\coopsignaltwo^N) - (W_2,\xbold_2^N). \label{eq:thouter0}
\end{align}
Finally, the overall outer bound on the symmetric secrecy rate is obtained by taking the minimum of these
outer bounds. The best performing outer bound depends on the value of $\alpha$ and the maximum possible
rate, i.e., $\max(m,n) \mathbf{1}_{\{C> 0\}} + \min(m,n)\mathbf{1}_{\{C=0\}}$ per user, where $\mathbf{1}_{A}$
is the indicator function, equal to $1$ if $A$ is true, and equal to $0$ otherwise.

In the derivation of the first outer bound, the encoded message $\xbold_i\:(i=1,2)$ is partitioned into two parts:
one part ($\xbold_{ia}$) which causes interference to the unintended receiver, and another part ($\xbold_{ib}$)
which is not received at the unintended receiver. Partitioning the message in this way helps to obtain an outer
bound on $2R_1 + R_2$, which leads to an outer bound on the symmetric secrecy rate. The following theorem
gives the outer bound on the symmetric secrecy rate.
\begin{theorem}\label{th:theoremSLDIC-outer1}
The symmetric rate of the $2$-user SLDIC with limited-rate transmitter cooperation and secrecy constraints at the receivers is upper bounded as:
\begin{align}
R_s \leq \left\{\begin{array}{l l}
    \frac{1}{3}\lsqb 2C + 3m-2n\rsqb &\mbox{for $\alpha \leq 1$ } \\ 	
    \frac{1}{3}\lsqb 2C + n\rsqb & \mbox{for $\alpha > 1$}.
 \end{array}\right. \label{eq:outerone_deter1}
\end{align}
\end{theorem}
\begin{proof}
The proof is provided in the Appendix~\ref{sec:theoremSLDIC-outer1}.
\end{proof}
\textit{Remarks:}
\begin{itemize}
\item Note that when $\alpha > 1$, the outer bound increases
with increasing $n$ for a given value of $C$. However, it is intuitive to think that the
achievable secrecy rate should decrease with increase in the value of $\alpha$, i.e., the outer bound is loose
in the high interference regime. Interestingly, it is found that the achievable secrecy rate also improves with
increase in the value $\alpha$ in the initial part of the high interference
regime, i.e., for $1 < \alpha < 2$, even when $C=0$. This will be discussed in Sec.~\ref{sec:numerical-SLDIC}.
\item  The outer bound stated above can be extended to obtain an outer bound on $2R_1 + R_2$
for the asymmetric setting. Using a similar approach as used in the proof of this theorem,  one can also obtain
 an outer bound on $R_1 + 2R_2$. Note that, these outer bounds are applicable
 over all the interference regimes. The outer bounds are as follows:
 \begin{align}
  2R_1 + R_2 & \leq C_{12} + C_{21} + \max\lcb m_1,n_1\rcb   + \max \lcb m_1, n_2\rcb
  - n_2 + \max \lcb m_2, n_1\rcb - n_1, \nonumber \\
  R_1 + 2R_2 & \leq C_{12} + C_{21} + \max\lcb m_2,n_2\rcb   + \max \lcb m_1, n_2\rcb
  - n_2 + \max \lcb m_2, n_1\rcb - n_1. \label{eq:extension1}
 \end{align}
\end{itemize}

The next outer bound, stated as Theorem \ref{th:theoremSLDIC-outer2}, focuses on the very high interference regime, i.e., for $\alpha \geq 2$. In the derivation of the bound, the encoded message $\xbold_i\:(i=1,2)$ at
each transmitter is partitioned into three parts, as shown in Fig.~\ref{fig:veryhighouter}. The partitioning is
based on whether (a) the bits are received at the intended receiver, and are received at the other receiver
without interference, (b) the bits are not received at the desired receiver, and received without interference
at the other receiver, and (c) the bits are not received at the intended receiver, and are received with interference at the other receiver. To motivate the development of the following outer bound, first consider the $C=0$ case. If receiver~$1$ can decode $\xbold_{1a}$ sent by transmitter~$1$, then receiver~$2$ can decode $\xbold_{1a}$ as well, since it gets these data bits without any interference. Hence, it is not possible to send any data bits securely on those
levels. Data transmitted at the remaining levels are not received by receiver~$1$, so they cannot be used
for secure data transmission either. Now, suppose a genie provides receiver~$1$ with the part of the signal
sent by transmitter~$1$ that is received without any interference at receiver~$2$, i.e.,
$\ybold_{2a}^N \triangleq (\xbold_{1a}^N,\xbold_{1b}^N)$. Then, by using the secrecy constraint for the
receiver~$2$, it is possible to bound the rate of user~$1$ by $I(W_1;\ybold_1^N|\ybold_{2a}^N)$. When
$\alpha \geq 2$, it is possible to show that $I(W_1;\ybold_1^N|\ybold_{2a}^N)=0$. When $C>0$, by using
the above mentioned approach and the relation in (\ref{eq:thouter0}), an outer bound on the symmetric
secrecy rate is derived for $\alpha \geq 2$, and is stated as the following theorem.
\begin{figure*}[t]
\centering
\mbox{\subfigure[]{\includegraphics[width=3in, height=2in]{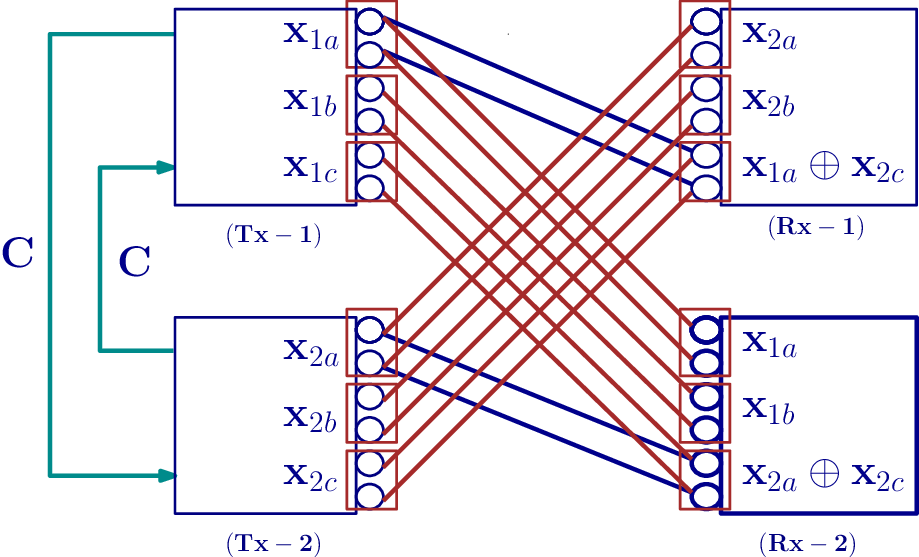} \label{fig:veryhighouter}}\quad \qquad \quad
\subfigure[]{\includegraphics[width=3in, height=2in]{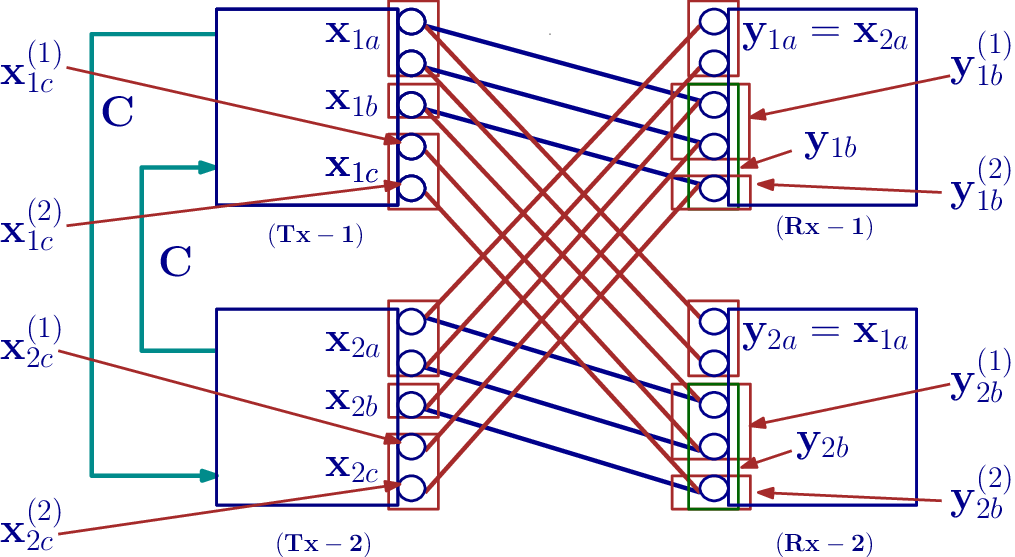} \label{fig:highouter}}}
\caption[]{\subref{fig:veryhighouter} SLDIC with $m=2$ and $n=6$ and \subref{fig:highouter} SLDIC with $m=3$ and $n=5$: Illustration of partitioning of the encoded message/output.}\label{fig:outersplit}
\end{figure*}

\begin{theorem}\label{th:theoremSLDIC-outer2}
In the very high interference regime, i.e., for $\alpha \geq 2$, the symmetric rate of the $2$-user SLDIC with
limited-rate transmitter cooperation and secrecy constraints at the receivers is upper bounded as: $R_s \leq  2C$.
\end{theorem}
\begin{proof}
The proof is provided in Appendix~\ref{sec:theoremSLDIC-outer2}.
\end{proof}
\textit{Remarks:}
\begin{itemize}
\item The outer bound in Theorem~$2$ can be extended to the asymmetric case
 under the following condition
 \begin{align}
   \min\lcb n_1, n_2\rcb > m_1 + m_2,  \label{eq:extension2}
 \end{align}
 and the outer bound  becomes
 \begin{align}
  R_1 \leq C_{12} + C_{21}, \qquad R_2 \leq  C_{12} + C_{21}. \label{eq:extension3}
 \end{align}
\item Theorem~\ref{th:theoremSLDIC-outer2} implies that, for $\alpha \geq 2$, it is not possible to achieve a rate
greater than $2C$, regardless of $m$ and $n$. In particular, when $C=0$, i.e., without cooperation, it is not possible
to achieve a nonzero rate. However, in the other
 interference regimes, it is possible to achieve rates greater than $2C$ (See Figs.~\ref{fig:gdof1}~and~\ref{fig:gdof2}).
\end{itemize}
The third outer bound, stated as Theorem~\ref{th:theoremSLDIC-outer3} below, is applicable in the high
interference regime, i.e., $1 < \alpha < 2$. The derivation of the outer bound involves partitioning of the output
and the encoded message based on whether the bits are received with interference at the intended receiver,
or causes interference to the other receiver, as shown in Fig. \ref{fig:highouter}. The outer bound on the
symmetric secrecy rate for the high interference regime is stated in the following theorem.
\begin{theorem}\label{th:theoremSLDIC-outer3}
In the high interference regime, i.e., for $1 < \alpha < 2$, the symmetric rate of the $2$-user SLDIC with
limited-rate transmitter cooperation and secrecy constraints at the receivers is upper bounded as: $R_s \leq 2C + 2m-n$.
\end{theorem}

The following theorem gives the outer bound on the symmetric secrecy rate for the $\alpha=1$ case. In this case,
both the receivers see the same signal. Hence, it is possible for receiver~$2$ decode any message that receiver~$1$
is able to decode, and vice-versa. Therefore, it is not possible to achieve a nonzero secrecy rate, irrespective of $C$.
A similar reasoning also holds for the Gaussian case, even though the receivers see independent noise instantiations.
\begin{theorem}\label{th:theoremSLDIC-outer4}
When $\alpha=1$, the symmetric rate of the $2$-user SLDIC with limited-rate transmitter cooperation and
secrecy constraints at the receivers is upper bounded as: $R_s = 0$.
\end{theorem}
\begin{proof}
	The proof is provided in Appendix~\ref{sec:theoremSLDIC-outer4}.
\end{proof}

A consolidated expression for the outer bound, obtained by taking minimum of the outer bounds in
Theorems~\ref{th:theoremSLDIC-outer1}-\ref{th:theoremSLDIC-outer4}, is stated as the following corollary. In
particular, the minimum of the outer bounds in Theorems~\ref{th:theoremSLDIC-outer1} and \ref{th:theoremSLDIC-outer3}
is taken for the high interference regime, and the minimum of the outer bounds in Theorems~\ref{th:theoremSLDIC-outer1}
and \ref{th:theoremSLDIC-outer2} is taken in the very high interference regime.
\begin{corollary}\label{cor:cons-outer-sldic}
An outer bound on the symmetric secrecy rate of the SLDIC, obtained by taking the minimum of the outer
bounds derived in this work, is given by:
\begin{align}
\frac{R_s}{m} \leq \left\{\begin{array}{l l}
    \frac{2\beta}{3}-\frac{2\alpha}{3}+1 &\mbox{for $\alpha < 1$ } \\ 	
    0 & \mbox{for $\alpha = 1$}\\
    \frac{2\beta}{3} + \frac{\alpha}{3} & \mbox{for $1 < \alpha < 2$, $\beta > \alpha -\frac{3}{2}$} \\
    & \mbox{or $\alpha \geq 2$, $\beta > \frac{\alpha}{4}$ }\\
    2\beta -\alpha + 2 & \mbox{for $\frac{3}{2} < \alpha < 2$, $0 \leq \beta < \alpha - \frac{3}{2}$}\\
    2\beta &\mbox{for $\alpha \geq 2$, $0 \leq \beta \leq \frac{\alpha}{4}$},
 \end{array}\right. \label{eq:cons-outer-sldic}
\end{align}
where $\beta \triangleq \frac{C}{m}$.
\end{corollary}

\textit{Remarks:}\begin{itemize}
\item Under cooperation asymmetry, all the outer bounds developed in the deterministic
 model still hold. This requires replacing $2C$ with $C_{12} + C_{21}$ in the expression
 for the outer bound. This is due to the fact that the entropy term $H(\coopsignalone, \coopsignaltwo)$ can be
upper bounded by $C_{12} + C_{21}$.
\item  There are cases where it is non-trivial to extend these bounds to the
asymmetric scenario (e.g.: Theorem~\ref{th:theoremSLDIC-outer3}). One of
the key techniques used
in the derivation of these outer bounds is the partitioning of the encoded
messages/outputs and careful selection of the side-information to be provided to
the receiver. This partitioning and side-information does not easily generalize to the asymmetric scenario.
\end{itemize}

Next, the achievable schemes for the SLDIC are presented.
\section{SLDIC: Achievable Schemes}\label{sec:sldic-ach}
\subsection{Weak interference regime $(0 \leq \alpha \leq \frac{2}{3})$}\label{sec:SLDIC-ach-weak}
In this regime, the proposed scheme uses interference cancelation. It is easy to see that data bits
transmitted on the lower $m-n$ levels $[1:m-n]$ remain secure, as these data bits do not cause interference
at the unintended receiver. Hence, it is possible to transmit $m-n$ bits securely, when $C=0$, as shown in Fig.~\ref{fig:sldic-modified-weak1}. However, with cooperation $(C > 0)$, it is possible to transmit on the top levels by
appropriately xoring the data bits with the cooperative bits in the lower levels prior to transmission. These
cooperative bits are precoded (xored) with the data bits at the levels $[1:\min\{n,C\}]$ to cancel interference
caused by the data bits sent by the other transmitter. When $C = n$, it can be shown that the proposed
scheme achieves the maximum possible rate of $\max\{m, n\}$~bits. When $C > n$, $C - n$ bits can be
discarded and $n$ cooperative bits can be used for encoding as above, to achieve $\max\{m,n\}$~bits. Hence,
in the sequel, it will not be explicitly mentioned that $C \leq n$. The proposed encoding scheme achieves the
following symmetric secrecy rate:
\begin{align}
R_{s} = m-n + C.  \label{weakach2}
\end{align}
A high level description of the achievable scheme is illustrated in Fig.~\ref{fig:sldic-modified-weak1}. The details of the encoding
scheme and the derivation of (\ref{weakach2}) can be found in \cite{partha-spawc-2013}.
\begin{figure*}[t]
\centering
\mbox{\subfigure[][]{\includegraphics[width=2.8in, height=2.4in]{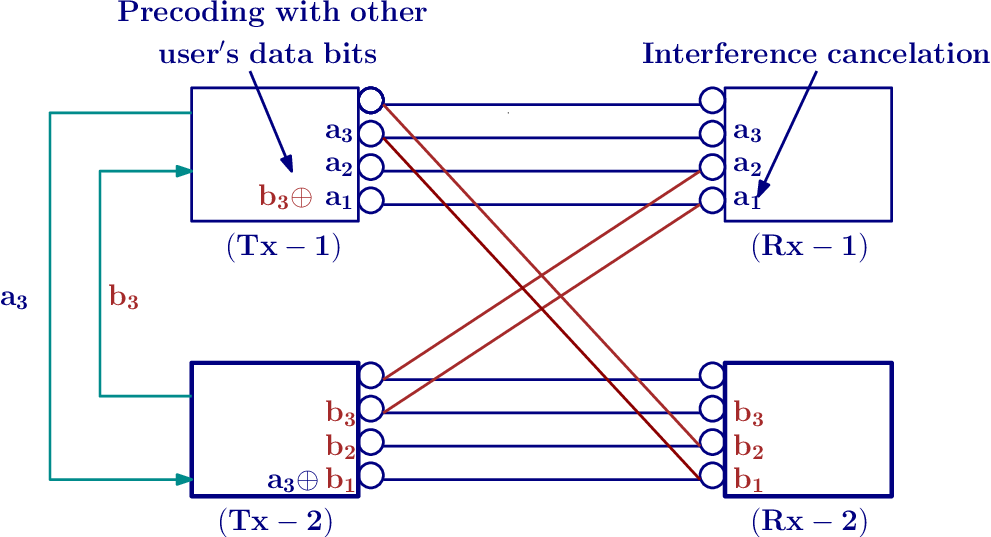}\label{fig:sldic-modified-weak1}} \quad
\subfigure[][]{\includegraphics[width=2.8in, height=2.4in]{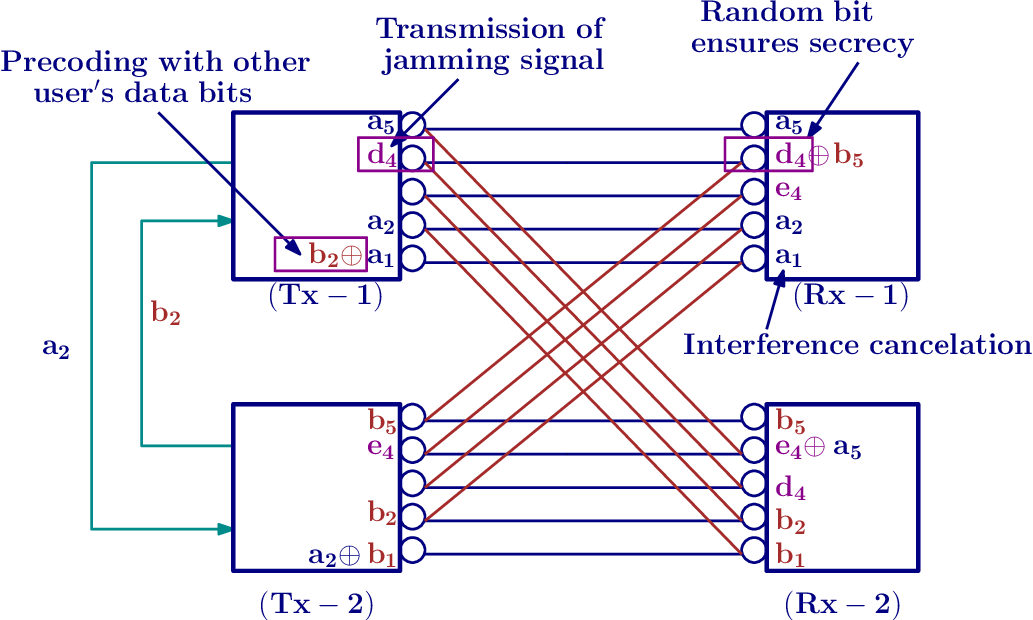}\label{fig:sldic-modified-weakmod1}}}
\caption[]{SLDIC: \subref{fig:sldic-modified-weak1} $m=4$, $n=2$, $C=1$ and $R_s = 3$, \subref{fig:sldic-modified-weakmod1}
 $m=5$ and $n=4$, $C=1$ and $R_s=3$.}\label{fig:sldic-weakmod4}
\end{figure*}

\textit{Remarks:}
\begin{enumerate}
\item In this regime, the proposed achievable scheme  meets the symmetric capacity of the SLDIC
\emph{without} secrecy constraints \cite{wang-TIT-2011} for all  values of $C$ (See Figs. \ref{fig:gdof1} and \ref{fig:gdof2}).
Thus, the secrecy constraints at the receivers do not reduce the symmetric capacity region of the SLDIC.
\item In this regime, the proposed scheme does not involve transmission of a jamming signal (or random
 bits), even when $C=0$. In the next subsection, it will be seen that the
 transmission of the jamming signal improves the achievable secrecy rate, when
 the capacity of the cooperative link is not sufficient to cancel interference
 at the unintended receiver.
 \end{enumerate}
\subsection{Moderate interference regime $(\frac{2}{3} < \alpha < 1)$}\label{sec:SLDIC-ach-weak-mod}
In this regime, the proposed scheme uses interference cancelation along with the transmission of random
bits. Without transmitter cooperation, it is possible to transmit at least $m-n$ bits securely, as in the weak
interference regime. Depending on the value of $C$ and $\alpha$, with the help of transmission of random bits,
it is possible to send additional data bits on the higher levels $[m-n+1:m]$ by carefully placing data bits along
with zero bits and random bits.

The proposed scheme achieves the following symmetric secrecy rate:
\begin{align}
R_s  = m-n + B(m-n) + q + C, \label{eq:sldic-modach1}
\end{align}
where $B \triangleq \left\lfloor\frac{g}{3r_2} \right\rfloor$, $g \triangleq \lcb n-(r_2 + C)\rcb^{+}$, $r_2 \triangleq m-n$,
$q \triangleq \min \lcb (t-r_2)^{+}, r_2\rcb$ and  $t\triangleq g\% \{3r_2\}$.

In the above equation, the first term corresponds to the number of data bits transmitted securely without using
random bits transmission or cooperation. The term $B(m-n)+q$ corresponds to the number of data bits that can
be securely transmitted using the help of random bits transmission. The last term $C$ represents the gain in
rate achievable due to cooperation.

A high level description of the achievable scheme is illustrated in Fig.~\ref{fig:sldic-modified-weakmod1}. The details of the
encoding scheme and the derivation of (\ref{eq:sldic-modach1}) can be found in \cite{partha-spawc-2013}.

\textit{Remark:}
In this regime, it is possible to transmit data bits securely in the higher levels $[m-n+1:m]$ by
intelligently choosing the placement of data and random bits, in addition to interference cancelation.

\subsection{Interference is as strong as the signal $(\alpha=1)$} In this case, from Theorem~\ref{th:theoremSLDIC-outer4},
it is not possible to achieve a nonzero secrecy rate.
\subsection{High interference regime $(1 < \alpha < 2)$}\label{sec:SLDIC-ach-highint}
The achievable scheme is similar to that proposed for the moderate interference regime, but it differs in the
manner the encoding of the message is performed at each transmitter. The proposed scheme achieves the
following secrecy rate:
\begin{enumerate}
\item When $(1 <\alpha \leq 1.5)$:
\begin{align}
R_s  = B(n-m)+  q + C, \label{eq:sldic-highach1}
\end{align}
where $B \triangleq \left\lfloor\frac{g}{3r_2}\right\rfloor$, $g \triangleq (m-C)^+$, $q \triangleq \min \lcb (t-r_2)^{+}, r_2\rcb$, $t\triangleq g\% \{3r_2\}$ and $r_2 \triangleq n-m$.
\item When $(1.5 <\alpha < 2)$:
\begin{align}
R_s = \left\{\begin{array}{l l}
   2m-n + C  &\mbox{for $0 \leq C \leq 4n-6m$} \\ 	
   4n-6m + C_{T_1}  
   + C_{T_2} + C_{T_3} + r_d & \mbox{for $4n-6m < C \leq n$},
 \end{array}\right. \label{eq:sldic-highach2}
\end{align}
where $C_{T_1} \triangleq  \min\lcb \mylceil \frac{C_{\text{rem}}}{2}\myrceil,2m-n\rcb$, $C_{\text{rem}} \triangleq (C' - C_{T_3})^+$,
$C_{T_3}\triangleq \min\lcb 2m-n,C''\rcb$, $C' \triangleq C - (4n-6m)$, $C'' \triangleq \mylceil \frac{C'}{3}\myrceil$,
$C_{T_2} \triangleq \min\lcb 2m-n, (C_{\text{rem}} - C_{T_1})^{+}\rcb$ and
$r_d \triangleq \min\lcb 2m-n-C_{T_3}, 2m-n-C_{T_2}\rcb$.
\end{enumerate}
The details of the encoding scheme and some illustrative examples can be found in Appendix~\ref{sec:appen-ach-high-sldic}.

\textit{Remarks:}
\begin{enumerate}
\item When $C=0$ and $1.5 < \alpha < 2$, the proposed scheme is capacity
achieving. The outer bound in Theorem~\ref{th:theoremSLDIC-outer3} helps to
establish this.
\item One can note that the achievable schemes for the moderate (Sec.~\ref{sec:SLDIC-ach-weak-mod})
and high interference regime (Sec.~\ref{sec:SLDIC-ach-highint}) use a combination of interference cancelation
and transmission of a jamming signal (random bits transmission). When precoding is done using the other
user's signal, it cancels the interference and also ensures secrecy. In the technique based on random bits
transmission, the transmitter self-jams its own receiver, so that the receiver cannot decode the other user's
data. But, in this process, transmitter causes interference to the other receiver, thereby adversely impairing
the achievable rate of secure communication. Thus, self jamming in that form only helps if the benefit to the
secrecy rate due to the interference caused at the own receiver outweighs the negative impact of the interference
caused at the other receiver. However, when the jamming signal can be canceled at an unintended receiver by
transmission of the same random bits by the other transmitter, its adverse impact is completely alleviated,
leading to larger achievable rates.
\end{enumerate}
\subsection{Very high interference regime $(\alpha \geq 2)$} In this case, when $C=0$, it is not possible to
achieve nonzero secrecy rate as established by the outer bound in Theorem \ref{th:theoremSLDIC-outer2}.
However, with cooperation $(C > 0)$, it is possible to achieve nonzero secrecy rate. The proposed scheme
uses interference cancelation, time sharing, and relaying the other user's data bits. In contrast to the
achievable schemes for other interference regimes,  the transmitters exchange data bits, random bits, or
both, depending on the capacity of the cooperative link. The proposed scheme achieves the following secrecy rate:
\begin{figure*}
	\centering
	\mbox{\subfigure[Random bits sharing: $R_s = 2$.]{\includegraphics[width=2.7in,height=2.45in]{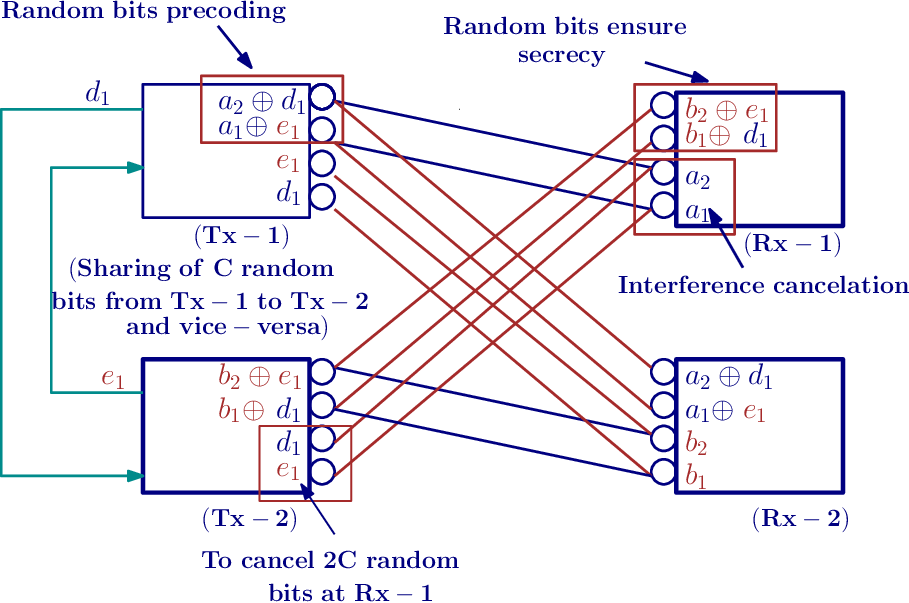}\label{fig:sldic-veryhigh1}} \qquad \qquad \qquad 
		\quad \subfigure[Data bits sharing: $R_s = 1$.]{\includegraphics[width=2.7in,height=2.42in]{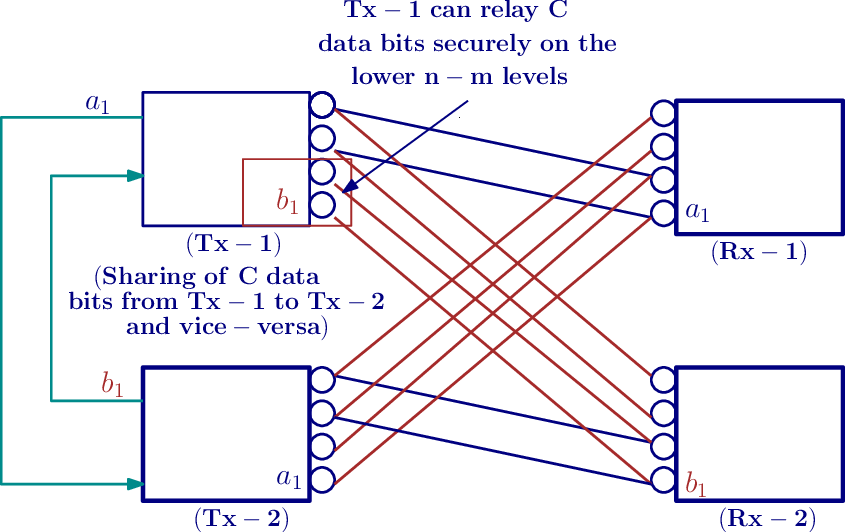}\label{fig:sldic-veryhigh2}}}
	\pullUp
	\caption[]{SLDIC with $m=2$, $n=4$ and $C=1$.}\label{fig:sldic-veryhigh3}
\end{figure*}
\begin{enumerate}
\item When $m$ is even:
\begin{align}
R_s = \left\{\begin{array}{l l}
   2C &\mbox{for $0 < C \leq \frac{m}{2}$} \\ 	
   \frac{m}{2} + C & \mbox{for $\frac{m}{2} < C \leq n- \frac{3m}{2}$} \\
   \frac{n}{2} - \frac{m}{4} + \frac{C}{2} & \mbox{for $n-\frac{3m}{2} < C < n-\frac{m}{2}$} \\
   C & \mbox{for $n-\frac{m}{2} \leq C \leq n$.}
 \end{array}\right. \label{eq:sldic-highach3a}
\end{align}
\item When $m$ is odd:
\begin{align}
& R_s = \left\{\begin{array}{l l}
    \min\{2C,m\} \qquad \mbox{for $0 < C \leq \frac{m+1}{2}$} \\ 	
     m + \min\lcb C - \frac{m+1}{2}, n-2m\rcb  \mbox{for $ \frac{m+1}{2}< C \leq \frac{2n-3m+1}{2}$} \\
    n-2m + \frac{1}{2}\lsqb \coneul + 2\coneuu\rsqb  + \frac{1}{2}\lsqb  \ctwouu + \conelu + \ctwoul\rsqb  \quad \mbox{for $\frac{2n-3m+1}{2} < C \leq n$},
 \end{array}\right. \label{eq:sldic-highach3}
\end{align}
\end{enumerate}
where $\coneuu \triangleq \left\lceil{\frac{C}{2}}\right\rceil$, $\coneul \triangleq (m-\coneuu)^+$,
$\ctwouu \triangleq (C - \ctwolu - C_2^r)^+$, $\conelu \triangleq (C - \coneuu - C_1^r)^{+}$,
$\ctwoul \triangleq \conell$, $\conell \triangleq \min\{2C_1^r, (m-\conelu)^{+})\}$,
$\ctwoll \triangleq \coneul$ and $C_2^r \triangleq \max\lcb \left\lceil \frac{\ctwoll}{2}\right\rceil, \left\lfloor \frac{\ctwoul}{2}\right\rfloor\rcb$,
$C_1^r \triangleq \left\lceil \frac{\coneul}{2} \right\rceil$.

The details of the achievable scheme can be found in Appendix~\ref{sec:appen-ach-veryhigh-sldic}.

\textit{Remarks:}
\begin{enumerate}
\item When $0 < C \leq \lceil\frac{m}{2}\rceil$, the capacity achieving scheme involves exchanging only
\emph{random bits} through the cooperative links. This is useful in scenarios where the transmitters trust
each other to follow the agreed-upon scheme, but are not allowed to share their data bits through the
cooperative link. The outer
bound in Theorem~$2$ establishes the optimality of the proposed scheme. The achievable scheme is illustrated for random bits sharing and data bits sharing for C = 1
in Figs.~\ref{fig:sldic-veryhigh1}~and~\ref{fig:sldic-veryhigh2}, respectively.
\item When $\frac{m}{2} < C  < n-\frac{m}{2}$ $ \lb \text{or } \frac{m+1}{2} < C \leq n \rb$ and $m$ is even
(or odd) valued,
the proposed scheme shares a combination of random bits and data bits through the cooperative links.
In Fig.~\ref{fig:sldic-veryhigh9}, a schematic representation of the achievable scheme for $m=2$ and $n=4$, with $C =2$ bits is
shown for the first time slot. In the second
time slot, the encoding for transmitters~$1$ and $2$ is reversed. In the second time slot, users~$1$ and $2$ achieve
a rate of $R_1=3$ and $R_2 = 2$, respectively. Hence, a symmetric rate of $R_s = 2.5$ is achievable.
\end{enumerate}
\begin{figure*}[t]
\centering
\includegraphics[width=4in, height=2.8in]{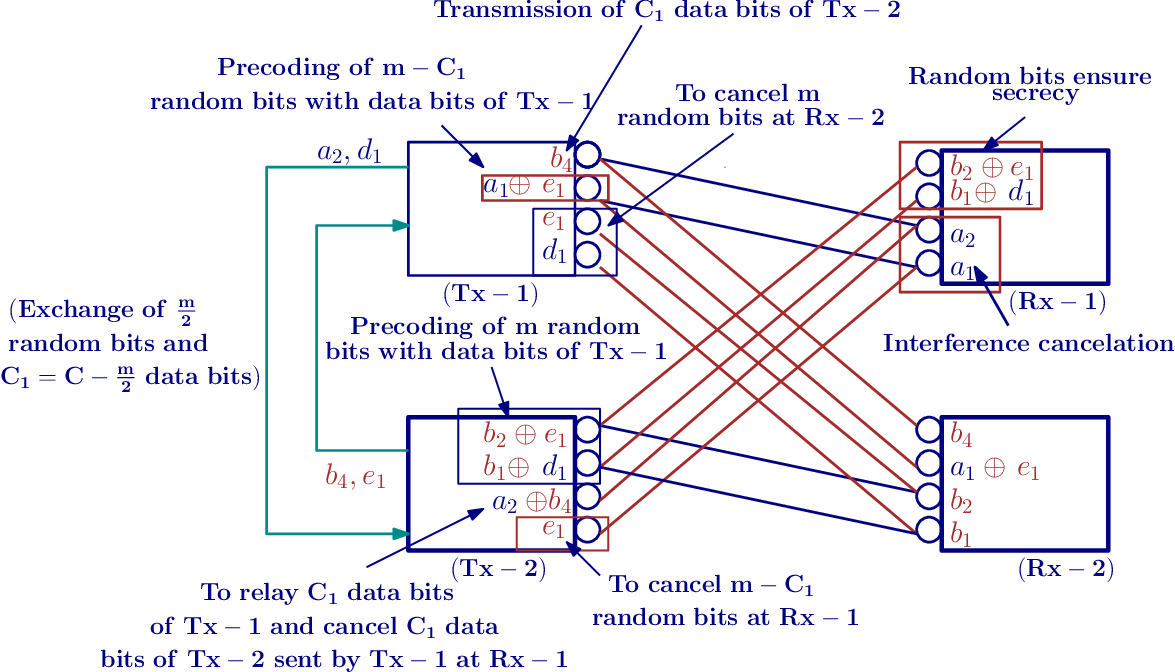}
\caption{SLDIC with $m=2$, $n=4$ and $C=2$: $R_1 = 2$ and $R_2 = 3$ is achievable in the first time slot. In the
second time slot, the role of transmitters~$1$ and $2$ is reversed and users~$1$ and $2$ achieve
a rate of $R_1=3$ and $R_2 = 2$, respectively. }\label{fig:sldic-veryhigh9}
\pullUp
\end{figure*}

Interestingly, it turns out that the symmetric capacity region of the SLDIC does
not change if the perfect secrecy constraint at the receiver is replaced with
the strong or the weak notion of secrecy, when the proposed scheme is capacity
achieving. This result is stated as the following theorem.
\begin{theorem}\label{th:newtheorem}
The symmetric secrecy capacity region of the deterministic SLDIC with transmitter cooperation satisfies the
following relationship, when the proposed scheme is capacity achieving:
\begin{align}
\mathcal{C}^{\text{perfect}} = \mathcal{C}^{\text{strong}} = \mathcal{C}^{\text{weak}}, \label{eq:coor-capacity}
\end{align}
where $\mathcal{C}^{\text{perfect}}$, $\mathcal{C}^{\text{strong}}$ and $\mathcal{C}^{\text{weak}}$
correspond to the capacity region with the perfect, strong and weak notions of secrecy, respectively.
\end{theorem}
\begin{proof}
Any communication scheme satisfying the perfect secrecy condition  will automatically satisfy the strong and
weak secrecy condition. Similarly, a communication scheme satisfying strong secrecy will automatically satisfy
the weak secrecy condition. Hence, the following holds
\begin{align}
\mathcal{C}^{\text{perfect}} \subseteq \mathcal{C}^{\text{strong}} \subseteq
\mathcal{C}^{\text{weak}}.
\label{eq:coor-capacity1}
\end{align}
The achievable results in Sec.~\ref{sec:sldic-ach} are obtained under perfect secrecy constraints at the
receivers. It is not difficult to show that the outer bounds on the secrecy rate in Theorems~\ref{th:theoremSLDIC-outer1}-\ref{th:theoremSLDIC-outer4}
do not change if the perfect secrecy is replaced with the weak notion of secrecy. When the achievable rates
meet the corresponding outer bounds, the relation in (\ref{eq:coor-capacity}) holds.
\end{proof}

Finally, this section is concluded with the following remarks:
\begin{enumerate}
 \item When $C=n$, i.e., the cooperative link is as strong as the
 interference, and  when $\alpha \neq 1$, the proposed scheme achieves the maximum possible rate of~$\max\{m,n\}$.
 \item In \cite{geng-globecom-2016},
it is shown that the proposed outer bound  in Theorem~$1$ in Sec.~\ref{sec:LDIC-outer} is tight for $\frac{3}{4}< \alpha < 1$
and $1 < \alpha < 1.5$, when $C=0$. Hence, the secrecy capacity is characterized for these regimes of $\alpha$ also. However,
the symmetric secrecy capacity region of the $2$-user SLDIC,  when $\frac{1}{2} < \alpha \leq \frac{3}{4}$ and
$C=0$, remains an open problem.
\item It is possible to extend the achievable scheme based on interference cancelation (involving exchange
  of data bits between the transmitters) as well as the scheme based on transmission of
  random bits to the asymmetric case for the deterministic model. However, it is not straightforward  to extend the
  achievable schemes which rely
  on the exchange of both data and random bits  to the asymmetric case.  The extension requires a careful 
  re-working of a scheme for sharing
  random bits and data bits in the asymmetric setting.
\end{enumerate}
In the following section, outer bounds for the GSIC are presented.
\section{GSIC: Outer Bounds}\label{sec:outerGaussian}
In this section, the outer bounds on the secrecy rate for the GSIC with limited-rate transmitter cooperation are
stated as Theorems~\ref{th:theorem_GSIC_outer1}-\ref{th:theorem_GSIC_outer3}. The extension of the outer
bounds from the deterministic model to the Gaussian model is non-trivial, because of the following well known
differences between the models:
 \begin{enumerate}
 \item In the deterministic model, interference or superposition of signals is modeled using the XOR operation.
 Hence, the levels do not interact with each other.
 \item In the deterministic model, noise is modeled using truncation.
 \item In the Gaussian model, due to finite rate cooperation between the transmitters, the differential entropy
 terms contain discrete as well as continuous random variables. This makes the derivation of the outer
 bounds more difficult in the Gaussian case.
 \end{enumerate}

Due to the above differences, the partitioning used in the derivation of the outer bounds for the deterministic
case is not directly applicable to the Gaussian case. To overcome this problem, either analogous quantities
that serve as side-information need to be found to mimic the partitioning of the encoded messages/outputs,
or the bounding steps need to be modified taking cue from the deterministic model. This is discussed in detail
in this section.

The outer bound derived in Theorem \ref{th:theoremSLDIC-outer1} partitions the encoded message into two
parts: $\xbold_{ia}^N$ (received at receiver~$j$, $j \neq i$) and $\xbold_{ib}^N$ (not received at receiver~$j$, $j \neq i$).
However, it is not possible to partition the message in this way for the Gaussian case. Hence, in the derivation
of Theorem~\ref{th:theorem_GSIC_outer1}, $\sbold_i^N = h_c \xbold_i^N + \zbold_j^N$ $(j \neq i)$ is used as
a proxy for $\xbold_{ia}^N$. In this section, the following notation is used: $\text{SNR} \triangleq h_d^2 P$,
$\text{INR} \triangleq h_c^2 P$ and $\rho \triangleq E[\xbold_1 \xbold_2]$.
\begin{theorem}\label{th:theorem_GSIC_outer1}
The symmetric rate of the $2$-user GSIC with limited-rate transmitter cooperation and secrecy constraints at the receiver is upper bounded as follows:
\begin{align}
& R_s \leq  \displaystyle\max_{0 \leq |\rho| \leq 1} \frac{1}{3}\lsqb 2C_G + 0.5\log \text{det}\lb \Sigma_{\ybar|\sbar}\rb + 0.5\log\lb 1 + \SNRsum + 2\rho\SNRprod\rb\rsqb, \label{eq:theorem_GSIC_outer1}
\end{align}
where $\Sigma_{\ybar|\sbar} = \Sigma_{\ybar} - \Sigma_{\ybar,\sbar}\Sigma_{\sbar}^{-1}\Sigma_{\ybar,\sbar}^{T}$,
\begin{align}
& \Sigma_{\ybar}\!\! =\! \!\lsqb\begin{array}{ll}
1 + \SNRsum + 2\rho\SNRprod & 2\SNRprod \nonumber \\ & + \rho(\SNRsum) \\
2\SNRprod + \rho(\SNRsum) & 1 + \SNRsum \\ & + 2\rho\SNRprod
           \end{array}\rsqb, 
& \Sigma_{\sbar} =  \lsqb\begin{array}{ll}
	1 + \text{INR} & \rho\text{INR} \\
	\rho\text{INR}   & 1 + \text{INR}
\end{array}\rsqb, \nonumber \\
& \Sigma_{\ybar, \sbar} = \lsqb\begin{array}{ll}
\SNRprod + \rho \text{INR} & \text{INR} + \rho\SNRprod \\
\text{INR} + \rho\SNRprod  & \SNRprod + \rho \text{INR}
           \end{array}\rsqb, \text{ and } \nonumber 
\end{align} and $\det(\cdot)$ represents the determinant of a matrix.
\end{theorem}
\begin{proof}
The proof is provided in Appendix \ref{sec:appendouter1}.
\end{proof}
\textit{Remarks:}
\begin{enumerate}
   \item The outer bound in Theorem~\ref{th:theorem_GSIC_outer1} for the Gaussian model can be extended
   to obtain an outer bound on $2R_1 +  R_2$ under the asymmetric setting. The
   outer bound becomes\footnote{With a slight abuse of notation, $C_{ij}$ has been used to represent the capacity
   of the cooperative link from transmitter~$i$ to transmitter~$j$ for both the deterministic and the Gaussian models in
   the asymmetric case.}
 \begin{align}
& 2R_1 + R_2  \leq \displaystyle\max_{0 \leq |\rho| \leq 1}  C_{12} + C_{21}  + 0.5\log\lb 1 + \SNRt_1 + \INRt_1
+ 2\rho\sqrt{\SNRt_1\: \INRt_1}\rb \nonumber \\
& \qquad \qquad \qquad + 0.5\log \text{det}\lb \Sigma_{\ybar|\sbar}\rb, \label{eq:theorem_GSIC_outer1_2}
\end{align}
where $\Sigma_{\ybar|\sbar} = \Sigma_{\ybar} - \Sigma_{\ybar,\sbar}\Sigma_{\sbar}^{-1}\Sigma_{\ybar,\sbar}^{T}$,
\begin{align}
& \Sigma_{\ybar} = \lsqb\begin{array}{ll}
\Sigma_{\ybar, 11} & \Sigma_{\ybar, 12}\\
\Sigma_{\ybar, 21}  &  \Sigma_{\ybar, 22}\end{array}\rsqb,  \text{ and }  \Sigma_{\sbar} =  \lsqb\begin{array}{ll}
1 + \INRt_2 & \rho\sqrt{\INRt_1 \: \INRt_2} \\
\rho\sqrt{\INRt_1\: \INRt_2}  & 1 + \INRt_1
\end{array}\rsqb, \nonumber \\
& \Sigma_{\ybar, 11} \triangleq 1 + \SNRt_1 + \INRt_1 + 2\rho\sqrt{\SNRt_1\: \INRt_1}, \nonumber \\
& \Sigma_{\ybar, 12} \triangleq \sqrt{\SNRt_1 \: \INRt_2} + \sqrt{\SNRt_2\:\INRt_1}   +
\rho \lb \sqrt{\SNRt_1\: \SNRt_2} + \sqrt{\INRt_1 \: \INRt_2}\rb,  \nonumber \\
& \Sigma_{\ybar, 21} \triangleq \sqrt{\SNRt_1 \:\INRt_2} + \sqrt{\SNRt_2\:\INRt_1}  + \rho \lb \sqrt{\SNRt_1\: \SNRt_2} + \sqrt{\INRt_1\: \INRt_2}\rb \nonumber \\
& \text{ and } \Sigma_{\ybar, 22} \triangleq 1 + \SNRt_2 + \INRt_2 + 2\rho\sqrt{\SNRt_2\: \INRt_2}, \nonumber \\
& \Sigma_{\ybar, \sbar}\!\!=\!\!\lsqb\begin{array}{ll}
\sqrt{\SNRt_1 \: \INRt_2} + \rho \sqrt{\INRt_1 \: \INRt_2} & \rho \sqrt{\SNRt_1 \: \INRt_1}   + \INRt_1 \\ \\
\rho\sqrt{\SNRt_2 \: \INRt_2} + \INRt_2  & \sqrt{\SNRt_2 \: \INRt_1}  + \rho \sqrt{\INRt_1 \: \INRt_2}\end{array}\rsqb, \nonumber \\
\end{align}
where $\SNRt_1\triangleq h_{11}^2 P_1$, $\SNRt_2 \triangleq h_{22}^2 P_2$, $\INRt_{1} \triangleq h_{12}^2
P_2$, and $\INRt_{2} \triangleq h_{21}^2 P_1$.
\item Using a similar approach as used in the proof of
Theorem~\ref{th:theorem_GSIC_outer1}, an outer bound on $R_1 + 2 R_2$ can be
obtained.
\end{enumerate}

The outer bound on the secrecy rate presented in the following theorem is based on the idea used in deriving
outer bounds in Theorems~\ref{th:theoremSLDIC-outer2} and \ref{th:theoremSLDIC-outer3} for case of the
SLDIC. But, in the Gaussian setting, it is not possible to partition the encoded message as was done for
the SLDIC. For example, in Theorem~\ref{th:theoremSLDIC-outer2}, a part of the output at receiver~$2$
which does not contain the signal sent by transmitter~$1$ is provided as side information to receiver~$1$.
Hence, the approach used in the derivation of the outer bound in case of SLDIC cannot be directly used for
the Gaussian case. To overcome this problem, for the Gaussian case, first $\xbold_2^N$ is provided as side
information to receiver~$1$; this eliminates the interference caused by transmitter~$2$. Then, the receiver~$1$
is provided with $\ybold_2^N$ as side-information. The outer bound on the symmetric secrecy rate is stated
in the following theorem.
\begin{theorem}\label{th:theorem_GSIC_outer2}
The symmetric rate of the $2$-user GSIC with limited-rate transmitter cooperation and secrecy constraints at 
the receiver is upper bounded as follows:
\begin{align}
& R_s \!\leq\!  \displaystyle\max_{0 \leq |\rho| \leq 1} \!\lsqb  0.5\log\lb 1 + \frac{\text{SNR} + \text{SNR}^2(1-\rho^2)}{1 + \text{SNR} + \text{INR} + 2\rho\sqrt{\text{SNR}\:\text{INR}}} \rb  + 2C_G\rsqb. \label{eq:theorem_GSIC_outer2}
\end{align}
\end{theorem}
\begin{proof}
The proof is provided in Appendix \ref{sec:appendouter2}.
\end{proof}
\textit{Remark:}
The outer bound in Theorem~\ref{th:theorem_GSIC_outer2}  can be extended to the asymmetric
setting, and the outer bound becomes
\begin{align}
& R_1 \leq \nonumber \\
&\displaystyle\max_{0 \leq |\rho| \leq 1} \lsqb  0.5\log\lb 1
+ \frac{ \SNRt_1 + \SNRt_1\: \SNRt_2(1-\rho^2)}{1 + \SNRt_2 + \INRt_2 + 2\rho\sqrt{\SNRt_2 \: \INRt_2}} \rb + C_{12} + C_{21} \rsqb. \label{eq:theorem_GSIC_outer2_2}
\end{align}
Using a similar approach as used in the proof of
Theorem~\ref{th:theorem_GSIC_outer2}, an outer bound on $R_2$ can be
obtained.

The outer bound presented in the following theorem is similar to the outer bound presented in
Theorem~\ref{th:theoremSLDIC-outer4} in case of the SLDIC. This kind of outer bound exists in the
literature (see, for example, \cite{shamai-TIT-2012}), but for the sake of completeness, it is presented in the
following theorem. Unlike the results in Theorems~\ref{th:theorem_GSIC_outer1} and \ref{th:theorem_GSIC_outer2},
this outer bound does not depend on the capacity of the cooperative link.
\begin{theorem}\label{th:theorem_GSIC_outer3}
The symmetric rate of the $2$-user GSIC with limited-rate transmitter cooperation and secrecy constraints at the receiver is upper bounded as follows:
\begin{align}
& R_s \leq \displaystyle\max_{0 \leq |\rho| \leq 1} 0.5\log\lsqb 1 + \SNRsum + 2\rho\SNRprod -  \frac{(2\SNRprod + \rho(\SNRsum))^2}{1 + \SNRsum + 2\rho\SNRprod} \rsqb. \label{eq:theorem_GSIC_outer3}
\end{align}
\end{theorem}
\begin{proof}
The proof is provided in Appendix~\ref{sec:appendouter3}.
\end{proof}
\textit{Remark:} The outer bound in Theorem~\ref{th:theorem_GSIC_outer3}  can be extended to the
asymmetric setting, and the outer bound becomes
\begin{align}
& R_1 \leq \displaystyle\max_{0 \leq |\rho| \leq 1} 0.5\log \lb 1 + \SNRt_1 + \INRt_1 + 2\rho\sqrt{\SNRt_1 \: \INRt_1} \right. \nonumber \\
&\left. - \frac{\lb \sqrt{\SNRt_1\: \INRt_2} \!+\! \sqrt{\SNRt_2\: \INRt_1} \!+\! \rho\lb \SNRt_{12} \!+\!  \INRt_{12}\rb\rb^2}{1 + \SNRt_2 + \INRt_2 + 2\rho\sqrt{\SNRt_2 \:
\INRt_2}}\rb, 
\end{align}
where $\SNRt_{12} \triangleq \sqrt{\SNRt_1\:\SNRt_2}$ and $\INRt_{12}  \triangleq  \sqrt{\INRt_1\: \INRt_2}$.

Using a similar approach as used in the proof of
Theorem~\ref{th:theorem_GSIC_outer3}, an outer bound on $R_2$ can be
obtained.
\subsection{Relation between the outer bounds for SLDIC and GSIC}\label{sec:OB-det-Gaussian}
In the following, it is shown that at high SNR and INR, the outer bounds developed for the Gaussian case
(Theorems~\ref{th:theorem_GSIC_outer1}~and~\ref{th:theorem_GSIC_outer2}) are approximately equal to
the outer bounds for the SLDIC, when $C=0$.\footnote{When $C \neq 0$, from Fig. \ref{fig:compare_outer_GSIC},
it appears that the approximate equivalence of the bounds for the GSIC and SLDIC will still hold.} In
Fig.~\ref{fig:compare_outer_GSIC}, the outer bounds on the achievable secrecy rate in Theorems~\ref{th:theorem_GSIC_outer1}-\ref{th:theorem_GSIC_outer3}
are compared as a function of $\alpha$, for $C_G=0$ and $C_G=1$, when $P=20$~dB and $h_d=1$. This validates that the approaches used in obtaining outer bounds in the two models are consistent with each other.
\begin{figure}
	\begin{center}
		\includegraphics[width=4.2in, height=2.5in]{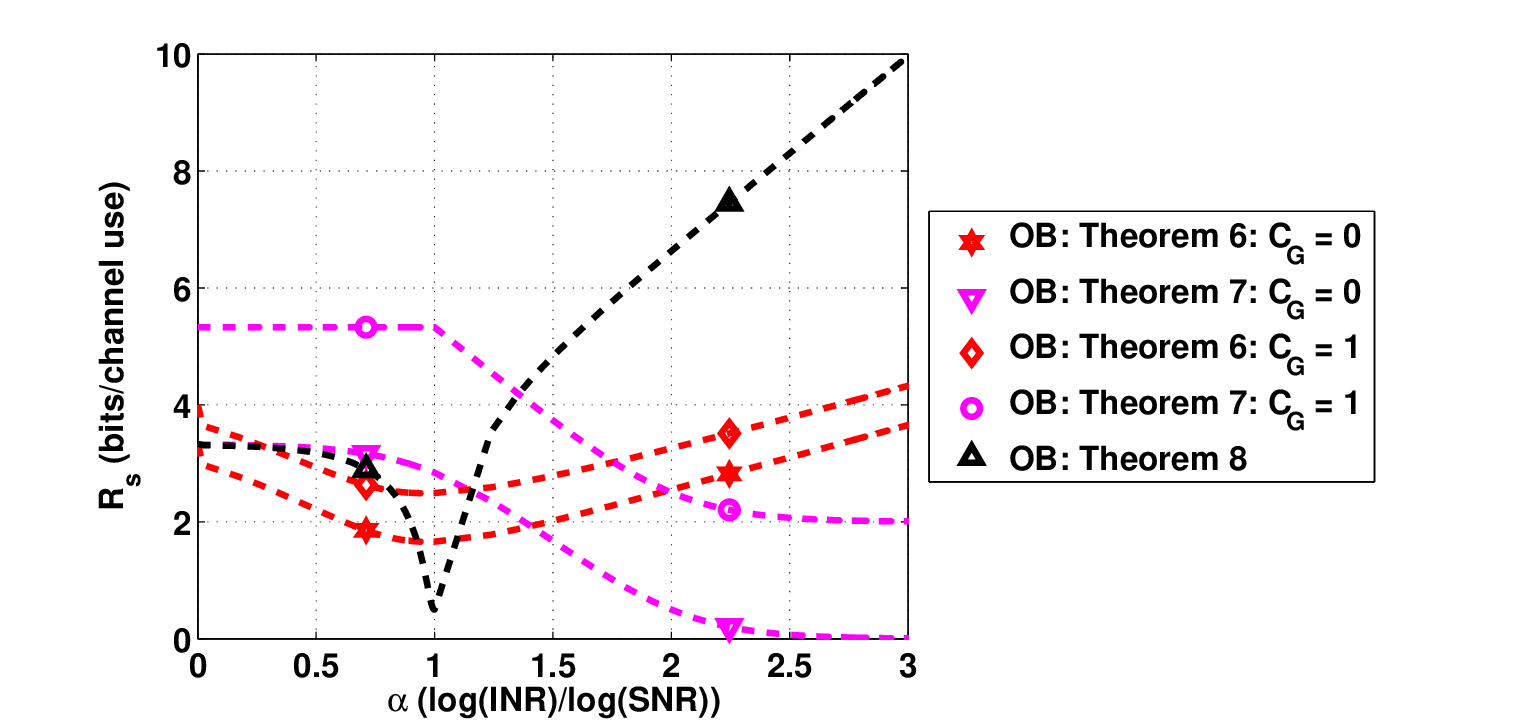}
		\vspace{-0.25cm}
		\caption{Comparison of different outer bounds on the secrecy rate for the GSIC with $P=20$~dB and $h_d=1$. In the legend, OB stands for the outer bound.} \label{fig:compare_outer_GSIC}
	\end{center}
\end{figure}

In the following, for ease of presentation, it is assumed that $0.5 \log \text{SNR}$ and $0.5 \log \text{INR}$
are integers. Recall that, the parameters $m$ and $n$ of the SLDIC are related to the GSIC as
$m = (\lfloor 0.5 \log \text{SNR}\rfloor)^{+}$ and $n =  (\lfloor 0.5 \log \text{INR}\rfloor)^{+}$, respectively.
\subsubsection{Outer bound in Theorem~\ref{th:theorem_GSIC_outer1}} Consider the following bound in the proof of Theorem~\ref{th:theorem_GSIC_outer1}, when $C_G=0$:
\begin{align}
N[R_1 + 2 R_2] & \leq h(\ybold_1^N) + h(\ybold_1^N,\ybold_2^N|\sboldtilde_1^N, \sboldtilde_2^N) - h(\widetilde{\zbold}_1^N) -h(\widetilde{\zbold}_2^N)  - h(\zbold_1^N) + N\epsilon'', \nonumber \\
& \leq h(\ybold_1^N) + h(\ybold_1^N|\sboldtilde_1^N) + h(\ybold_2^N| \sboldtilde_2^N)- h(\widetilde{\zbold}_1^N) -h(\widetilde{\zbold}_2^N)  - h(\zbold_1^N) + N\epsilon'', \nonumber \\
\text{or } R_1 + 2 R_2 
&  \leq 0.5 \lsqb \log(1 + \text{SNR} + \text{INR}) + 2\log\lb1 + \frac{\text{SNR} + \text{INR}}{1 + \text{INR}}\rb\rsqb, \nonumber \\
& \approx 0.5 [\log( \text{SNR} + \text{INR}) + 2\log(\text{SNR} + \text{INR}) - 2\log \text{INR} ], \label{eq:discussion2}
\end{align}
where the last equation is obtained for high SNR and INR. Using the above mentioned definitions of $m$ and $n$, (\ref{eq:discussion2}) reduces to:
\begin{align}
R_s \leq \left\{\begin{array}{l l}
    \frac{1}{3}\lsqb 3m-2n\rsqb &\mbox{for $\alpha \leq 1$ } \\ 	
    \frac{n}{3} & \mbox{for $\alpha > 1$}.
 \end{array}\right. \label{eq:discussion2a}
\end{align}
The above is the same as the outer bound for the SLDIC in Theorem~\ref{th:theoremSLDIC-outer1}, when
$C=0$.
\subsubsection{Outer bound in Theorem~\ref{th:theorem_GSIC_outer2}} When $C_G=0$, the outer bound in Theorem~\ref{th:theorem_GSIC_outer2} reduces to the following, in the high SNR and high INR regime:
\begin{align}
R_s &\leq \displaystyle 0.5\log\lb 1 + \frac{\text{SNR} + \text{SNR}^2}{1 + \text{SNR} + \text{INR}} \rb, \nonumber \\
& = 0.5\log \lb 1 + 2\SNRt + \INRt + \SNRt^2\rb   - 0.5\log\lb 1 + \SNRt + \INRt\rb, \nonumber \\
& \approx 0.5 \lsqb \log\max\lb \SNRt^2, \INRt\rb   - \log(\text{INR})\rsqb.  \label{eq:discussion5}
\end{align}
Using the above mentioned definitions of $m$ and $n$, (\ref{eq:discussion5}) reduces to:
\begin{align}
R_s \leq \left\{\begin{array}{l l}
    2m-n &\mbox{for $1 < \alpha < 2$ } \\ 	
    0 & \mbox{for $\alpha \geq 2$}.
 \end{array}\right. \label{eq:discussion2b}
\end{align}
The above is the same as the outer bound for the SLDIC in Theorem~\ref{th:theoremSLDIC-outer2},
when $C=0$.

In the following section, achievable schemes on the secrecy rate for the GSIC are presented.
\section{GSIC: Achievable Schemes}\label{sec:achGaussian}
\subsection{Weak/moderate interference regime $(0 \leq \alpha \leq 1)$}\label{sec:weak-mod-achGaussian}
The achievable scheme is based on the approach used in Secs.~\ref{sec:SLDIC-ach-weak} and
\ref{sec:SLDIC-ach-weak-mod}, for the SLDIC. Again, the achievable scheme proposed for the deterministic model
is not directly applicable  to the Gaussian case due to the differences between the two models
mentioned earlier. In the case of the SLDIC, the achievable scheme used a
combination of interference cancelation, transmission of random bits, or both, depending on the value of
$\alpha$ and $C$. That scheme is extended to the Gaussian setting, as follows.

The message at transmitter $i$ is split into two parts: a \emph{non-cooperative private} part $(w_{pi})$
and a \emph{cooperative private} part $(w_{cpi})$. The non-cooperative private message is encoded using
stochastic encoding \cite{wyner-bell-1975}, and the cooperative private part is encoded using Marton's coding
scheme~\cite{marton-TIT-1979, wang-TIT-2011}. For the SLDIC, data bits transmitted at the lower levels $[1:m-n]$ are not received at the unintended receiver. Hence, these data bits remain secure. However, there is no one-to-one analogue of this in the GSIC, so the scheme does not extend directly. In the Gaussian case, for the non-cooperative private part, stochastic encoding is used to ensure secrecy. The transmitter~$i$ encodes the non-cooperative part $w_{pi} \in \mathcal{W}_{pi} = \{1,2,\ldots,2^{NR_{pi}}\}$ into $\xbold_{pi}^N$. A stochastic encoder is specified by a  conditional probability density $f_{pi}(x_{pi,k}|w_{pi})$ $(i=1,2)$, where $x_{pi} \in \mathcal{X}_{pi}$ and $w_{pi} \in \mathcal{W}_{pi}$, and it satisfies the following condition:
\begin{align}
\mysum_{x_{pi,k} \in \mathcal{X}_{pi}}f_{pi}(x_{pi,k}|w_{pi}) =1, \quad \qquad k=1,2,\ldots,N, \label{eq:weakint1}
\end{align}
where $f_{pi}(x_{pi,k}|w_{pi})$ is the probability that $x_{pi,k}$ is output by the stochastic encoder, when
message $w_{pi}$ is to be transmitted.

The cooperative private message $\wcpone \in \mathcal{W}_{cp1} = \{1,2,\ldots, 2^{NR_{cp1}}\}$ and
$\wcptwo \in \mathcal{W}_{cp2} = \{1,2,\ldots, 2^{NR_{cp2}}\}$ at transmitters $1$ and $2$ are encoded using Marton's coding scheme. One of the key aspects of the achievable scheme is in the proposed method for encoding  of the cooperative private message, which is chosen to ensure that this part of the message is completely canceled at the non-intended receiver. This corresponds to the scheme used for interference cancelation in the SLDIC. This serves two purposes: it cancels interference over the air, and simultaneously ensures secrecy. The transmitter~$2$ sends a dummy message along with the cooperative private message and the non-cooperative private message. Note that stochastic encoding is sufficient to ensure secrecy of the non-cooperative private message. However, the additional dummy message sent by the transmitter~$2$ can enhance the achievable secrecy rate, depending on the values of $\alpha$ and
$C$. In this case, both the receivers treat the dummy message as noise.
\subsubsection{Encoding and decoding}\label{sec:weak-encoder}
For the non-cooperative private part, transmitter $i$ $(i=1,2)$ generates $2^{N(R_{pi} + R_{pi}')}$ i.i.d. sequences of length $N$ at random according to
\begin{align}
P(\xbold_{pi}^N) = \prod_{k=1}^N P(x_{pi,k}). \label{eq:weakint2}
\end{align}
The $2^{N(R_{pi} + R_{pi}')}$ codewords in the codebook $C_{pi}$ are randomly grouped into $2^{NR_{pi}}$ bins, with each bin containing $2^{NR_{pi}'}$ codewords. Any codeword in $C_{pi}$ is indexed as $\xbold_{pi}^N(w_{pi},w_{pi}')$ for $w_{pi} \in \mathcal{W}_{pi}$ and $w_{pi}' \in \mathcal{W}_{pi}' = \{1,2,\ldots,2^{NR_{pi}'}\}$. In order to transmit $w_{pi}$, transmitter~$i$ selects a $w_{pi}' \in \mathcal{W}_{pi}'$ randomly and transmits the codeword $\xbold_{pi}^N(w_{pi},w_{pi}')$.

In order to transmit a dummy message, transmitter~$2$ generates $2^{NR_{d2}}$ i.i.d. sequences of length $N$ at random according to
\begin{align}
P(\xbold_{d2}^N) = \prod_{k=1}^N P(x_{d2,k}). \label{eq:weakint3}
\end{align}
The $2^{NR_{d2}}$ codewords in codebook $C_{d2}$ are randomly grouped into $2^{NR_{d2}'}$ bins, with
each bin containing $2^{NR_{d2}''}$ codewords (and thus $R_{d2}= R_{d2}' + R_{d2}''$). Any codeword in
$C_{d2}$ is indexed as $\xbold_{d2}^N(w_{d2}',w_{d2}'')$, where $w_{d2}' \in \mathcal{W}_{d2}' = \{1,2,\ldots,2^{NR_{d2}'}\}$
and $w_{d2}'' \in W_{d2}'' = \{1,2,\ldots,2^{NR_{d2}''}\}$. During encoding, transmitter~$2$ selects
$w_{d2}' \in \mathcal{W}_{d2}'$ and $w_{d2}'' \in \mathcal{W}_{d2}''$ independently at random and sends
the codeword $\xbold_{d2}^N(w_{d2}',w_{d2}'')$.

For the cooperative private message, the transmitter generates the cooperative private vector codeword
$\xbold_{cp}^N(\wcpone,\wcptwo)$ based on Marton's coding scheme according to
\begin{align}
P(\xbold_{cp}^N,\ubold_{1}^N,\ubold_{2}^N) = \prod_{k=1}^N P(\xbold_{cp,k},u_{1,k},u_{2,k}), \label{eq:weakint3a}
\end{align}
where $\ubold_1^N$ and $\ubold_2^N$ are auxiliary codewords. The choice of these codewords are
discussed in the proof of Theorem \ref{th:theorem_ach_weakint}. Finally, the non-cooperative private
codeword and cooperative private codeword are superimposed to form the transmit codeword at the
transmitter~$1$ and the non-cooperative private codeword, cooperative private codeword and the dummy
message codeword are superimposed to form the transmit codeword at the transmitter~$2$:
\begin{align}
& \xbold_{1}^N(\wcpone,\wcptwo,\wprvone,\wprvone') = \underline{\xbold}_{cp}^N[1] + \xbold_{p1}^N, \text{and }\nonumber \\
 &\xbold_{2}^N(\wcpone,\wcptwo,\wprvtwo,\wprvtwo',w_{d2}',w_{d2}'') \!=\! \underline{\xbold}_{cp}^N[2] +  \xbold_{p1}^N + \xbold_{d2}^N,  \label{eq:weakint4}
\end{align}
where  $\underline{\xbold}_{cp}^N$ is defined in (\ref{eq:pfhighint54}) in the proof of Theorem \ref{th:theorem_ach_weakint}.

For decoding, receiver~$i$ looks for a unique message tuple  such that $(\ybold_{i}^N, \ubold_{i}^N(\hat{\widetilde{w}}_{cpi}), \xbold_{pi}^N(\hat{w}_{pi},\hat{w}_{pi}'))$ is jointly typical. Based on the above coding strategy, the
 following theorem gives the achievable result on the secrecy rate.
\begin{theorem}\label{th:theorem_ach_weakint}
In the weak/moderate interference regime, the following rate is achievable for the GSIC with limited-rate transmitter cooperation and secrecy constraints at the receivers:
\begin{align}
& R_{1} + R_{p1}' \leq I(\ubold_1,\xbold_{p1};\ybold_1), \nonumber \\
& R_{1} + R_{p1}' \leq I(\xbold_{p1};\ybold_1|\ubold_1) + \min\lcb C_G, I(\ubold_1;\ybold_1|\xbold_{p1})\rcb, \label{eq:weakint5}
\end{align}
where $R_{p1}' = I(\xbold_{p1};\ybold_2|\xbold_{p2},\ubold_2)$. The achievable secrecy rate for the user~$2$
can be obtained by exchanging the indices $1$ and $2$ in (\ref{eq:weakint5}).
\end{theorem}
\begin{proof}
The proof is provided in Appendix \ref{sec:appendweakintf}.
\end{proof}
The achievable \emph{symmetric} secrecy rate for the GSIC is stated in the following Corollary.
\begin{corollary}\label{cor:cor_ach_weakint}
Using the achievable result in Theorem \ref{th:theorem_ach_weakint} and time-sharing between transmitters, following symmetric secrecy rate is achievable for the GSIC with limited-rate transmitter cooperation:
\begin{align}
R_s = \frac{1}{2}\lsqb R_1^*(1) + R_1^*(2)\rsqb, \label{eq:cor_ach_weakint1}
\end{align}
where $R_1^*(1)$ and $R_1^*(2)$ are the achievable secrecy rates for transmitter~$1$ in the first and second time slots, respectively, which are obtained by maximizing $R_s$ over parameters $\theta_i, \eta_i$ and $\beta_i$  $(i = 1, 2)$. The achievable rates for users~$1$ and $2$ in the first time slot are as follows:
\begin{align}
& R_1(1) \!\leq \!\left\{\begin{array}{l}
 			\!\!0.5\log\lb 1 + \frac{\sigma_u^2 + h_d^2 P_{p1}}{1 + h_{c}^2 P_{d2} + h_c^2 P_{p2}}\rb - R_{p1}',  \\
 			\!\!0.5\log\lb 1 + \frac{h_d^2P_{p1}}{1 + h_c^2P_{d2}+h_c^2P_{p2}}\rb  
 			+ \min\lcb C_G, 0.5\log\lb 1 + \frac{\sigma_u^2}{1 + h_c^2P_{d2} + h_c^2 P_{p2}}\rb \rcb-R_{p1}',\label{eq:cor_ach_weakint2}
                    \end{array}\right. \\
& R_2(1)\!\leq \! \left\{\begin{array}{l}
 			\!0.5\log\lb 1 + \frac{\sigma_u^2 + h_d^2 P_{p2}}{1 + h_{d}^2 P_{d2} + h_c^2 P_{p1}}\rb - R_{p2}', \\
 			\!0.5\log\lb 1 + \frac{h_d^2P_{p2}}{1 + h_d^2P_{d2}+h_c^2P_{p1}}\rb + \min\lcb C_G, 0.5\log\lb 1 + \frac{\sigma_u^2}{1 + h_d^2P_{d2} + h_c^2 P_{p1}}\rb \rcb-R_{p2}',
                    \end{array}\right. \label{eq:cor_ach_weakint3}
\end{align}
where $R_{p1}' = 0.5\log\lb 1 + \frac{h_c^2 P_{p1}}{1 + h_d^2 P_{d2}}\rb$, $R_{p2}' = 0.5\log\lb 1 + \frac{h_c^2P_{p2}}{1 + h_c^2 P_{d2}}\rb$, $\sigma_u^2 \triangleq (h_d^2 - h_c^2)^2 \sigma_z^2$, $\sigma_z^2 \triangleq \frac{\theta_1}{\theta_1 + \theta_2}\frac{P_1}{h_d^2 + h_c^2}$, $P_{p1} \triangleq \frac{\theta_2}{\theta_1 + \theta_2}P_1$, $P_{p2}=\frac{\eta_1}{\eta_1 + \eta_2}P'$, $P_{d2}=\frac{\eta_2}{\eta_1 + \eta_2}P'$, $P' = (P_2 - (h_d^2 + h_c^2)\sigma_z^2)$, $P_i \triangleq \beta_i P$ $(i=1,2)$ and $0 \leq (\theta_i, \eta_i, \beta_i) \leq 1$. The rate equations for the second time slot can be obtained by exchanging indices  $1$ and $2$ in (\ref{eq:cor_ach_weakint2}) and (\ref{eq:cor_ach_weakint3}).
\end{corollary}
\begin{proof}
The proof is provided in Appendix~\ref{sec:corappendweakintf}.
\end{proof}
\subsection{High/very high interference regime $(\alpha > 1)$}\label{sec:high-veryhigh-achGaussian}
The achievable scheme is based on the approach used for the SLDIC in case of high interference regime. The
achievable scheme for the SLDIC in Sec.~\ref{sec:SLDIC-ach-highint} used a combination of interference
cancelation, relaying of the other user's data bits, and transmission of random bits. In the case of the SLDIC,
as some of the interfering links are not present to the intended receiver, the levels corresponding to these
links can be directly used for the other user's data transmission. But, in the Gaussian setting, it is not possible
to relay the other user's data directly in this manner. The relationship between the corresponding achievable
schemes for the SLDIC and the GSIC will be made precise in the following paragraphs.

In the proposed scheme, user~$1$ sends a non-cooperative private message $(w_{p1})$ and a cooperative
private message $(w_{cp1})$. The other user transmits cooperative private message $(w_{cp2})$ along with
a dummy message $(w_{d2})$. For the SLDIC, the achievable scheme required transmission of random bits
for ensuring secrecy of data bits, in addition to the data bits that were sent with the help of cooperation.
Similarly, for the GSIC, the proposed scheme requires stochastic encoding and transmission of a dummy
message by the other user, in order to ensure secrecy of the non-cooperative private message sent by
user~$1$. It is important to note that stochastic encoding alone cannot ensure secrecy of the non-cooperative
private part of the message. For the cooperative private part of the message $(w_{cpi})$, the coding scheme
is the same as that mentioned in Sec.~\ref{sec:weak-mod-achGaussian}.

The transmission of the dummy message $x_{d2}$ by transmitter~$2$ can be considered as using another 
stochastic encoder $f_{d2}$, which is specified by a probability density function $f_{d2}(x_{d2,k})$, with  
$x_{d2,k}~\in~\mathcal{X}_{d2}$ and $\mysum_{x_{d2,k} \in \mathcal{X}_{d2}}f_{d2}(x_{d2,k})=1$. The rate 
$R_{d2}$ of the dummy message sent by transmitter~$2$ and the rate sacrificed by transmitter~$1$ in 
stochastic encoding in order to confuse the eavesdroppers at receivers~$1$ and $2$, respectively, are chosen 
such that  the non-cooperative private message sent by transmitter~$1$ remains secure at receiver~$2$, and receiver~$1$ is able to decode the dummy message. At transmitter~$1$, the cooperative private message and the non-cooperative private message are superimposed to form the transmit codeword $(\xbold_1^N)$. Finally, at transmitter~$2$, the cooperative private message and the dummy information are superimposed to form the transmit codeword $(\xbold_2^N)$.  In contrast to the achievable scheme for
the weak/moderate interference regime, the dummy message sent by one of the transmitters~$i$ is 
required to be decodable at the receiver~$j$ $(i \neq j)$.
\subsubsection{Encoding and decoding}
The encoding for the non-cooperative private message at transmitter $1$ and the cooperative private message at both the transmitters are the same as described in Sec.~\ref{sec:weak-encoder}. In order to transmit the dummy message, transmitter~$2$ chooses $\xbold_{d2}^N(\dumytwo)$ for $\dumytwo \in \mathcal{W}_{d2}$. The codewords transmitted from the two transmitters are given by:
\begin{align}
& \xbold_{1}^N(\wcpone,\wcptwo,\wprvone,\wprvone') = \underline{\xbold}_{cp}^N[1] + \xbold_{p1}^N,  \text{ and } \nonumber \\
&\xbold_{2}^N(\wcpone,\wcptwo,\dumytwo) = \underline{\xbold}_{cp}^N[2] + \xbold_{d2}^N, \label{eq:highint4}
\end{align}
where  $\underline{\xbold}_{cp}^N$ is defined in (\ref{eq:pfhighint54}) in the proof of Theorem~\ref{th:theorem_ach_weakint}.

For decoding, receiver~$1$ looks for a unique message tuple  such that $(\ybold_{1}^N,\ubold_{1}^N(\tildewcponehat),\xbold_{d2}^N(\dumytwohat),$ $\xbold_{p1}^N(\wprvonehat,\wprvonehat'))$ is jointly typical. Receiver~$2$ looks for a index $\wcptwohat$ such that $(\ybold_{2}^N, \ubold_{2}^N(\tildewcptwohat))$ is jointly typical.

Based on the above coding strategy, the following theorem gives the achievable result on the secrecy rate.
\begin{theorem}\label{th:theorem_ach_highint}
In the high interference regime, the following rate is achievable for the GSIC with limited-rate transmitter cooperation and secrecy constraints at the receivers:
\begin{align}
& R_{1} + R_{p1}' \leq \min \lsqb I(\ubold_1,\xbold_{p1};\ybold_1|\xbold_{d2}), I(\xbold_{p1};\ybold_1|\ubold_1,\xbold_{d2})  + \min\lcb I(\ubold_1;\ybold_1|\xbold_{p1},\xbold_{d2}),C_G\rcb\rsqb, \nonumber \\
& R_{1} + R_{p1}' + R_{d2} \leq \min\lsqb I(\ubold_1,\xbold_{p1},\xbold_{d2};\ybold_1),  I(\xbold_{p1},\xbold_{d2};\ybold_1|\ubold_1) +\min\lcb I(\ubold_1;\ybold_1|\xbold_{p1},\xbold_{d2}),C_G \rcb, \right. \nonumber \\
& \qquad \qquad \quad\left.I(\xbold_{p1};\ybold_1|\ubold_1,\xbold_{d2}) + I(\ubold_1,\xbold_{d2};\ybold_1|\xbold_{p1})\rsqb, \nonumber \\
& R_{1} + R_{p1}' + 2R_{d2} \!\leq \! I(\xbold_{p1},\xbold_{d2};\ybold_1|\ubold_1) \! + \! I(\ubold_1,\xbold_{d2};\ybold_1|\xbold_{p1}), \nonumber \\
& R_2 \leq \min\lcb I(\ubold_2;\ybold_2), C_G\rcb, \nonumber \\
& R_{d2} \leq I(\xbold_{d2};\ybold_1|\ubold_1, \xbold_{p1}), \label{eq:highint5}
\end{align}
where $R_1 \triangleq R_{p1} + R_{cp1}$, $R_2 \triangleq R_{cp2}$, $R_{p1}' \triangleq I(\xbold_{p1};\ybold_2|\ubold_2)$, and $R_{d2} \triangleq I(\xbold_{d2};\ybold_2|\xbold_{p1}, \ubold_2)$.
\end{theorem}
\begin{proof}
The proof is provided in Appendix \ref{sec:appendhighintf}.
\end{proof}
The achievable symmetric secrecy rate is stated in the following Corollary.
\begin{corollary}\label{cor:cor_ach_highint}
Using the achievable result in Theorem \ref{th:theorem_ach_highint} and time-sharing between transmitters, the following symmetric secrecy rate is achievable for the GSIC with limited-rate transmitter cooperation:
\begin{align}
R_s = \frac{1}{2}\lsqb R_1^*(1) + R_1^*(2)\rsqb, \label{eq:cor_ach_highint1}
\end{align}
where $R_1^*(1)$ and $R_1^*(2)$ are the achievable secrecy rates for transmitter~$1$ in the first and second time slots, respectively, which are obtained by maximizing $R_s$ over parameters $\theta_i, \eta_i$ and $\beta_i$  $(i = 1, 2)$. The achievable rates for users~$1$ and $2$ in the first time slot are as follows:
\begin{align}
&  R_1(1) \leq  \left\{\begin{array}{l}
			\min \lsqb 0.5\log(1 + \sigma_{u}^2  + h_d^2 P_{p1}), 0.5\log(1 + h_d^2 P_{p1}) + \min\lcb 0.5\log(1 + \sigma_{u}^2), C_G\rcb\rsqb - R_{p1}', \\ 
			\min \lsqb 0.5\log(1 + \sigma_{u}^2 + h_d^2 P_{p1} + h_c^2 P_{d2}), 
			 0.5\log(1 + \sigma_{u}^2 + h_c^2 P_{d2}) \right. \\ \left.  \quad + \min\lcb 0.5\log(1 + \sigma_{u}^2),C_G\rcb,  0.5\log(1 + h_d^2 P_{p1}) + 0.5 \log(1 + \sigma_{u}^2 + h_c^2 P_{d2})\rsqb \nonumber \\ \qquad \qquad - (R_{p1}' + R_{d2}), \\ 
0.5\log(1 + h_d^2 P_{p1} + h_c^2 P_{d2}) + 0.5\log(1 + \sigma_{u}^2 + h_c^2 P_{d2}) - (R_{p1}' + 2R_{d2}) \end{array}\right.
\end{align}
and
\begin{align}
 & R_2(1) = \min\lcb 0.5\log\lb 1 + \frac{\sigma_u^2}{1 + h_d^2 P_{d2} + h_c^2 P_{p1}}\rb, C_G\rcb, \label{eq:cor_ach_highint3}
\end{align}
where $R_{p1}' = 0.5\log\lb 1 + \frac{h_c^2 P_{p1}}{1 + h_d^2 P_{d2}}\rb$, $R_{d2} = 0.5\log(1 + h_d^2 P_{d2})$, $\sigma_u^2 \triangleq (h_d^2 - h_c^2)^2 \sigma_z^2$, $\sigma_z^2 \triangleq \frac{\theta_1}{\theta_1 + \theta_2}\frac{P_1}{h_d^2 + h_c^2}$, $P_{p1} \triangleq \frac{\theta_2}{\theta_1 + \theta_2}P_1$, $P_{d2} \triangleq (P_2 - (h_d^2 + h_c^2)\sigma_z^2)^+$, $P_i \triangleq \beta_i P$ and $0 \leq (\theta_i, \beta_i) \leq 1$. The achievable rate equation for the second time slot can be obtained by exchanging indices~$1$~and~$2$ in (\ref{eq:cor_ach_highint3}).
\end{corollary}
\begin{proof}
The proof is provided in Appendix~\ref{sec:corappendhighintf}.
\end{proof}

\textit{Remarks:}
\begin{itemize}
 \item The terms $R_{p1}' = I(\xbold_{p1};\ybold_2|\xbold_{p2},\ubold_2)$ and
 $R_{p1}' = I(\xbold_{p1};\ybold_2|\ubold_2)$ in Theorems~\ref{th:theorem_ach_weakint} and
 \ref{th:theorem_ach_highint}, respectively, correspond the loss in rate due to
 stochastic encoding. As the capacity of the cooperative link increases, more
 power is assigned to the cooperative private message and hence, the loss in
 rate due to stochastic encoding decreases.
 \item When $1<\alpha < 2$ and  $C_G=0$, the proposed scheme cannot achieve non-zero secrecy rate
 without transmission of dummy message in case of GSIC. In the case
 of SLDIC, also, the proposed scheme uses random bits transmission to achieve
 non-zero secrecy rate.
\item When $C_G \approx \lceil 0.5\log\lb 1 + h_c^2 P\rb \rceil $, the achievable secrecy rate is
 very close to the outer bound (See Fig.~\ref{fig:GICresult4}).
\item  In all the interference regimes, the proposed scheme always achieves nonzero secrecy rate with
 cooperation $(\text{i.e.}, C>0 \text{ and } C_G >0 )$ in case of SLDIC as well as GSIC, except for the
 $\alpha = 1$ case.
 \item  The achievable schemes stated in Theorems~\ref{th:theorem_ach_weakint}  and \ref{th:theorem_ach_highint} 
 for the Gaussian
 case can be extended to the asymmetric case, when $h_{ji} < h_{ii}$ (Theorem~\ref{th:theorem_ach_weakint}) and $h_{ji} > h_{ii}$
 (Theorem~\ref{th:theorem_ach_highint}), respectively. For Theorem~\ref{th:theorem_ach_weakint}, the condition $h_{ji} < h_{ii}$  is required to enable
transmitter~$i$ to send the non-cooperative message securely to receiver~$i$ using stochastic encoding.
For Theorem~\ref{th:theorem_ach_highint}, the condition $h_{ji} > h_{ii}$ is required due to the fact that the dummy message sent by
transmitter~$i$ needs to be decodable at receiver~$j$ $(j \neq i)$ and
should not be decodable at receiver~$i$. However, for the extension, the choice of the auxiliary
codewords for the cooperative private message needs to be modified as mentioned in the proof of Theorem~\ref{th:theorem_ach_weakint}.
\end{itemize}
\subsection{Relation between the achievable rates for SLDIC and GSIC}\label{sec:ach-det-Gaussian}
In the following, it is shown that the achievable rates for both the models are
approximately the same at high SNR and INR, when $(0 \leq \alpha \leq \frac{1}{2})$
and for all the values of the capacity of the cooperative link. For ease of presentation, it is assumed that $0.5\log\SNRt$, $0.5\log\INRt$ and $C_G$ are integers. Without loss of generality, it is also assumed that $h_d=1$.

Consider the lower bound on the secrecy rate for transmitter~$1$ derived in Corollary~\ref{cor:cor_ach_weakint}, when the transmit power of the dummy message is set to zero. Correspondingly, notice that, in the deterministic case, transmission of jamming signal (random bits transmission) is not used in this regime of $\alpha$. Consider the following power allocation
\begin{align}
& P_{p1} = P_{p2}= \frac{1}{h_c^2}, \nonumber \\
&  \sigma_{1z}^2 = \sigma_{2z}^2 = \sigma_z^2 = 
\frac{1}{2} \lb P - \frac{1}{h_c^2}\rb, \nonumber \\
\text{ and } & P_{d2} = 0. \label{eq:intcanc3}
\end{align}
For high SNR $(\SNRt \gg 1)$, the power allocation to the non-cooperative private
message in (\ref{eq:intcanc3}) is always feasible. With this power allocation, the two bounds on $R_1$ in
(\ref{eq:cor_ach_weakint2}) reduce to
\begin{align}
& R_1 \leq 0.5\log \lb 1 + \frac{1}{2} \lsqb \frac{1}{2}(1 - h_c^2)^2\lb P - \frac{1}{h_c^2}\rb + \frac{1}{h_c^2}
\rsqb\rb - 0.5, \label{eq:intcanc4}  \\
& R_1 \leq 0.5\log\lb 1 + \frac{1}{2} \frac{1}{h_c^2}\rb + \min \lcb C_G, \underbrace{0.5\log\lb 1 + \frac{1}{4} (1 - h_c^2)^2\lb P - \frac{1}{h_c^2}\rb\rb}_{I_1}\rcb - 0.5.  \label{eq:intcanc5}
\end{align}
First consider the bound given in (\ref{eq:intcanc4}).
\begin{align}
 R_{11}  & = 0.5\log \lb 1 \!+\! \frac{1}{2} \lsqb \frac{1}{2} \lb 1 \! -\! \frac{\INRt}{\SNRt} \rb^2 \!\! \lb \SNRt \!-\! \frac{\SNRt}{\INRt}\rb \!+\! \frac{\SNRt}{\INRt}\rsqb\rb  - 0.5, \nonumber \\
 & \stackrel{(a)}{\approx} 0.5\log \frac{\SNRt}{4} - 0.5, \nonumber \\
 & = m - 1.5, \label{eq:intcanc6}
\end{align}
where (a) is obtained using the fact that for $ \lb 0 \leq \alpha \leq \frac{1}{2}
\rb$, $\SNRt \geq \INRt^2$ and $\SNRt, \INRt \gg 1$.

When $C_G \leq I_1$, the RHS of (\ref{eq:intcanc5}) becomes
\begin{align}
R_{12} & = 0.5\log\lb 1 + \frac{1}{2h_c^2}\rb + C_G - 0.5, \nonumber \\
& = 0.5\log \lb 1 + \frac{1}{2} \frac{\SNRt}{\INRt}\rb + C_G -0.5, \nonumber \\
& \approx m- n + C_G - 1. \label{eq:intcanc7}
\end{align}
When $C_G > I_1$, the RHS of (\ref{eq:intcanc5}) becomes
\begin{align}
R_{13} 
& = 0.5\log \lb 1 + \frac{\SNRt}{2\INRt} \rb + 0.5\log\lb 1 +  \frac{1}{4}\lb1 - \frac{\INRt}{\SNRt}\rb^2  \lb \SNRt  - \frac{\SNRt}{\INRt}\rb\rb - 0.5. \label{eq:intcanc8}
\end{align}
One can show that $R_{11} \leq R_{13}$, and hence, the approximate achievable
secrecy rate for high SNR and INR becomes
\begin{align}
R_1 = \min\lcb m - 1.5, m -n + C_G - 1\rcb.  \label{eq:intcanc9}
\end{align}
From (\ref{weakach2}), the result in the deterministic case for this range of $\alpha$ is
\begin{align}
R_1 = \min\lcb m, m - n + C\rcb. \label{eq:intcanc10}
\end{align}
From (\ref{eq:intcanc9}) and (\ref{eq:intcanc10}), it is clear that the
achievable results for both the models are approximately equal for all
values of $C_G$, when $\lb 0 \leq \alpha \leq \frac{1}{2}\rb$.

In the other interference regimes, it is tedious to establish a precise connection between the achievable results for the
deterministic model and the Gaussian model, due to the complexity of the rate expressions of both the models.
Nonetheless, it can be noticed that there exists a close resemblance in the behavior of the rate plots against $\alpha$ for
both the models (See~Figs.~\ref{fig:gdof1}~and~\ref{fig:GICresult2}, Figs.~\ref{fig:gdof2} and \ref{fig:GICresult3}).
\section{Discussion and Numerical Examples}\label{sec:results-discussion}
\subsection{Comparison with existing results}
Some observations on how the bounds derived in this work stand in relation to existing works are as follows:
\begin{enumerate}
\item When $C=0$ and $\alpha = \frac{1}{2}$, the achievable rate result for the SLDIC in Section~\ref{sec:SLDIC-ach-weak}
reduces to the achievable rate result for the SLDIC in \cite{yates-isit-2008} with semi-secret constraint at each
receiver. The semi-secret constraint at each receiver depends on trusting the other transmitters.
\item When $(0 \leq \alpha \leq \frac{1}{2})$, the achievable rate result for the SLDIC in Sec~\ref{sec:SLDIC-ach-weak} is found to match with the achievable result for the SLDIC in~\cite{wang-TIT-2011} (See Figs.~\ref{fig:gdof1} and \ref{fig:gdof2}). As $\alpha$ increases, in~\cite{wang-TIT-2011}, the receiver can decode some part of interference and can achieve higher rate. Here, due to the secrecy constraints, the receivers cannot decode other users' messages, and hence, the achievable scheme is completely different. Also, for some values of $\alpha$, the achievable scheme proposed in this paper for the SLDIC requires the exchange of only random bits through the cooperative link, in contrast with the achievable scheme in~\cite{wang-TIT-2011}.
\item When $C_G = 0$, the system reduces to the $2$-user GSIC without cooperation, which was studied in \cite{liu-TIT-2008}. The achievable rate result in Theorem~\ref{th:theorem_ach_weakint} and Corollary~\ref{cor:cor_ach_weakint} reduce to the results reported in \cite{liu-TIT-2008} in this case.
\item When $C_G=0$, the achievable result in Theorem \ref{th:theorem_ach_highint} reduces to the achievable result in  \cite[Theorem~$3$]{tang-TIT-2011} for the high/very high interference regime $(\alpha > 1)$ for the wiretap channel with  a helping interferer.
\item When the capacity of the cooperative links are sufficiently large, then the GSIC with transmitter cooperation reduces to a $2$-user Gaussian MIMO broadcast channel (GMBC) with two antennas at transmitter and one antenna at each receiver. The achievable rate result in Corollaries \ref{cor:cor_ach_weakint} and \ref{cor:cor_ach_highint} are found to be very close to the achievable rate result in \cite[Theorem~$1$]{liu-TIT-2009} for the GMBC, as shown in Fig. \ref{fig:comp-GMBC}.
\item The proposed outer bounds for the GSIC with limited rate transmitter cooperation in Theorems \ref{th:theorem_GSIC_outer1}-\ref{th:theorem_GSIC_outer3} are compared with existing outer bounds for the GSIC with secrecy constraints at each receiver~\cite{he-CISS-2009, tang-TIT-2011}, when $C_G=0$, in Fig. \ref{fig:GICresult1}. It can be observed that the outer bounds derived in this work improve over the best known outer bounds in the literature even in the absence of cooperation, i.e., when $C_G=0$.
\end{enumerate}
\begin{figure}
	\begin{center}
		\setxysizeo
		\epsffile{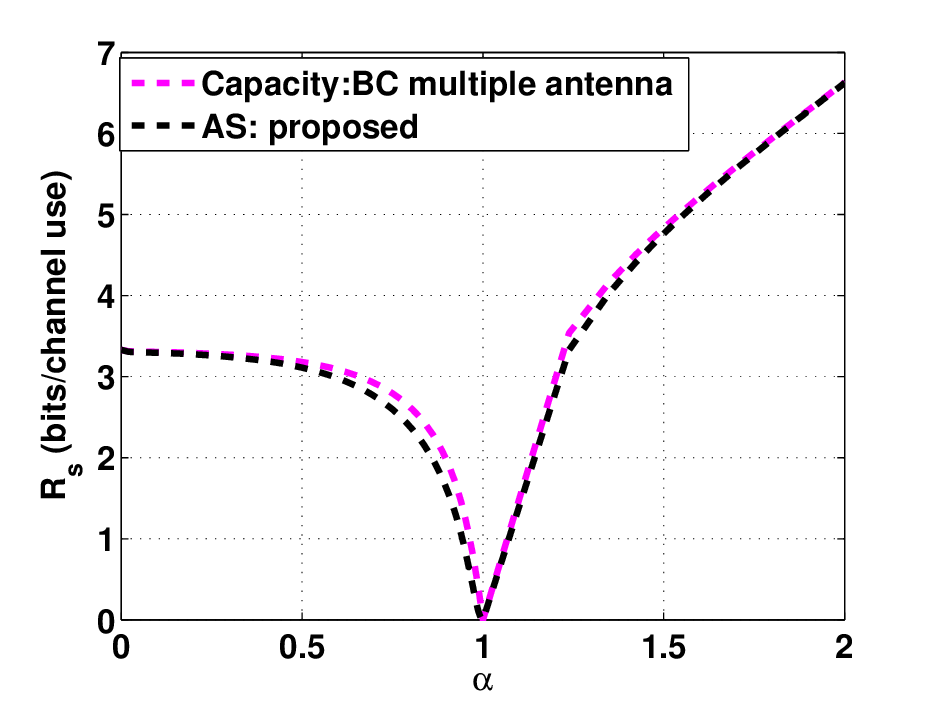}
		\vspace{-0.25cm}
		\caption{Achievable secrecy rate for the GSIC with large $C_G$, and the capacity of GMBC with two transmit antennas and one receive antenna at each receiver \cite{liu-TIT-2009}. For the GSIC and GMBC the individual power constraints at each transmitter are $P=100$ and $P=200$, respectively. The channel gain to the intended receivers in both cases is $h_d=1$. In the legend, \textit{AS: proposed} stands for the achievable scheme proposed in this work.}\label{fig:comp-GMBC}
	\end{center}
\end{figure}
\begin{figure}
	\begin{center}
		\setxysizeo
		\epsffile{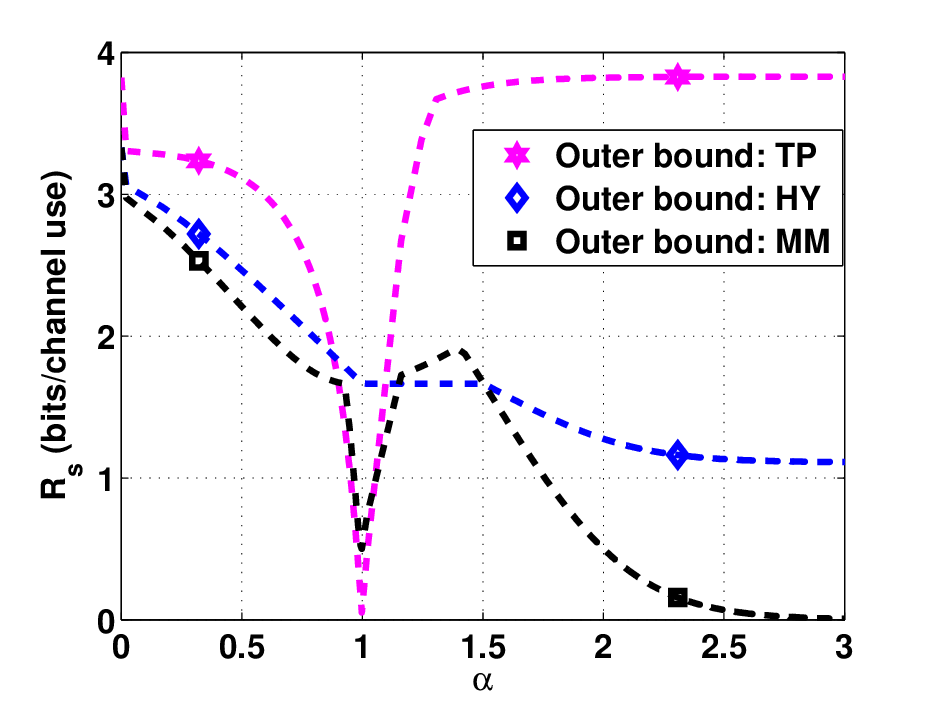}
		\vspace{-0.25cm}
		\caption{Outer bound on the symmetric secrecy rate for GSIC with $C_G=0$, $P=100$ and $h_d=1$. In the legend, TP stands for the outer bound derived in \cite{tang-TIT-2011}, HY stands for the outer bound on secrecy rate in \cite{he-CISS-2009} and MM stands for the outer bound derived in this work.}\label{fig:GICresult1}
	\end{center}
\end{figure}

In the following sections, some numerical examples are presented for the deterministic and Gaussian cases, to get insights into the bounds for different values of $C$, over different interference regimes.
\subsection{Numerical examples in case of SLDIC}\label{sec:numerical-SLDIC}
\begin{figure}
	\begin{center}
		\setxysizeo
		\epsffile{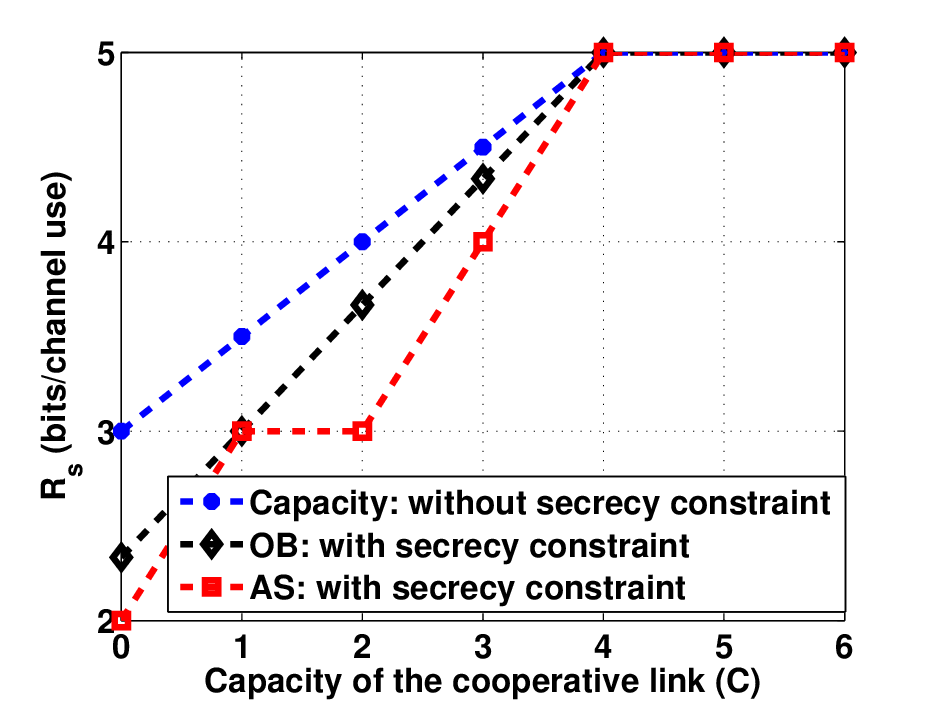}
		\vspace{-0.25cm}
		\caption{Outer bound on the secrecy rate of the SLDIC with $m=5$ and $n=4$. The capacity of the SLDIC without secrecy constraints is known in the literature  \cite{wang-TIT-2011}. In the legend, OB and AS stand for the outer bound and achievable scheme derived in this work, respectively.} \label{fig:ratecomp1}
	\end{center}
\end{figure}
In Fig.~\ref{fig:ratecomp1}, the outer bound in Theorem~\ref{th:theoremSLDIC-outer1} is plotted along with
the achievable secrecy rate given in (\ref{eq:sldic-modach1}) for the $(m,n) = (5,4)$ case. Also plotted is the
per user capacity of the SLDIC with transmitter cooperation, but without the secrecy constraints \cite{wang-TIT-2011}.
It can be observed that the proposed scheme is optimal, when $C=1$ and $C \geq 4$. However, it is not possible
to achieve the capacity without the secrecy constraint, when $C \leq 3$. When $C \geq 4$, there is no loss in
the achievable rate due to the secrecy constraint at receivers.
\begin{figure}
	\begin{center}
		\setxysizeo
		\epsffile{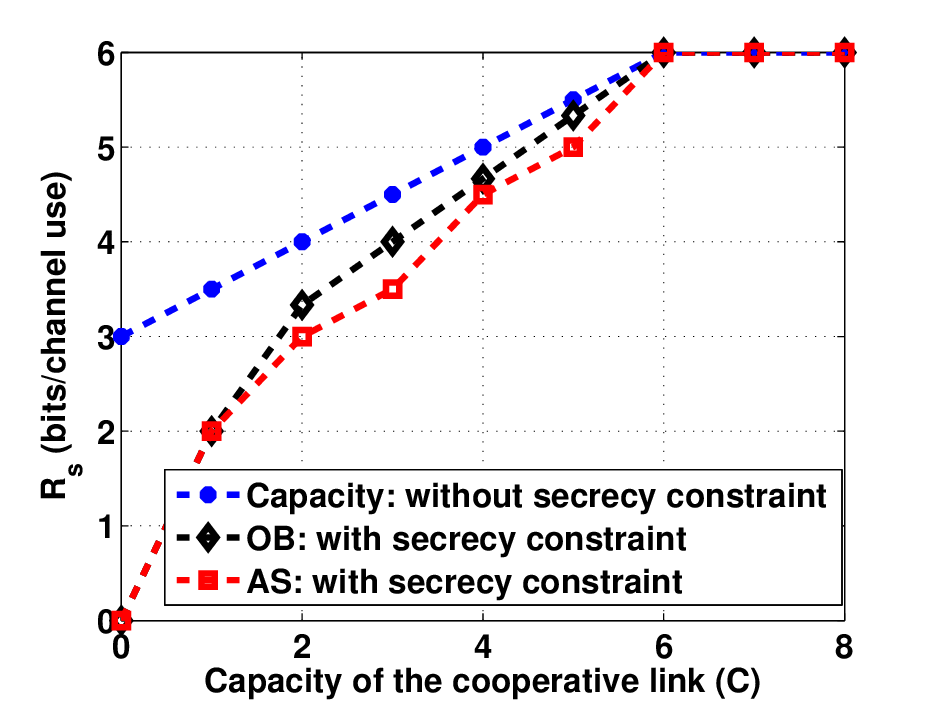}
		\vspace{-0.25cm}
		\caption{Outer bound on the secrecy rate of the SLDIC with $m=3$ and $n=6$. The capacity of the SLDIC without secrecy constraints is known in the literature  \cite{wang-TIT-2011}. In the legend, OB and AS stand for the outer bound and achievable scheme derived in this work, respectively. In this case the inner and outer bound match for $C=0,1$, and are fairly close at higher values of $C$. For $C=1$, the capacity achieving scheme uses random bits sharing through the cooperative link.} \label{fig:ratecomp2}
	\end{center}
\end{figure}
In Fig.~\ref{fig:ratecomp2}, the minimum of the outer bounds in Theorems~\ref{th:theoremSLDIC-outer1} and \ref{th:theoremSLDIC-outer2} is plotted as a function of $C$, with $(m,n) = (3,6)$. Also plotted is the achievable secrecy rate given in (\ref{eq:sldic-highach3}). From the plot, it can be observed that it is not possible to achieve a nonzero secrecy rate without cooperation between the transmitters, i.e., when $C=0$. The achievable scheme, which uses random bits sharing through the cooperative link and interference cancelation, is optimal for $C=1$. It can be observed the secrecy constraint results in a positive rate penalty, in the sense that it is not possible to achieve the capacity without the secrecy constraint, for $C \leq 5$.

In Figs.~\ref{fig:gdof1} and \ref{fig:gdof2}, the outer bound on the symmetric rate is plotted against~$\alpha$
for a given value of $C$, along with the per user capacity of the SLDIC with transmitter cooperation, but without
the secrecy constraints \cite{wang-TIT-2011}, and the inner bounds for the SLDIC with secrecy constraints
at the receiver. In order to generate these plots, $m$ is chosen to be $400$ and $n$ is varied from $0$ to $4m$,
and the rates are normalized by $m$.

In Fig.~\ref{fig:gdof1}, the achievable secrecy rate and the capacity without secrecy constraints \cite{wang-TIT-2011}
match when $0 \leq \alpha \leq \frac{1}{2}$. Hence, for this regime, it is not required to derive an outer bound.
When $\frac{1}{2} < \alpha \leq \frac{2}{3}$, in the absence of the secrecy constraint, the capacity increases with increase in the value of $\alpha$, as the receivers are able to decode some part of the interference. However, with the secrecy constraint, the receiver cannot decode the other user's message, and, hence, the achievable rate decreases with $\alpha$. When $\frac{2}{3} < \alpha < 1$, the achievable secrecy rate meets the outer bound at some of the points and the fluctuating behavior of the achievable rate is due to the floor-operation involved in the rate expression. In this regime, the transmission of random bits help to compensate for the loss in rate, to some extent. At $\alpha=1$, there exists a point of discontinuity, as no nonzero secrecy rate is achievable.
Intuitively, one would expect that the achievable secrecy rate should monotonically decrease with $\alpha$, because of the reasoning mentioned above. Interestingly, the achievable secrecy rate increases with increase in the value of $\alpha$, when $1 < \alpha \leq 1.5$, although the increase is not monotonic in nature due to the floor operation involved in the rate expression. The increase in the achievable secrecy rate arises due to the improved ability of the transmitters to jam the data bits at the unintended receivers by sending random bits as, $\alpha$ increases. However, when $1.5 < \alpha < 2$, the achievable secrecy rate decreases with increase in the value of $\alpha$ and the outer bound meets the inner bound. When $\alpha \geq 2$, it is no longer possible to achieve a nonzero secrecy rate.
\begin{figure}
	\begin{center}
		\setxysizeo
		\epsffile{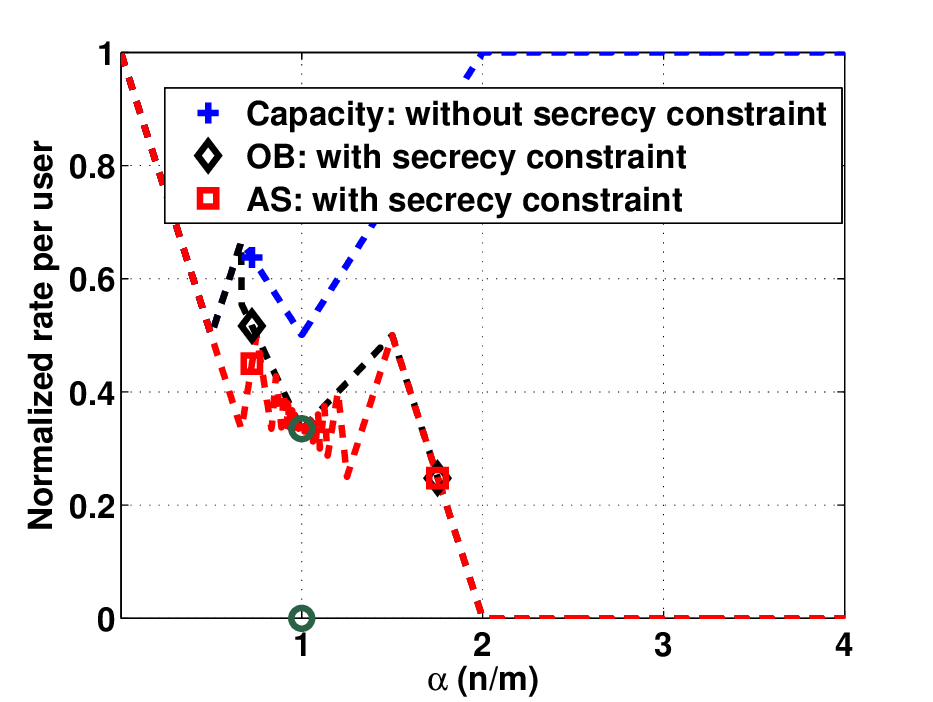}
		\vspace{-0.25cm}
		\caption{Rate normalized w.r.t $m$ for the SLDIC with $C=0$. Although there exists a gap between the inner
			and outer bounds with the secrecy constraint, they match at many points too. In particular, in the initial part
			of the weak interference regime, i.e., for $0 < \alpha < \frac{1}{2}$ and regime $(\alpha  \geq 1.5)$, the
			capacity of the SLDIC is achieved by the proposed scheme.}\label{fig:gdof1}
	\end{center}
\end{figure}

In Fig. \ref{fig:gdof2}, compared to the $C=0$ case, the achievable secrecy rate is higher in all the
interference regimes due to the cooperation, except when $\alpha=1$. The cooperation between the
transmitters not only eliminates the interference, but at the same time ensures secrecy. Also, the utility of
random bit transmission decreases with increase in the value of $C$. Interestingly, it is possible to achieve a
nonzero secrecy rate even when $\alpha \geq 2$, and the achievable scheme is optimal in this case.
\begin{figure}
	\begin{center}
		\setxysizeo
		\epsffile{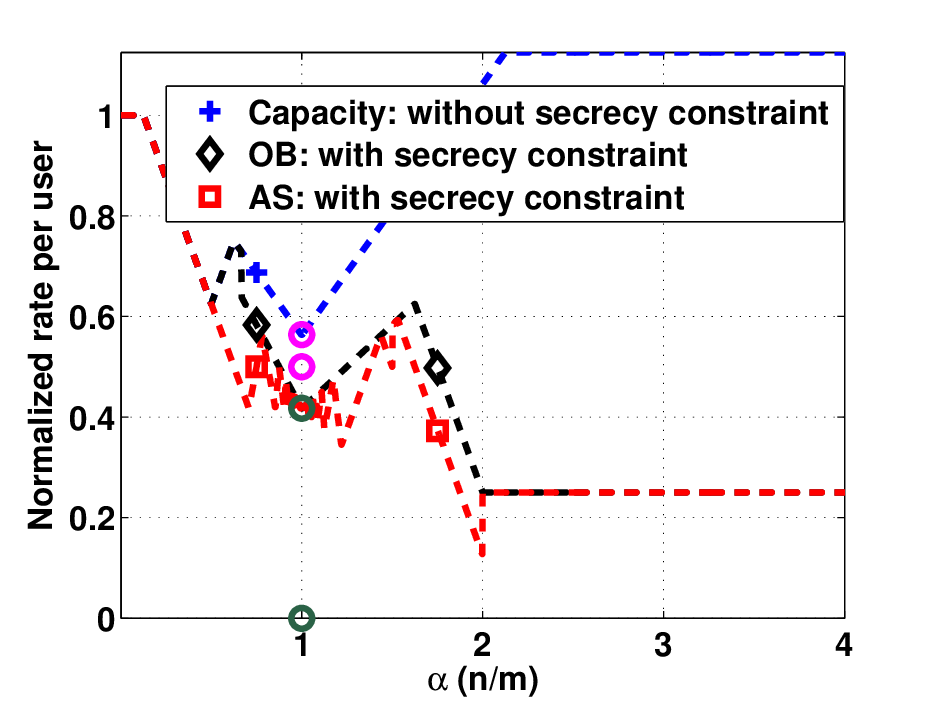}
		\vspace{-0.25cm}
		\caption{Normalized rate w.r.t $m$ for the SLDIC with $C=50$. In this case, we obtain the capacity of the
			SLDIC in the initial part of the weak interference regime $(0 < \alpha \leq \frac{1}{2})$ and in the very high
			interference regime $(\alpha \geq 2)$.}\label{fig:gdof2}
	\end{center}
\end{figure}
\subsection{Numerical examples in case of GSIC}
\begin{figure}
	\begin{center}
		\setxysizeo
		\epsffile{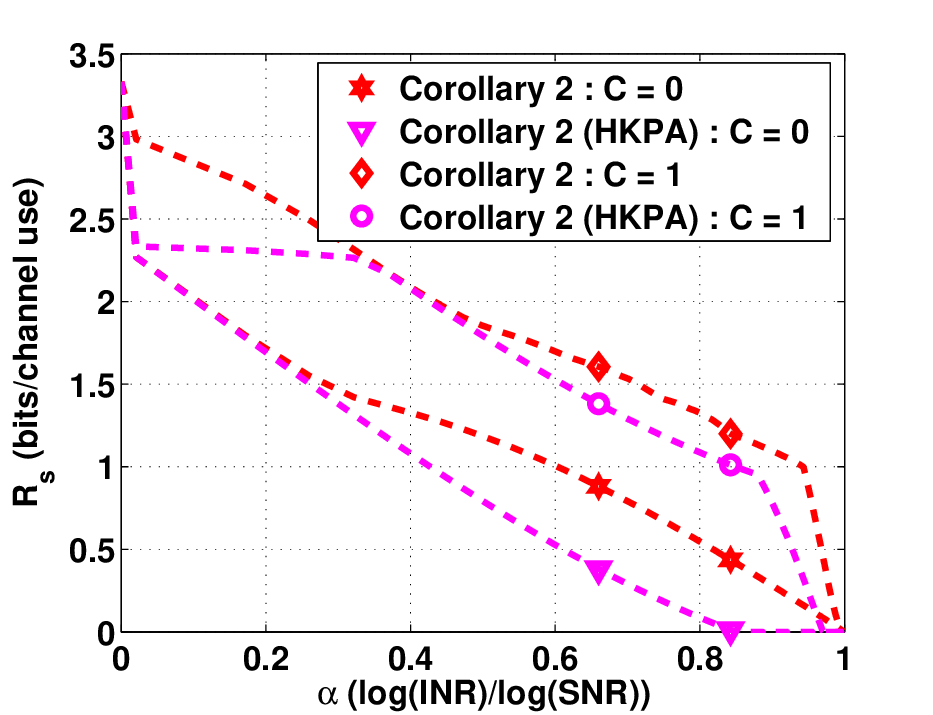}
		\vspace{-0.25cm}
		\caption{Comparison of achievable schemes in Corollary \ref{cor:cor_ach_weakint} with different power allocations: $P=20$~dB and $h_d=1$.} \label{fig:compare_ach_HK}
	\end{center}
\end{figure}
In Fig.~\ref{fig:compare_ach_HK}, the achievable result in Corollary~\ref{cor:cor_ach_weakint} is plotted against $\alpha$, for different values of $C_G$, with two types of power allocations. In the first case, no power is allotted for transmitting the dummy message. The power allocations for the non-cooperative private message and cooperative private message are discussed below. For the SLDIC, in the weak and moderate interference regimes, the data bits transmitted on the lower levels $[1:m-n]$ will not be received at the unintended receiver. For the GSIC, this corresponds to transmitting the non-cooperative private message such that it is received at the noise floor of the unintended receiver. In the existing literature, this type of power allocation has been used for the private message\footnote{In \cite{etkin-TIT-2008}, there is no secrecy constraint at the receiver and the terminology \textit{private} arises due to the fact that this part of the message is not required to be decodable at the unintended receiver.} in the Han-Kobayashi (HK)-scheme \cite{etkin-TIT-2008}, and hence, this special case is termed as
HKPA (HK-type power allocation) scheme in this paper. The remaining power is allotted for transmitting the cooperative private message. In the second case, the achievable result in Corollary~\ref{cor:cor_ach_weakint}, which involves transmission of a dummy message, is plotted. When $C_G=0$ and $\alpha > 0.4$, the scheme in Corollary~\ref{cor:cor_ach_weakint} outperforms the HKPA scheme. The gain in the achievable rate largely arises from the transmission of the dummy message. When $C_G=1$, the gap between the two schemes decreases, except for the initial part of the weak interference regime.
\begin{figure}
	\begin{center}
		\setxysizeo
		\epsffile{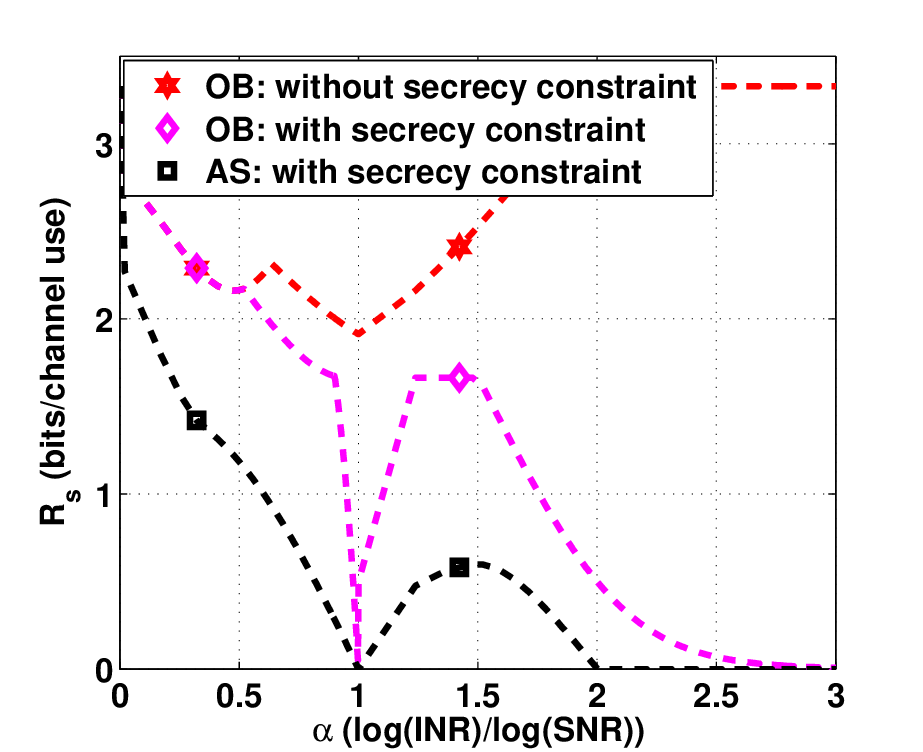}
		\vspace{-0.25cm}
		\caption{Secrecy rate in case of GSIC with $P=100$ and $C_G=0$. In the legend, OB stands for the outer bound and
			AS stands for the achievable scheme. Interestingly, the achievable secrecy rate increases with increase in the
			value of $\alpha$, when $1 < \alpha \leq 1.5$. However, when $1.5 < \alpha <2$, the achievable secrecy rate
			decreases with increase in the value of $\alpha$.} \label{fig:GICresult2}
	\end{center}
\end{figure}
In Fig. \ref{fig:GICresult2}, the achievable symmetric secrecy rate in Corollaries \ref{cor:cor_ach_weakint} and \ref{cor:cor_ach_highint} are plotted against $\alpha$, for $C_G=0$ and $P=100$. Also plotted is the outer bound on the symmetric rate in case of GSIC without the secrecy constraint at receiver~\cite{wang-TIT-2011}. While plotting the outer bound with secrecy constraint, the minimum of the outer bounds derived in this work and outer bounds in \cite{wang-TIT-2011,he-CISS-2009,tang-TIT-2011} is taken for the $C_G=0$ case. When $(0 \leq \alpha \leq 1)$, the achievable secrecy rate decreases with increase in the value $\alpha$. At $\alpha=1$, the achievable secrecy rate becomes zero. The figure also reveals an interesting trade off between stochastic encoding and dummy message transmission in the high interference regime. Initially, as $\alpha$ increases, receiver~$i$ can decode more of the interference caused due to the dummy message transmission by transmitter~$j$, $j \neq i$, which, along with a relatively minimal rate loss due to stochastic encoding, leads to a net increase in rate with $\alpha$. However, with further increase in $\alpha$, the loss in rate due to stochastic encoding eventually outweighs the gain in rate due to the receiver's ability to decode the dummy message, as the transmissions need to be protected against a stronger cross-channel. Hence, for $\alpha \geq 1.5$, the achievable secrecy rate starts decreasing with $\alpha$.

\begin{figure}
	\begin{center}
		\setxysizeo
		\epsffile{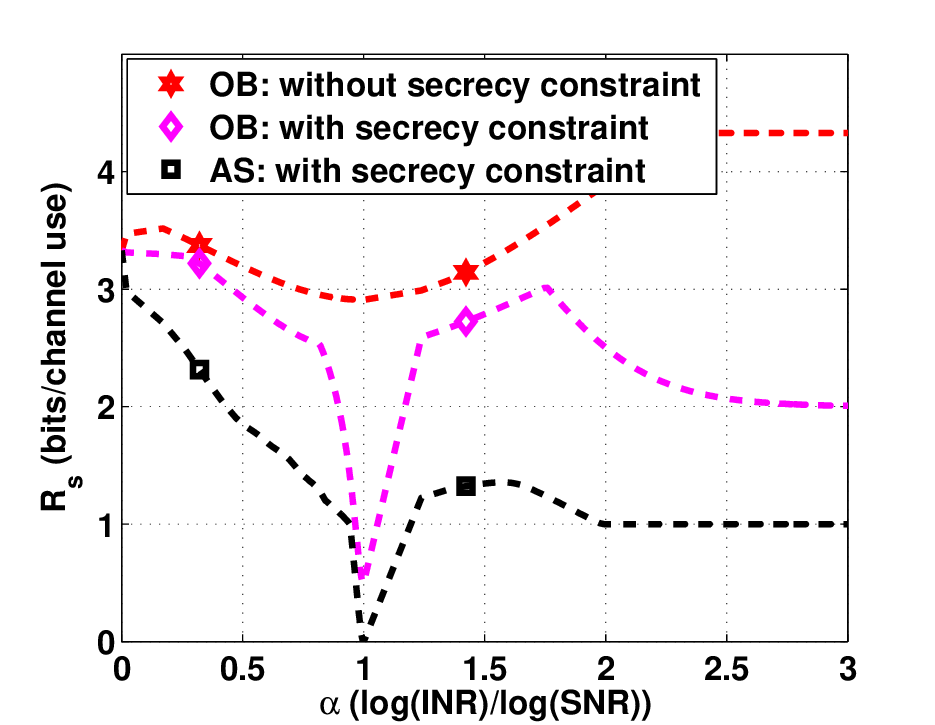}
		\vspace{-0.25cm}
		\caption{Secrecy rate in case of GSIC with $P=100$ and $C_G=1$. In the legend, OB and AS stand for the
			outer bound and achievable scheme, respectively.} \label{fig:GICresult3}
	\end{center}
\end{figure}

In Fig.~\ref{fig:GICresult3}, the achievable symmetric secrecy rate is plotted against $\alpha$ for $P=100$ and $C_G=1$, along with the outer bounds. For plotting the outer bound with secrecy constraints, the minimum of the outer bounds derived in this work and the outer bound in \cite{wang-TIT-2011} is used. When $\alpha > 1$, the achievable secrecy rate initially increases, and later decreases with $\alpha$. Finally, the achievable secrecy rate saturates when $(\alpha \geq 2)$, and this is due to the fact that it is no longer possible to transmit any non-cooperative private message and the gain in the achievable secrecy rate as compared to $C_G=0$ case is due to cooperation only. Hence, when $C_G > 0$, the proposed scheme achieves nonzero secrecy rate in all the interference regimes except for the $\alpha=1$ case. Hence, as the value of $C_G$ increases, it is required to assign lower powers for transmitting the dummy message and the non-cooperative private message. By assigning lower power to the
non-cooperative private message, the penalty in the achievable secrecy rate due to stochastic encoding also decreases. In the following example, no power is allocated for transmitting the non-cooperative private message and the dummy message.
\begin{figure}
	\begin{center}
		\setxysizeo
		\epsffile{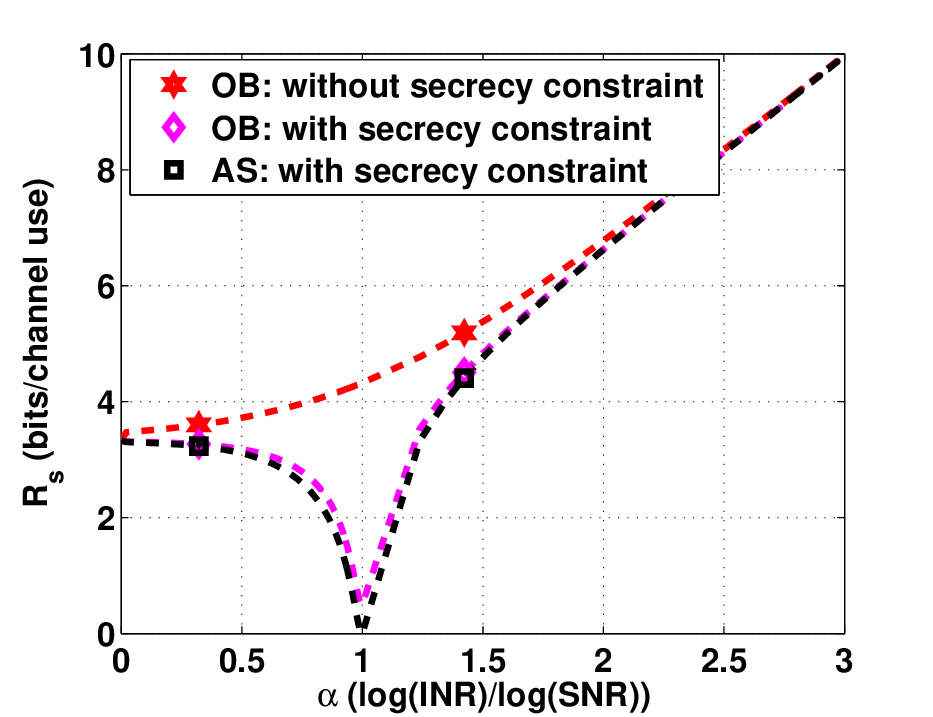}
		\vspace{-0.25cm}
		\caption{Secrecy rate in case of GSIC with $P=100$ and $C_G=10$. In the legend, OB stands for the outer bound and AS stands for the achievable scheme. In the later part of the very high interference regime, the achievable secrecy rate matches with the outer bound without secrecy constraint. Hence, there is no loss in the achievable rate due to the secrecy constraint at receiver.} \label{fig:GICresult4}
	\end{center}
\end{figure}

In Fig.~\ref{fig:GICresult4}, the achievable symmetric secrecy rate is plotted against $\alpha$ for $P=100$ and $C_G=10$, along with the outer bounds. Here, the achievable secrecy rate and outer bounds are very close to each other. In this case, both the users transmit cooperative private messages only.
\section{Conclusions}\label{sec:conc}
This work explored the role of limited-rate transmitter cooperation in facilitating secure communication over
$2$-user IC. For the deterministic case, the achievable scheme used a combination of interference cancelation, random bit transmission, relaying of the other user's data bits, and time sharing, depending on the values of $\alpha$ and $C$. Also, outer bounds on the secrecy rate were derived for the deterministic case. The novelty in the derivation of the outer bound lies in providing side information to receiver in a carefully chosen manner, use of the secrecy constraints at the receivers, and partitioning the encoded message/output, depending on the value of $\alpha$. The study of the deterministic model gave useful insights for obtaining achievable schemes and outer bounds on the secrecy rate for the Gaussian case. The achievable scheme used a combination of Marton's coding scheme and stochastic encoding along with dummy message transmission. It was found that, with limited-rate cooperation, it
is possible to achieve nonzero secrecy rate in almost all cases, except when $\alpha=1$.
A fundamental finding of this work is that a limited rate secure cooperative link between the transmitters of a $2$-user IC can greatly enhance the achievable rates for secure communication. 
Future work could investigate the value of other forms of limited rate cooperation between nodes. Another 
related problem could be the study of the case where the nodes cannot completely trust each other. The 
achievable schemes and outer bounds results proposed in this paper can give useful insight and facilitate 
further studies of the IC with secrecy constraints.
\appendix
\subsection{Proof of Theorem \ref{th:theoremSLDIC-outer1}}\label{sec:theoremSLDIC-outer1}
Using Fano's inequality, the rate of user $1$ is upper bounded as
\begin{align}
	NR_1 & \leq I(W_1; \ybold_1^N) + N\epsilon_1, \nonumber \\
	& = H(\ybold_1^N) - H(\ybold_1^N|W_1) + N\epsilon_1, \nonumber \\
	& \stackrel{(a)}{\leq} H(\ybold_1^N) - H(\ybold_1^N|W_1,\xbold_1^N) + N\epsilon_1, \nonumber \\
	& = H(\ybold_1^N) - H(\xbold_{2a}^N | W_1,\xbold_1^N) + N\epsilon_1, \text{ where, } \xbold_{2a}^N \triangleq \mathbf{D}^{q-n}\xbold_2^N, \nonumber \\ 
	& \leq H(\ybold_1^N) - H(\xbold_{2a}^N | W_1,\xbold_1^N,\coopsignalone^N,\coopsignaltwo^N) \nonumber + N\epsilon_1, \nonumber  \\
	& \stackrel{(b)}{=} H(\ybold_1^N) - H(\xbold_{2a}^N |\coopsignalone^N,\coopsignaltwo^N) + N\epsilon_1, \nonumber \\
	\text{or }  H(\xbold_{2a}^N |\coopsignalone^N,\coopsignaltwo^N) & \leq H(\ybold_1^N) - N R_1 + N\epsilon_1, \label{eq:outer-sldic1} 
\end{align}
where (a) follows by using the fact that the entropy cannot increase by additional conditioning and (b) follows by using the relation in (\ref{eq:thouter0}). 

Adopting  similar steps used to obtain (\ref{eq:outer-sldic1}), the following bound on the conditional entropy is obtained.
\begin{align}
	H(\xbold_{1a}^N|\coopsignalone^N,\coopsignaltwo^N) & \leq H(\ybold_2^N) - NR_2 + N\epsilon_2, \text{ where, } \xbold_{1a}^N \triangleq \mathbf{D}^{q-n}\xbold_1^N. \label{eq:oute-sldic2} 
\end{align}
The rate of user $1$ can also be bounded as
\begin{align}
	NR_1 &\leq I(W_1;\ybold_1^N) + N\epsilon_1, \nonumber \\
	& \stackrel{(a)}{\leq} I(W_1;\ybold_1^N,\ybold_2^N) + N\epsilon_1, \nonumber \\
	& \stackrel{(b)}{=} I(W_1;\ybold_1^N|\ybold_2^N) + N\epsilon_1, \nonumber \\
	&  \leq H(\ybold_1^N|\ybold_2^N) + N\epsilon_1, \nonumber \\
	&  \stackrel{(c)}{=} H(\ybold_1^N,\ybold_2^N) - H(\ybold_2^N) + N\epsilon_1, 
	\end{align}
	\begin{align}
	&  \stackrel{(d)}{\leq} H(\ybold_1^N,\ybold_2^N,\xbold_{1a}^N,\xbold_{2a}^N) - H(\ybold_2^N) + N\epsilon_1, \nonumber \\
	&  \stackrel{(e)}{=}  H(\xbold_{1a}^N,\xbold_{2a}^N) +  H(\ybold_1^N,\ybold_2^N|\xbold_{1a}^N,\xbold_{2a}^N)- H(\ybold_2^N) + N\epsilon_1, \nonumber \\
	&  \stackrel{(f)}{\leq} H(\xbold_{1a}^N,\xbold_{2a}^N, \coopsignalone^N,\coopsignaltwo^N) +  H(\ybold_1^N,\ybold_2^N|\xbold_{1a}^N,\xbold_{2a}^N)- H(\ybold_2^N)+ N\epsilon_1,\nonumber \\
	&   \leq H(\coopsignalone^N,\coopsignaltwo^N) + H(\xbold_{1a}^N | \coopsignalone^N,\coopsignaltwo^N) + H(\xbold_{2a}^N| \vonetwob^N,\vtwooneb^N) + H(\ybold_1^N,\ybold_2^N|\xbold_{1a}^N,\xbold_{2a}^N) - H(\ybold_2^N) + N\epsilon_1, \label{eq:oute-sldic3} 
\end{align}
where (a) is due to a genie providing $\ybold_2^N$ to receiver $1$; (b) is due to the perfect secrecy condition at receiver 2, i.e., $I(W_1;\ybold_2^N) = 0$; (c) is obtained from the joint entropy $H(\ybold_1^N,\ybold_2^N) = H(\ybold_2^N) + H(\ybold_1^N | \ybold_2^N)$; (d), (e) and (f): follows from the chain rule for joint entropy; (f) is obtained using chain rule for joint entropy; and (g) is obtained using the fact that removing conditioning cannot decrease the entropy. 

Using  (\ref{eq:outer-sldic1}) and (\ref{eq:oute-sldic2}), (\ref{eq:oute-sldic3}) becomes
\begin{align}
	NR_1 & \leq H(\coopsignalone^N,\coopsignaltwo^N) + H(\ybold_1^N) - N\lsqb R_1 + R_2\rsqb + H(\ybold_1^N,\ybold_2^N|\xbold_{1a}^N,\xbold_{2a}^N) + N\epsilon_1, \nonumber \\
	\text{or } N[2R_1 + R_2] & \leq H(\coopsignalone^N,\coopsignaltwo^N) + H(\ybold_1^N) + H(\mathbf{D}^{q-m}\xbold_1^N|\xbold_{1a}^N) +  H(\mathbf{D}^{q-m}\xbold_2^N|\xbold_{2a}^N) + N\epsilon_1.  \label{eq:outer-sldic4} 
\end{align}
The above equation is simplified under the following cases.\\
\textbf{Case 1 $(m \geq n)$}: In this case $q = m$ and (\ref{eq:outer-sldic4}) becomes
\begin{align}
	N[2R_1 + R_2] & \leq H(\coopsignalone^N,\coopsignaltwo^N) + H(\ybold_1^N) + H(\xbold_1^N|\xbold_{1a}^N) +  H(\xbold_2^N|\xbold_{2a}^N) + N\epsilon_1, \nonumber \\
	\text{or } R &\leq \frac{1}{3}\lsqb 2C + 3m-2n\rsqb. \label{eq:outer-sldic5} 
\end{align}
The above equation is obtained by using the fact that the entropy $H(\vonetwob,\vtwooneb)$, $H(\ybold_i)$ and  $H(\xbold_i|\xbold_{ia})$ are upper bounded by $2C$, $m$ and $m-n$, respectively.\\
\textbf{Case 2 $(m \leq n)$}: In this case $q = n$ and (\ref{eq:outer-sldic4}) becomes
\begin{align}
	N[2R_1 + R_2] & \leq H(\coopsignalone^N,\coopsignaltwo^N) + H(\ybold_1^N) + H(\mathbf{D}^{n-m}\xbold_1^N|\xbold_{1}^N) +  H(\mathbf{D}^{n-m}\xbold_2^N|\xbold_{2}^N) + N\epsilon_1, \nonumber \\
	\text{or } R &\leq \frac{1}{3}\lsqb 2C + n\rsqb. \label{eq:outer-sldic6} 
\end{align}
The above equation is obtained by using the fact that the entropy $H(\vonetwob,\vtwooneb)$ and $H(\ybold_1)$ are upper bounded by $2C$ and $n$, respectively. Also, given $\xbold_i^N$, there is no uncertainty about $\mathbf{D}^{n-m}\xbold_i^N$. Combining (\ref{eq:outer-sldic5}) and (\ref{eq:outer-sldic6}) results in (\ref{eq:outerone_deter1}). This completes the proof.
\subsection{Proof of Theorem \ref{th:theoremSLDIC-outer2}}\label{sec:theoremSLDIC-outer2}
Using Fano's inequality, the rate of user $1$ is upper bounded as
\begin{align}
	NR_1 & \leq I(W_1;\ybold_1^N) + N\epsilon_1, \nonumber \\
	& \stackrel{(a)}{\leq} I(W_1;\ybold_1^N,\ybold_{2a}^N) + N\epsilon_1, \text{ where }\ybold_{2a}^N \triangleq (\xbold_{1a}^N,\xbold_{1b}^N), \nonumber \\
	& = I(W_1;\ybold_{2a}^N) + I(W_1;\ybold_1^N|\ybold_{2a}^N)+ N\epsilon_1. \label{eq:outertwo-sldic1}
\end{align}
where (a) is due to a genie providing $\ybold_{2a}^N$ to receiver $1$. From the secrecy constraint at receiver $2$, following holds.
\begin{align}
	& I(W_1;\ybold_2^N) = 0, \nonumber \\
	\text{or }& I(W_1;\ybold_{2a}^N,\ybold_{2b}^N) = 0, \text{ where } \ybold_{2b}^N = \xbold_{2a}^N \oplus \xbold_{1c}^N, \nonumber \\
	\text{or }& I(W_1;\ybold_{2a}^N) + I(W_1;\ybold_{2b}^N|\ybold_{2a}^N) = 0. \label{eq:outertwo-sldic2}
\end{align}
As mutual information cannot be negative, $I(W_1;\ybold_{2a}^N) = 0$ and (\ref{eq:outertwo-sldic1}) becomes
\begin{align}
	NR_1 & \leq I(W_1;\ybold_1^N|\ybold_{2a}^N)+ N\epsilon_1, \nonumber \\
	& = H(\ybold_1^N|\ybold_{2a}^N) - H(\ybold_1^N|\ybold_{2a}^N, W_1) + N\epsilon_1, \nonumber \\
	& \stackrel{(a)}{=} H(\xbold_{2a}^N, \xbold_{2b}^N, \xbold_{1a}^N \oplus \xbold_{2c}^N|\xbold_{1a}^N, \xbold_{1b}^N)- H(\xbold_{2a}^N, \xbold_{2b}^N, \xbold_{1a}^N \oplus \xbold_{2c}^N|\xbold_{1a}^N, \xbold_{1b}^N,W_1) + N\epsilon_1, \nonumber \\
	& = H(\xbold_2^N|\xbold_{1a}^N, \xbold_{1b}^N) - H(\xbold_2^N|\xbold_{1a}^N, \xbold_{1b}^N,W_1) + N\epsilon_1, \nonumber \\
	& \leq H(\coopsignalone^N,\coopsignaltwo^N,\xbold_2^N|\xbold_{1a}^N, \xbold_{1b}^N) - H(\xbold_2^N|\xbold_{1a}^N, \xbold_{1b}^N,W_1) + N\epsilon_1, \nonumber \\
	& \stackrel{(b)}{\leq} H(\coopsignalone^N,\coopsignaltwo^N|\xbold_{1a}^N, \xbold_{1b}^N) + H(\xbold_2^N|\coopsignalone^N,\coopsignaltwo^N,\xbold_{1a}^N, \xbold_{1b}^N) - H(\xbold_2^N|\coopsignalone^N,\coopsignaltwo^N,\xbold_{1a}^N, \xbold_{1b}^N,W_1) + N\epsilon_1, \nonumber \\
	& \stackrel{(c)}{\leq} H(\vonetwob^N,\vtwooneb^N) + H(\xbold_2^N|\coopsignalone^N,\coopsignaltwo^N,\xbold_{1a}^N, \xbold_{1b}^N)- H(\xbold_2^N|\coopsignalone^N,\coopsignaltwo^N,\xbold_{1a}^N, \xbold_{1b}^N,W_1) + N\epsilon_1, \nonumber \\
	& \stackrel{(d)}{\leq} H(\coopsignalone^N,\coopsignaltwo^N) + N\epsilon_1, \nonumber \\
	\text{or } R_1 & \leq 2C.\label{eq:outertwo-sldic3}
\end{align}
where (a) is obtained by splitting the message into three parts as shown in Fig. \ref{fig:veryhighouter}; (b) is due to the fact that conditioning cannot increase the entropy; (c) is due to the fact that removing conditioning cannot decrease the entropy; and (d) follows by using the relationship in (\ref{eq:thouter0}). This completes the proof. 
\subsection{Proof of Theorem \ref{th:theoremSLDIC-outer3}}\label{sec:theoremSLDIC-outer3}
Using Fano's inequality, rate of user $1$ is bounded as
\begin{align}
& NR_1 \nonumber \\
& \leq I(W_1;\ybold_1^N) + N\epsilon_1, \nonumber \\
& \stackrel{(a)}{\leq}  I(W_1;\ybold_1^N, \ybold_{2a}^N) + N\epsilon_1, \nonumber 
\\
& \stackrel{(b)}{=} I(W_1;\ybold_1^N|\ybold_{2a}^N) + N\epsilon_1, \nonumber \\
& = H(\ybold_1^N|\xbold_{1a}^N) - H(\ybold_1^N|\xbold_{1a}^N,W_1) + N\epsilon_1, \nonumber \\
& \leq H(\ybold_1^N,\coopsignalone^N,\coopsignaltwo^N|\xbold_{1a}^N) - H(\ybold_1^N|\xbold_{1a}^N,W_1) + N\epsilon_1, \nonumber \\
& \leq H(\coopsignalone^N,\coopsignaltwo^N|\xbold_{1a}^N) + H(\ybold_1^N|\coopsignalone^N,\coopsignaltwo^N,\xbold_{1a}^N) - H(\ybold_1^N|\coopsignalone^N,\coopsignaltwo^N,\xbold_{1a}^N,W_1) + N\epsilon_1, \nonumber\\
& \leq H(\coopsignalone^N,\coopsignaltwo^N) + H(\ybold_1^N|\coopsignalone^N,\coopsignaltwo^N,\xbold_{1a}^N) - H(\ybold_1^N|\coopsignalone^N,\coopsignaltwo^N,\xbold_{1a}^N,W_1) + N\epsilon_1, \nonumber \\
& \stackrel{(c)}{=} H(\coopsignalone^N,\coopsignaltwo^N) + H(\xbold_{2a}^N,\ybold_{1b}^N|\coopsignalone^N,\coopsignaltwo^N,\xbold_{1a}^N)  - H(\xbold_{2a}^N,\ybold_{1b}^N|\coopsignalone^N,\coopsignaltwo^N,\xbold_{1a}^N,W_1) + N\epsilon_1, \nonumber \\
& = H(\coopsignalone^N,\coopsignaltwo^N) + H(\xbold_{2a}^N|\coopsignalone^N,\coopsignaltwo^N,\xbold_{1a}^N)  \!+\! H(\ybold_{1b}^N|\coopsignalone^N,\coopsignaltwo^N,\xbold_{1a}^N,\xbold_{2a}^N) \!-\! H(\xbold_{2a}^N|\coopsignalone^N,\coopsignaltwo^N,\xbold_{1a}^N,W_1) \nonumber \\
&  \qquad \! -\! H(\ybold_{1b}^N|\coopsignalone^N,\coopsignaltwo^N,\xbold_{1a}^N,\xbold_{2a}^N,W_1) \!+\! N\epsilon_1, \nonumber \\
& \stackrel{(d)}{=} H(\coopsignalone^N,\coopsignaltwo^N) +  H(\ybold_{1b}^N|\coopsignalone^N,\coopsignaltwo^N,\xbold_{1a}^N,\xbold_{2a}^N)- H(\ybold_{1b}^N|\coopsignalone^N,\coopsignaltwo^N,\xbold_{1a}^N,\xbold_{2a}^N,W_1)+ N\epsilon_1, \label{eq:outerthree-sldic1}
\end{align}
where (a) is due to a genie providing $\ybold_{2a}^N$ to receiver~$1$; (b) is obtained using the secrecy constraint at receiver~$2$; (c) is obtained by partitioning of the encoded message and output as shown in Fig.~\ref{fig:highouter}; and (d) is obtained using the relation in~(\ref{eq:thouter0}).

Once again, as the encoded messages at transmitters are correlated, it is not straightforward to bound or simplify the entropy terms in (\ref{eq:outerthree-sldic1}). To overcome this problem, the output $\ybold_{1b}$ is partitioned  into two parts as follows:
\begin{itemize}
 \item $\ybold_{1b}^{(1)}$: contains $\xbold_{1a}$ sent by transmitter~$1$ and the interference caused by transmitter~$2$ due to transmission on the levels $[2m-n+1:m]$
 \item $\ybold_{1b}^{(2)}$: contains $\xbold_{1b}$ sent by transmitter~$1$ and the interference caused by transmitter~$2$ due to transmission on the levels $[1:2m-n]$
\end{itemize}
The partitioning of $\ybold_{1b} = (\ybold_{1b}^{(1)}, \ybold_{1b}^{(2)})$ is illustrated in the Fig.~\ref{fig:highouter}. Now consider the second term in (\ref{eq:outerthree-sldic1}):
\begin{align}
H(\ybold_{1b}^N|\coopsignalone^N,\coopsignaltwo^N,\xbold_{1a}^N,\xbold_{2a}^N) & = H(\ybold_{1b}^{(1)N},\ybold_{1b}^{(2)N}|\coopsignalone^N,\coopsignaltwo^N,\xbold_{1a}^N,\xbold_{2a}^N), \nonumber \\
& = H(\xbold_{2b}^N, \xbold_{2c}^{(1)N}|\coopsignalone^N,\coopsignaltwo^N,\xbold_{1a}^N,\xbold_{2a}^N) + H(\ybold_{1b}^{(2)N}|\coopsignalone^N,\coopsignaltwo^N,\xbold_{1a}^N,\xbold_{2a}^N,\ybold_{1b}^{(1)N}), \nonumber \\
& = H(\xbold_{2b}^N, \xbold_{2c}^{(1)N}|\coopsignalone^N,\coopsignaltwo^N,\xbold_{2a}^N)   + H(\ybold_{1b}^{(2)N}|\coopsignalone^N,\coopsignaltwo^N,\xbold_{1a}^N,\xbold_{2a}^N,\ybold_{1b}^{(1)N}), \label{eq:outerthree-sldic2}
\end{align}
where $\xbold_{ic}^{(1)}$ and $\xbold_{ic}^{(2)}$ correspond to the bits
transmitted on the levels $[\min(n-m, 2m-n)+1: n-m - \min(n-m, 2m-n) + 1]$ and $[1:\min(n-m,
2m-n)]$ of transmitter~$i$, respectively.

The above equation is obtained using the fact that $I(\xbold_{2b}^N, \xbold_{2c}^{(1)N};\xbold_{1a}^N|\coopsignalone^N,\coopsignaltwo^N,\xbold_{2a}^N) = 0$. This can be obtained using the relation in (\ref{eq:thouter0}).
In a similar way, the third term in (\ref{eq:outerthree-sldic1}) can be simplified as follows:
\begin{align}
& H(\ybold_{1b}^N|\coopsignalone^N,\coopsignaltwo^N,\xbold_{1a}^N,\xbold_{2a}^N,W_1) = H(\xbold_{2b}^N, \xbold_{2c}^{(1)N}|\coopsignalone^N,\coopsignaltwo^N,\xbold_{2a}^N) + H(\ybold_{1b}^{(2)N}|\coopsignalone^N,\coopsignaltwo^N,\xbold_{1a}^N,\xbold_{2a}^N,\ybold_{1b}^{(1)N},W_1). \label{eq:outerthree-sldic3}
\end{align}
From (\ref{eq:outerthree-sldic2}) and (\ref{eq:outerthree-sldic3}), and dropping the last term in (\ref{eq:outerthree-sldic3}), (\ref{eq:outerthree-sldic1}) becomes
\begin{align}
NR_1 & \leq H(\coopsignalone^N,\coopsignaltwo^N) + H(\ybold_{1b}^{(2)N}|\coopsignalone^N,\coopsignaltwo^N,\xbold_{1a}^N,\xbold_{2a}^N,\ybold_{1b}^{(1)N}) + N\epsilon_1, \nonumber \\
\text{or }R_1 & \leq H(\coopsignalone,\coopsignaltwo) + H(\ybold_{1b}^{(2)}) \leq 2C + 2m-n. \label{eq:outerthree-sldic4}
\end{align}
In the above equation, the term $H(\coopsignalone,\coopsignaltwo)$ is upper bounded by $2C$. From the definition of
$\ybold_{1b}^{(2)}$, it can be seen that the term $H(\ybold_{1b}^{(2)})$ can be upper bounded by $2m-n$. This completes the proof.
\subsection{Proof of Theorem \ref{th:theoremSLDIC-outer4}}\label{sec:theoremSLDIC-outer4}
Using Fano's inequality, the rate of user-$1$ is upper bounded as
\begin{align}
	NR_1 & \leq I(W_1;\ybold_1^N) + N\epsilon_1 \stackrel{(a)}{=} I(W_1;\ybold_2^N) + N\epsilon_1, \nonumber \\
	\text{or } R_1 & \stackrel{(b)}{=} 0, \label{eq:snreqinrone}
\end{align}
where (a) is obtained using the fact that $\ybold_1 = \ybold_2$ and (b) is obtained using the perfect secrecy condition. This completes the proof.
\subsection{Details of the achievable scheme for the SLDIC when $(1 < \alpha <2)$}\label{sec:appen-ach-high-sldic}
\subsubsection{When $(1 <\alpha \leq 1.5)$}
The achievable scheme uses transmission of random bits, interference cancelation, or both, depending on the capacity of the cooperative link. The bits received through the cooperative links are transmitted on the levels $[1:C]$. As the $n-m$ lower levels are not present to the intended receiver, these levels can be used for relaying other user's data bits. Any data bits transmitted on the levels higher than $n-m$ will cause interference. The cooperative data bits transmitted by transmitter $i$ on levels higher than $n-m$ $[n-m+1:C]$ are canceled by transmitter $j$ by transmitting the same data bits along with the data bits of user $i$ $(i \neq j)$. 

In the remaining higher levels, it is possible to transmit data bits securely with the help of transmission of random bits. The transmission of random bits is similar to that in case of moderate interference regime.
The message of transmitter $1$ is encoded as follows
\begin{itemize}
	\item Case 1 $((t-r_2)^+=0))$:
	\begin{align}
		\mathbf{x}_{1} =  \lsqb\begin{array}{l}
			\mathbf{a}_{p \times 1}^e \\
			\mathbf{0}_{s \times 1}
		\end{array} \rsqb \oplus
		\lsqb\begin{array}{l}
			\mathbf{0}_{(n-C) \times 1}^e \\
			\mathbf{b}_{C \times 1}
		\end{array} \rsqb \oplus
		\lsqb\begin{array}{l}
			\mathbf{0}_{v \times 1}^e \\
			\mathbf{a}_{(C-r2)^+ \times 1}^l
		\end{array} \rsqb, \label{eq:justmod1}
	\end{align}
	where $\mathbf{a}^{e} \triangleq [\mathbf{d}_1, \mathbf{u}_2, \mathbf{z}_3, \mathbf{d}_4, \mathbf{u}_5, \mathbf{z}_6, \ldots, \mathbf{d}_{3B-2}, \mathbf{u}_{3B-1}, \mathbf{z}_{3B},]^{T}$, $\mathbf{u}_l \triangleq [ a_{n-(l-1)r_2}, \\ a_{n-(l-1)r_2-1},  \ldots, a_{n-lr_2+1}]$, $\mathbf{d}_l \triangleq \lsqb d_{n-(l-1)r_2}, d_{n-(l-1)r_2-1}, \ldots, d_{n-lr_2+1} \rsqb$, $\mathbf{z}_l$ is a zero vector of size $1 \times r_2$, $\mathbf{b} \triangleq \lsqb b_C, b_{C-1}, \ldots, b_1\rsqb^T$, $\mathbf{a}^l = \lsqb a_C, a_{C-1}, \ldots,a_{r_2+1}\rsqb^T$ $p \triangleq 3Br_2$, $s \triangleq n-p$ and $v \triangleq (n-(C-r_2)^+)$.
	
	\item Case 2 $((t-r_2)^+ \neq 0)$:
	\begin{align}
		\mathbf{x}_{1}^{\text{mod}} =  \mathbf{x}_{1} \oplus \lsqb\begin{array}{l}
			\mathbf{0}_{w \times 1} \\
			\mathbf{a'}_{t \times 1} \\
			\mathbf{0}_{s \times 1}
		\end{array} \rsqb, \label{eq:justmod2}
	\end{align}
	where $\mathbf{x}_1$ is as defined in (\ref{eq:justmod1}), $\mathbf{a'} \triangleq \lsqb \mathbf{d}_{11},  \mathbf{d}_{12}, \mathbf{u}_{11}, \mathbf{u}_{12}, \mathbf{z'}\rsqb^{T}$, $\mathbf{u}_{11} = \lsqb a_{n-3Br_2-q-v'}, \right. \\ \left.a_{n-3Br_2-q-v'-1}, \ldots, a_{n-3Br_2-2q-v'+1}\rsqb$, $\mathbf{d}_{11} \triangleq \lsqb d_{n-3Br_2}, d_{n-3Br_2-1}, \ldots, d_{n-3Br_2-q+1}\rsqb$. Also, $\mathbf{d}_{12}$ and $\mathbf{u}_{12}$  are zero vectors of size $1 \times v'$ and $\mathbf{z'}$ is a zero vector of size $1 \times f$. Here, $v' \triangleq (n-m-q)^+$, $f \triangleq (t-2(q+v'))^{+}$, $s \triangleq  n-m + C$ and $w \triangleq (n-t-s)^+$.
\end{itemize}
The proposed encoding scheme achieves the following symmetric secrecy rate:
\begin{align}
	R_s  = B(n-m)+ C + q. \label{eq:justmod3}
\end{align}
\subsubsection{When $(1.5 <\alpha < 2)$}
The links in the SLDIC can be classified into three categories: Type I, Type II, and Type III, as shown in Fig. \ref{fig:highsec1}. The classification is based on whether the data bits are received cleanly or with interference at the intended receiver, and whether or not they are present to the intended receiver. \\
\textbf{Case $1$} \textit{When $(0 \leq C \leq 4n-6m)$:} In this case, the achievable scheme uses a combination of interference cancelation, transmission of random bits and relaying of other user's data bits. The data bits transmitted by transmitter $i$ on the levels associated with Type II links $[n-m+1:m]$ will be received at the unintended receiver $j$ $(j \neq i)$. In order to ensure secrecy, transmitter $j$ transmits random bits on the levels $[2(n-m)+1:n]$. The remaining levels can be used for transmitting other user's data bits received through cooperation. The cooperative bits are transmitted on the levels corresponding to Type I and Type III links. The $C_1 = \lfloor \frac{C}{2}\rfloor$ data bits transmitted by transmitter $i$ for transmitter $j$ on the levels corresponding to Type III links will not be received at the receiver $i$ and hence, will remain secure. The data bits transmitted on the levels corresponding to Type I links by transmitter $i$ for transmitter $j$ will cause interference at 
receiver $i$. The interference is required to be canceled by transmitter $j$ by transmitting the same bits. 

The achievable scheme is shown for $m=5$ and $n=8$ for $C=2$ in Fig. \ref{fig:highsec2}. The message of transmitter $1$ is encoded as follows
\begin{align}
	\mathbf{x}_{1} =  \lsqb\begin{array}{l}
		\mathbf{d}_{l \times 1} \\ \\ \\ \\
		\mathbf{0}_{(n-l) \times 1}
	\end{array} \rsqb \oplus
	\lsqb\begin{array}{l}
		\mathbf{0}_{(n-m) \times 1} \\ \\
		\mathbf{a}_{l \times 1} \\ \\
		\mathbf{0}_{(n-m) \times 1}
	\end{array} \rsqb \oplus
	\lsqb\begin{array}{l}
		\mathbf{0}_{(n-m-C_2) \times 1} \\ 
		\mathbf{b}_{C_2 \times 1}^u \\
		\mathbf{0}_{l \times 1} \\
		\mathbf{b}_{C_1 \times 1}^{l} \\
		\mathbf{0}_{(n-m-C_1) \times 1}
	\end{array} \rsqb \oplus
	\lsqb\begin{array}{l}
		\mathbf{0}_{m \times 1} \\ \\ 
		\mathbf{a'}_{C_2 \times 1} \\ \\
		\mathbf{0}_{(n-m-C_2) \times 1}
	\end{array} \rsqb, \label{eq:justmodpart1}
\end{align}
where  $\mathbf{d} \triangleq [d_{n}, d_{n-1}, \ldots, d_{n-l+1}]^T$, $\mathbf{a} \triangleq [a_{l}, a_{l-1}, \ldots, a_{1}]^T$,  $\mathbf{b}^l \triangleq [b_{n-m}, a_{n-m-1}, \ldots, b_{n-m-C_1+1}]^T$, $\mathbf{b}^u \triangleq [b_{m+C_2}, b_{m+C_2-1}, \ldots, b_{m+1}]^T$,  $\mathbf{a'} \triangleq [a_{m+C_2}, a_{m+C_2-1}, \ldots, a_{m+1}]^T$ and $l\triangleq 2m-n$. The proposed scheme achieves the following secrecy rate
\begin{align}
	R_s = 2m-n + C.  \label{eq:justmodpart2}
\end{align}
\textbf{Case $2$} \textit{When $(4n-6m < C \leq n)$:} In this case, $4n-6m$ cooperative data bits out of $C$ cooperative bits obtained through cooperation are used in a similar way as described in the previous case. Define $C_1 \triangleq C_2 \triangleq 2n-3m$. The remaining cooperative bits $C' \triangleq C - (4n-6m)$ is used as explained below. Let $C'' \triangleq \mylceil \frac{C'}{3}\myrceil$. The number of data bits that can be relayed by transmitter $i$ for transmitter $j$ on the lower levels corresponding to Type III links is $C_{T_3} = \min\lcb 2m-n,C''\rcb$. The remaining cooperative bits $C_{\text{rem}} = (C' - C_{T_3})^+$ are transmitted on the levels corresponding to Type I and II links. The $C_{T_1} \triangleq  \min\lcb \mylceil \frac{C_{\text{rem}}}{2}\myrceil,2m-n\rcb$ cooperative bits sent by transmitter $i$ causes interference at receiver $i$. These bits are canceled by transmitter $j$ by sending the same $C_{T_1}$ bits. The remaining $C_{T_2} = \min\lcb 2m-n, (C_{\text{rem}} - C_{T_1})^{+}\rcb$ bits are transmitted on the Type II links by transmitter $i$. These data bits cause interference at receiver $i$, which is canceled by transmitter $j$. The number of data bits that can be sent on the Type II links with the help of transmission of random bits is $r_d = \min\lcb (2m-n-C_{T_3})^{+}, 2m-n-C_{T_2}\rcb$. The achievable scheme is shown for $m=5$, $n=8$ and $C=4$ in Fig. \ref{fig:highsec4}. The message of transmitter $1$ is encoded as follows.
\begin{align}
	\mathbf{x}_{1} & =  \lsqb\begin{array}{l}
		\mathbf{0}_{(2m-n)\times 1} \\ 
		\mathbf{b}_{(2n-3m) \times 1}^{u} \\
		\mathbf{0}_{(2m-n) \times 1} \\
		\mathbf{b}_{(2n-3m) \times 1}^{l} \\
		\mathbf{0}_{(2m-n)\times 1} 
	\end{array} \rsqb \oplus
	\lsqb\begin{array}{l}
		\mathbf{0}_{(n-m-C_{T_1})^+ \times 1} \\ 
		\mathbf{b'}_{C_{T_1} \times 1}^u \\ \\
		\mathbf{0}_{(m-C_{T_3}) \times 1} \\
		\mathbf{b'}_{C_{T_3} \times 1}^l \\
	\end{array} \rsqb \oplus
	\lsqb\begin{array}{l}
		\mathbf{0}_{(m-C_{T_2}) \times 1} \\ \\ 
		\mathbf{b}_{C_{T_2} \times 1}^m \\ \\
		\mathbf{0}_{(n-m) \times 1}
	\end{array} \rsqb \oplus 
	\lsqb\begin{array}{l}
		\mathbf{0}_{m \times 1} \\ \\
		\mathbf{a}_{(2n-3m) \times 1}^u \\ 
		\mathbf{0}_{(2m-n-C_{T_2}) \times 1} \\ 
		\mathbf{a}_{C_{T_2} \times 1}^{m} 
	\end{array} \rsqb \oplus \nonumber \\
	& \qquad \qquad \qquad \qquad  \lsqb\begin{array}{l}
		\mathbf{0}_{(m-C_{T_1}) \times 1} \\ \\
		\mathbf{a'}_{C_{T_1} \times 1}^u \\ 
		\mathbf{0}_{(n-m) \times 1}^{m} 
	\end{array} \rsqb \oplus
	\lsqb\begin{array}{l}
		\mathbf{d}_{r_d \times 1} \\ 
		\mathbf{0}_{(n-m-r_d) \times 1}\\
		\mathbf{a}_{r_d \times 1}^e \\ 
		\mathbf{0}_{(m-r_d) \times 1} 
	\end{array} \rsqb, \label{eq:justmodpart3}
\end{align}
where $\mathbf{b}^l \triangleq [b_{n-m}, d_{n-m-1}, \ldots, b_{2m-n+1}]^T$, $\mathbf{b}^u \triangleq [b_{2(n-m)}, d_{2(n-m)-1}, \ldots, b_{m+1}]^T$, $\mathbf{b'}^l \triangleq [b_{C_{T_3}},\\ b_{C_{T_3}-1}, \ldots, b_{1}]^T$, $\mathbf{b'}^u \triangleq [b_{2(n-m) + C_{T_1}}, b_{2(n-m) + C_{T_1}-1}, \ldots, b_{2(n-m)+1}]^T$, $\mathbf{b'}^m \triangleq [b_{n-m + C_{T_2}}, b_{n-m + C_{T_2}-1}, \\ \ldots, b_{n-m+1}]^T$, $\mathbf{a}^u \triangleq [a_{2(n-m)}, a_{2(n-m)-1}, \ldots,  a_{m + 1}]^T$, $\mathbf{a}^m \triangleq [a_{n-m + C_{T_2}}, a_{n-m + C_{T_2}-1}, \ldots,  a_{n-m + 1}]^T$, $\mathbf{a'}^u \triangleq [a_{2(n-m) + C_{T_1}}, a_{2(n-m) + C_{T_1}-1}, \ldots,  a_{2(n-m) + 1}]^T$, $\mathbf{d} \triangleq [d_n, d_{n-1}, \ldots,  d_{n-r_d + 1}]^T$ and $\mathbf{a}^e \triangleq [a_{2m-n}, a_{2m-n-1}, \ldots,  a_{2m-n-r_d + 1}]^T$. The proposed scheme achieves the following secrecy rate
\begin{align}
	R_S = 4n-6m + C_{T_1} + C_{T_2} + C_{T_3} + r_d. \label{eq:justmodpart4}
\end{align}
\begin{figure}
	\centering
	\mbox{\subfigure[][]{\includegraphics[width=2in,height=2in]{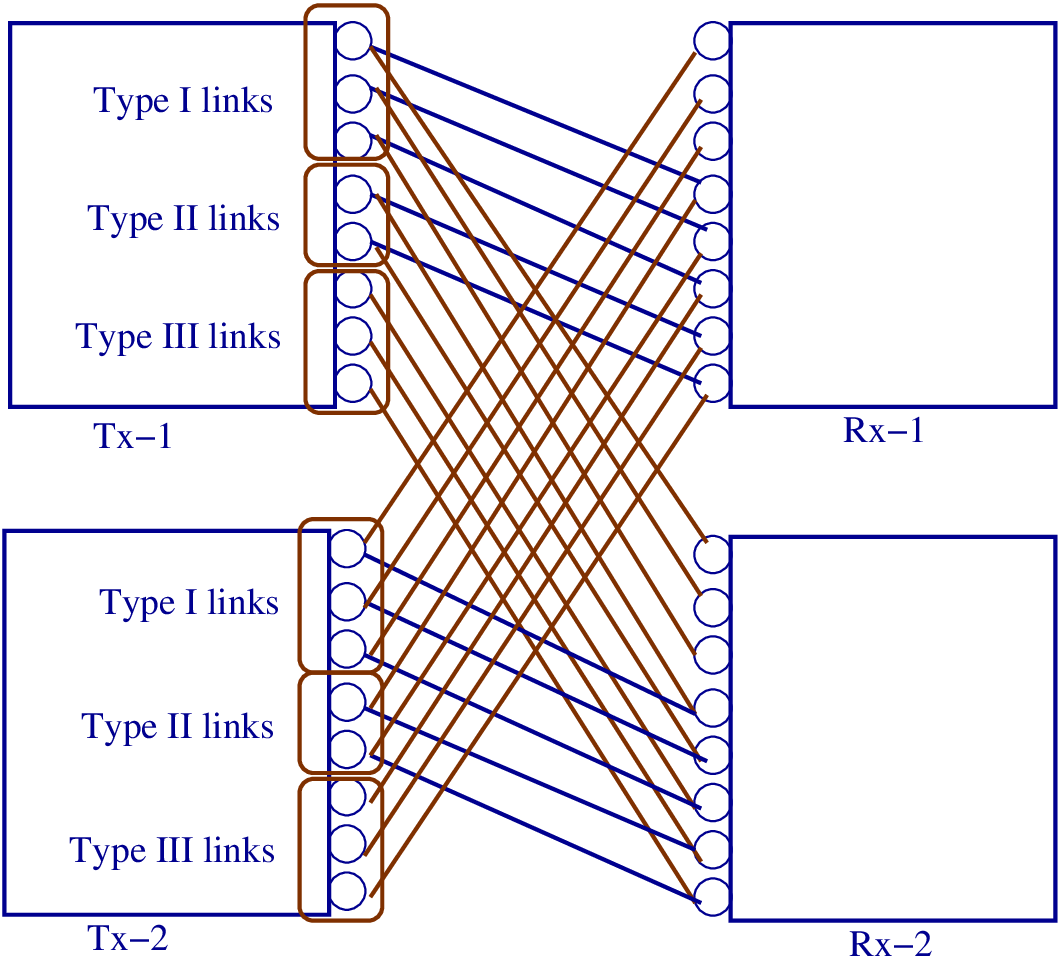}\label{fig:highsec1}} \quad
		\subfigure[][]{\includegraphics[width=2in,height=2in]{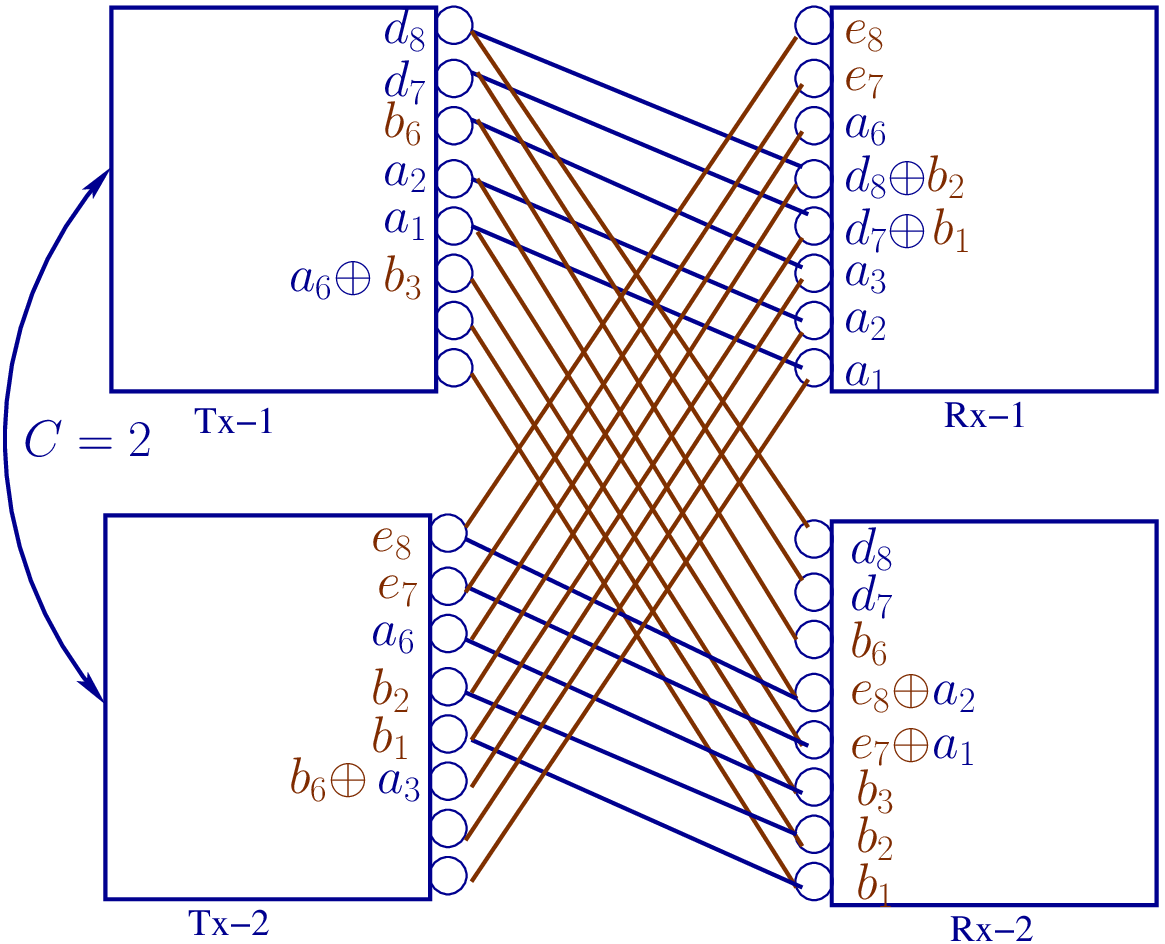}\label{fig:highsec2}} \quad
		\subfigure[][]{\includegraphics[width=2in,height=2in]{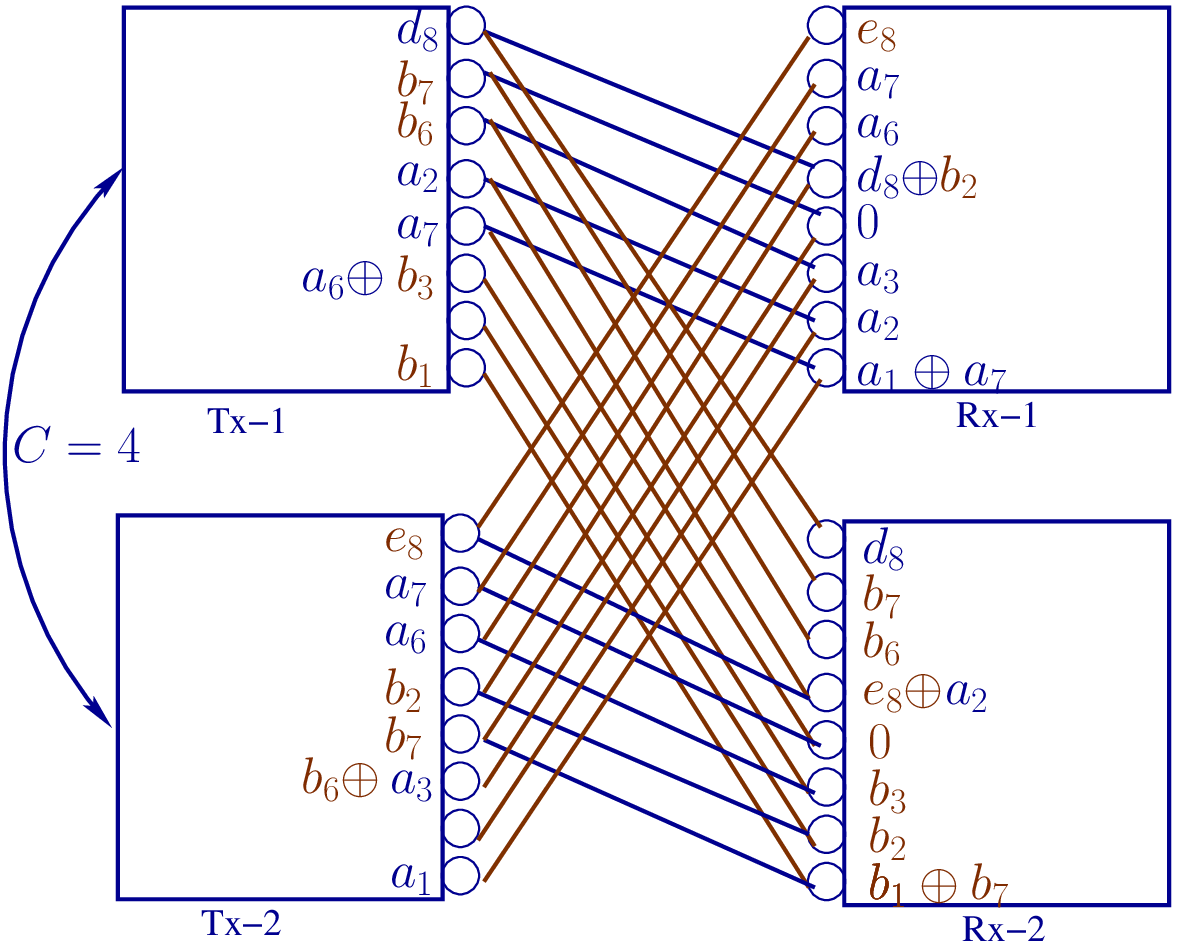}\label{fig:highsec3}}} \\
	\caption[]{SLDIC with $m=5$ and $n=8$: \subref{fig:highsec1} Different types of links, \subref{fig:highsec2} $C=2$, $R_S=4$ \subref{fig:highsec3} $C=4$, $R_S=5$.}\label{fig:highsec4}
	\pullUp
\end{figure}
\subsection{Details of the achievable scheme for the SLDIC when $(\alpha \geq 2)$}\label{sec:appen-ach-veryhigh-sldic}
\subsubsection{When $0 < C \leq \frac{m}{2}$ and $m$ is even}\label{sec:veryhigh1} In this case, interestingly, transmitters
share only random bits through the cooperative links. Each transmitter generates $C$ random bits independent
of data bits with $\operatorname{Bern} \left({\frac{1}{2}}\right)$. The achievable scheme involves transmitting
the data bits xored with the random bits. The same random bits are transmitted by the other transmitter, so as
to cancel them out at the desired receiver. In contrast to the achievable schemes in Secs. \ref{sec:SLDIC-ach-weak-mod}
and \ref{sec:SLDIC-ach-highint}, the random bits transmission causes jamming to the unintended receiver only.
Through careful observation it is found that sharing random bits through the cooperative links can achieve higher
secrecy rate compared to sharing data bits only.

In this case, the signal of transmitter $1$ is encoded as follows:
\begin{align}
\mathbf{x}_{1} & =  \lsqb\begin{array}{l}
	           \mathbf{0}_{(m-2C)^+ \times 1} \\
		    \mathbf{a}_{2C \times 1} \\
                   \mathbf{0}_{(n-m) \times 1}
	           \end{array} \rsqb \oplus
		 \lsqb\begin{array}{l}
	           \mathbf{0}_{(n-2C) \times 1} \\ \\		
                   \mathbf{d}_{2C \times 1}^1
	           \end{array} \rsqb  \oplus
		\lsqb\begin{array}{l}
	           \mathbf{0}_{(m-2C)^+ \times 1} \\ 	
                   \mathbf{d}_{2C \times 1}^2 \\
		   \mathbf{0}_{(n-m) \times 1}
	           \end{array} \rsqb, \label{eq:sldic-veryhigh1}
\end{align}
where $\mathbf{a} \triangleq [a_{2C}, a_{2C-1}, \ldots, a_1]^T$, $\mathbf{d}^1 \triangleq [e_C,d_C, \ldots, e_1,d_1]^T$
and $\mathbf{d}^2 \triangleq [d_C,e_C, \ldots, d_1,e_1]^T$. 

The proposed scheme achieves the following secrecy rate:
\begin{align}
R_s = 2C. \label{eq:sldic-veryhigh2}
\end{align}
Note that, with data bits sharing, the achievable scheme achieves
\begin{align}
R_s = C. \label{eq:sldic-veryhigh3}
\end{align}
Hence, under the proposed scheme, one can achieve higher rate by sharing random bits than by
sharing the data bits.
\subsubsection{When $(\frac{m}{2} < C \leq n- \frac{3m}{2})$ and $m$ is even}\label{sec:sldic-veryhigh2} In this case, the
transmitters exchange $\frac{m}{2}$ random bits and $(C- \frac{m}{2})$ data bits. The random bits are used
in an analogous fashion as described in the previous subsection. The links corresponding to the levels from
$[m+1:n-m]$ are present only at the unintended receiver and data bits transmitted on these levels are
received without interference at the unintended receiver. Hence, any data bits of the other user relayed
using these levels will remain secure. In this case, the signal of transmitter~$1$ is encoded as follows:
\begin{align}
\mathbf{x}_{1} & =  \lsqb\begin{array}{l}
	            \mathbf{a}_{m \times 1} \\ \\
                   \mathbf{0}_{(n-m) \times 1}
	           \end{array} \rsqb \oplus
		 \lsqb\begin{array}{l}
	           \mathbf{0}_{(n-m) \times 1} \\ \\		
                   \mathbf{d}_{m \times 1}^1
	           \end{array} \rsqb \oplus
		\lsqb\begin{array}{l}
	           \mathbf{d}_{m \times 1}^2 \\ \\
		   \mathbf{0}_{(n-m) \times 1}
	           \end{array} \rsqb  \oplus
	           \lsqb\begin{array}{l}
	           \mathbf{0}_{(n-C-\frac{m}{2}) \times 1} \\ 		
                   \mathbf{b}_{(C - \frac{m}{2}) \times 1}^c \\
                   \mathbf{0}_{m \times 1}
	           \end{array} \rsqb, \label{eq:veryhigh3a}
\end{align}
where  $\mathbf{a} \triangleq [a_{m}, a_{m-1}, \ldots, a_1]^T, \mathbf{d}^1 \triangleq [e_{\frac{m}{2}},d_{\frac{m}{2}}, \ldots, e_1,d_1]^T$,
$\mathbf{d}^2 \triangleq [d_{\frac{m}{2}},e_{\frac{m}{2}}, \ldots, d_1,e_1]^T$ and
$\mathbf{b}^c \triangleq [ b_{\frac{m}{2}+C}, b_{\frac{m}{2}+C-1}, \ldots, b_{m+1}]^T$.

The proposed scheme achieves the following secrecy rate:
\begin{align}
R_s  = \frac{m}{2} + C. \label{eq:veryhigh3b}
\end{align}
\subsubsection{When $(n-\frac{3m}{2} < C < n-\frac{m}{2})$ and $m$ is even}\label{sec:sldic-veryhigh3} The novelty of the
proposed scheme is in precoding the data bits of the user partly with the other user's data bits and/or with
random bits. The random bits used for precoding may be generated at its own transmitter or obtained from
the other transmitter through the cooperative link. Then, by appropriately transmitting data bits or random
bits on the levels of the SLDIC, the random bits are canceled at the intended receiver, or the data bits of the
other user are canceled out at the unintended receiver. The details of the achievable scheme is as follows.

The achievable scheme uses transmission of random bits, interference cancelation, time sharing and relaying
of the other user's data bits. The transmitters share a combination of random bits and data bits through the
cooperative links. To simplify the understanding of the achievable scheme, first consider the $\alpha=2$ case.
In this case, both the transmitters share $\frac{m}{2}$ random bits along with $C_1 \triangleq C - \frac{m}{2}$
data bits. In the first time slot, transmitter~$1$ sends $m$ random bits ($d_i$ and $e_i$) on alternate levels
in $[1:m]$. In order to eliminate the interference caused by these random bits at receiver~$2$, the data bits of
transmitter~$2$ are precoded (xored) with these $m$ random bits and transmitted on the levels from $[m+1:2m]$
from transmitter~$2$. The random bits are not canceled at receiver~$1$. Further, receiver~$1$ has no knowledge
of these random bits. Hence, it cannot decode the bits intended to receiver~$2$. Also, the data bits of transmitter~$2$
received through
the cooperative link are transmitted at the upper levels $[n-C_1+1:n]$ from transmitter~$1$. Again, in order to
ensure secrecy at receiver~$1$, transmitter~$2$ sends the same data bits at levels $[m-C_1+1:m]$ along
with the $C_1$ data bits of transmitter~$1$, also received through cooperation. This not only cancels the
interference due to the bits sent on levels $[n-C_1+1:n]$ at receiver~$1$, but also enables transmitter~$2$ to
relay the data bits of transmitter~$1$.

In the remaining upper levels $[m+1:n-C_1]$, transmitter~$1$  sends its own data bits xored with random bits.
Transmitter~$2$ transmits the same random bits on levels $[1:C_1]$ to cancel the random bits at receiver~$1$.
In this way, transmitter $1$ sends $m-C_1$ data bits of its own and $C_1$ data bits of transmitter~$2$,
in the first time slot. Simultaneously, transmitter~$2$ is able to send $m$ data bits of its own and $C_1$
data bits of transmitter~$1$. In the second time slot, the roles of transmitters~$1$ and $2$ are reversed.

In contrast to the achievable schemes for other interference regimes, transmitters exchange both random
bits and data bits through the cooperative links. However, as the capacity of the cooperative links increases,
it is required to exchange less number of random bits.

When $\alpha > 2$, it is straightforward to extend the achievable scheme described above. Both the transmitter
exchanges $\frac{m}{2}$ random bits and $C' \triangleq C - \frac{m}{2}$ data bits. Out of $C'$ data bits
obtained through cooperation, $n-2m$ data bits are securely relayed using the levels $[m+1:n-m]$. The $m$
random bits and the remaining $C_1 \triangleq C' - n + 2m$ data bits obtained through cooperation are
used in a similar manner as explained for the $\alpha=2$ case. The signal of transmitter~$1$ in the first time
slot is encoded as follows:
\begin{align}
\xbold_1 & = \lsqb\begin{array}{l}
	           \mathbf{0}_{(n-m) \times 1} \\ \\
		   \mathbf{d}_{m \times 1}^1
	           \end{array} \rsqb \oplus
		\lsqb\begin{array}{l}
	           \mathbf{b}_{C_1 \times 1}^c \\
	           \mathbf{a}_{(m-C_1) \times 1} \oplus \mathbf{d}_{(m-C_1) \times 1}^{2}\\
	           \mathbf{0}_{(n-m) \times 1}
	           \end{array} \rsqb  \oplus
	         \lsqb\begin{array}{l}
	           \mathbf{0}_{m \times 1} \\
		   \mathbf{b'}_{(n-2m) \times 1}^c \\
		   \mathbf{0}_{m \times 1}
	           \end{array} \rsqb  , \label{eq:sldic-veryhigh4}
\end{align}
where $\mathbf{d}^1 \triangleq [e_{m/2}, d_{m/2}, \ldots, e_1,d_1]^T$, $\mathbf{b}^c \triangleq [b_n, b_{n-1}, \ldots, b_{n-C_1+1}]^T$,
$\mathbf{a} \triangleq [a_{m-C_1}, \ldots, a_2,a_1]^T$, $\mathbf{d}^2 \triangleq [d_q, e_q, \ldots, d_1,e_1]^T$
if $m-C_1$ is even, $\mathbf{d}^2 \triangleq [e_{q+1}, d_q, e_q, \ldots, d_1,e_1]^T$ if $m-C_1$ is odd,
$q \triangleq \lfloor \frac{m-C_1}{2}\rfloor$ and $\mathbf{b'}^c \triangleq [b_{n-m}, b_{n-m-1} \ldots, b_{m+1}]^T$.

The signal of transmitter $2$ in the first time slot is encoded as follows:
\begin{align}
\xbold_2 & = \lsqb\begin{array}{l}
	           \mathbf{b}_{m \times 1} \oplus \mathbf{e}_{m \times 1}^2 \\ \\
		   \mathbf{0}_{(n-m) \times 1}
	           \end{array} \rsqb \oplus
		\lsqb\begin{array}{l}
	           \mathbf{0}_{(n-m) \times 1} \\
		   \mathbf{b}_{C_1 \times 1}^l \oplus \mathbf{a}_{C_1 \times 1}^c\\
		   \mathbf{e}_{(m-C_1) \times 1}^{1}
	           \end{array} \rsqb \oplus
 	          \lsqb\begin{array}{l}
	           \mathbf{0}_{m \times 1} \\
		   \mathbf{a'}_{(n-2m) \times 1}^c \\
		   \mathbf{0}_{m \times 1}
	           \end{array} \rsqb, \label{eq:sldic-veryhigh5}
\end{align}
where $\mathbf{b} \triangleq [b_m, b_{m-1}, \ldots, b_1]^T$, $\mathbf{e}^2 \triangleq [e_{m/2}, d_{m/2},
\ldots, e_1,d_1]^T$, $\mathbf{b}^l \triangleq [b_n, b_{n-1},\ldots, b_{n-C_1+1}]^T$,
$\mathbf{a}^c \triangleq [a_m, a_{m-1}, \ldots, a_{m-C_1+1}]^T$, $\mathbf{e}^1 \triangleq [d_q, e_q, \ldots, d_1,e_1]^T$ if
$m-C_1$ is even, $\mathbf{e}^1 \triangleq [e_{q+1}, d_q, e_q, \ldots, d_1,e_1]^T$ if
$m-C_1$ is odd, $q \triangleq \lfloor \frac{m-C_1}{2}\rfloor$ and $\mathbf{a'}^c \triangleq [a_{n-m}, a_{n-m-1}, \ldots, a_{m+1}]^T$. 

In the second time slot, the encoding for transmitters $1$ and $2$ is reversed. The proposed scheme achieves the
following secrecy rate:
\begin{align}
R_s =  \frac{n}{2} - \frac{m}{4} + \frac{C}{2}.  \label{eq:sldic-veryhigh6}
\end{align}
\subsubsection{When $( n-\frac{m}{2} \leq C \leq n)$ and $m$ is even}\label{sec:sldic-veryhigh4} In this case, both the transmitters
share $C$ data bits and the achievable scheme uses interference cancelation. The signal of transmitter $1$
is encoded as follows:
\begin{align}
\xbold_1 = \lsqb\begin{array}{l}
	           \mathbf{0}_{(n-C+m)^+ \times 1} \\ \\
		   \mathbf{a}_{(C-m) \times 1}
	           \end{array} \rsqb \oplus
		\lsqb\begin{array}{l}
	           \mathbf{0}_{(n-C)^+ \times 1} \\ \\
		   \mathbf{b}_{C \times 1}
	           \end{array} \rsqb, \label{eq:sldic-veryhigh7}
\end{align}
where $\mathbf{a} \triangleq [a_{C}, a_{C-1}, \ldots, a_{m+1}]^T$ and $\mathbf{b}  \triangleq [b_C, b_{C-1}, \ldots, b_1]^T$. 

The proposed scheme achieves the following secrecy rate:
\begin{align}
R_s = C. \label{eq:sldic-veryhigh8}
\end{align}
\subsubsection{When $0 < C \leq \frac{m+1}{2}$ and $m$ is odd}\label{sec:sldic-veryhigh5} The achievable scheme is same as that mentioned for $\alpha \geq 2$ and even valued $m$ case. 
The proposed scheme achieves the following secrecy rate:
\begin{align}
	R_s = \min\{2C,m\}. \label{eq:sldic-veryhigh9}
\end{align}
\subsubsection{When $ \frac{m+1}{2}< C \leq \frac{2n-3m+1}{2}$ and $m$ is odd}\label{sec:sldic-veryhigh6} In this case, the achievable scheme is same as that mentioned for $(\frac{m}{2} < C \leq n- \frac{3m}{2})$ and even valued $m$ case. The message of transmitter $1$ is encoded as follows
\begin{align}
	\mathbf{x}_{1} =  \lsqb\begin{array}{l}
		\mathbf{d}_{m \times 1}^u \\
		\mathbf{0}_{(n-2m) \times 1} \\
		\mathbf{d}_{m \times 1}^l
	\end{array} \rsqb \oplus
	\lsqb\begin{array}{l}
		\mathbf{0}_{(n-m-C') \times 1} \\ 	    
		\mathbf{b}_{C' \times 1}^c \\
		\mathbf{0}_{m \times 1}
	\end{array} \rsqb, \label{eq:sldic-veryhigh10}
\end{align}
where $\mathbf{d}^u = [e_{\frac{m+1}{2}}, \ldots, d_1, e_1]^T$, $\mathbf{d}^l = [d_{\frac{m+1}{2}}, \ldots, e_1, d_1]^T$, $\mathbf{b}^c = [a_{m+C'}, a_{m+C'-1}, \ldots, a_{m+1}]^T$ and $C' = C - \frac{m+1}{2}$. The proposed scheme achieves the following secrecy rate
\begin{align}
	R_s = m + \min\lcb C - \frac{m+1}{2}, n-2m\rcb. \label{eq:sldic-veryhigh11}
\end{align}
\subsubsection{When $\frac{2n-3m+1}{2} < C \leq n$ and $m$ is odd}\label{sec:sldic-veryhigh7} In this case, the achievable scheme is similar to that described in \ref{sec:sldic-veryhigh3}. The proposed scheme differs in the way encoding is performed. To simplify the understanding of the encoding scheme, it is explained for the $\alpha=2$ case. In the first time slot, transmitter $1$ sends $\coneuu = \left\lceil{\frac{C}{2}}\right\rceil$ data bits of transmitter $2$ received through cooperation on the upper levels $[n - \coneuu+1:n]$. In order to ensure secrecy at receiver $1$, transmitter $2$ sends the same data bits at levels $[m-\ctwolu + 1:m]$ along with the $\ctwolu = \coneuu$ data bits of transmitter $1$, also received through cooperation. In the remaining upper levels $[m+1:n-\coneuu]$, transmitter $1$ sends $\coneul = (m-\coneuu)^+$ its own data bits xored with random bits. Transmitter $2$ sends the same random bits on levels $[1:\ctwoll]$ $(\ctwoll = \coneul)$ to cancel the random bits at receiver 
$1$. The number of random cooperative bits used in such transmission is $C_1^r = \left\lceil \frac{\coneul}{2} \right\rceil$. Also, transmitter $1$ relays $\conelu = (C - \coneuu - C_1^r)^{+}$ data bits of transmitter $2$ received through cooperation on the levels $[m-\conelu+1:m]$. As the links corresponding to these levels are not present to receiver $1$, these data bits remain secure. Transmitter $1$ sends $\conell = \min\{2C_1^r, (m-\conelu)^{+})\}$ of random bits on the levels $[1:\conell]$ to ensure secrecy for transmitter $2$ data. Transmitter $2$ sends its $\ctwoul=\conell$ data bits precoded with the same random bits transmitted on the levels $[1: \conell]$, to eliminate the random bits at receiver $2$. The number of cooperative random bits used by transmitter $2$ is $C_2^r = \max\lcb \left\lceil \frac{\ctwoll}{2}\right\rceil, \left\lfloor \frac{\ctwoul}{2}\right\rfloor\rcb$. Then. transmitter $2$ can relay the remaining $\ctwouu = (C - \ctwolu - C_2^r)^+$ cooperative data bits on the upper levels $[
n-\ctwouu +1:n]$. As these bits will cause interference at receiver $2$, transmitter $1$ sends the same data bits on the levels $[m-\ctwouu+1:m]$ to cancel the interference at receiver $2$. In the first time slot, the message of transmitter $1$ is encoded as follows.
\begin{align}
	\mathbf{x}_1 & = \lsqb\begin{array}{l}
		\mathbf{b'}_{\coneuu \times 1}^{uu} \\ 
		\mathbf{a}_{\coneul \times 1}^{ul} \\ \\
		\mathbf{0}_{(n-\coneuu-\coneul) \times 1}
	\end{array} \rsqb \oplus
	\lsqb\begin{array}{l}
		\mathbf{0}_{(n-m) \times 1} \\ 
		\mathbf{b'}_{\conelu  \times 1}^{lu} \\
		\mathbf{0}_{(m-\conelu - \conell)^+ \times 1} \\
		\mathbf{d}_{\conell \times 1}
	\end{array} \rsqb \oplus   
	\lsqb\begin{array}{l}
		\mathbf{0}_{\coneuu \times 1} \\ 
		\mathbf{d}_{\coneul  \times 1}^{ul} \\ \\
		\mathbf{0}_{(n-\coneuu-\coneul) \times 1}
	\end{array} \rsqb \oplus 
	\lsqb\begin{array}{l}
		\mathbf{0}_{m \times 1} \\
		\mathbf{b}_{(n-2m) \times 1}^c \\
		\mathbf{a}_{\ctwouu  \times 1}^{lu} \\ 
		\mathbf{0}_{(m-\ctwouu)^+ \times 1}
	\end{array} \rsqb,\label{eq:sldic-veryhigh12}
\end{align}
where  $\mathbf{b'}^{uu} = [b_{n}, b_{n-1}, \ldots, b_{n-\coneuu+1}]^T$, $\mathbf{b'}^{lu} = [b_{m}, b_{m-1}, \ldots, b_{m-\conelu+1}]^T$,  $\mathbf{a}^{ul} = [a_{\coneul}, a_{\coneul-1},\\ \ldots, a_{1}]^T$,  $\mathbf{d}^{ul} = [d_{C_1^r}, e_{C_1^r}, \ldots, d_1,e_1]^T$ if $C_1^r$ is even,  $\mathbf{d}^{ul} = [e_{C_1^r}, d_{C_1^r-1}, e_{C_1^r-1}, \ldots, d_1,e_1]^T$ if $C_1^r$ is odd,   $\mathbf{a}^{lu} = [a_{n}, a_{n-1}, \ldots, a_{n-\ctwouu+1}]^T$, $\mathbf{d}^{ll} = [e_{\conell}, d_{\conell}, \ldots, e_1,d_1]^T$ if $\conell$ is even, $\mathbf{d}^{ll} = [d_{\conell}, e_{\conell-1},d_{\conell-1} \ldots, e_1,d_1]^T$ if $\conell$ is odd and $\mathbf{b'}^{c} = [b_{n-m}, b_{n-m-1}, \ldots, b_{m+1}]^T$. The message of transmitter $2$ is encoded as follows
\begin{align}
	\mathbf{x}_2 & = \lsqb\begin{array}{l}
		\mathbf{0}_{(m-\ctwoul) \times 1}^{uu} \\
		\mathbf{b}_{\ctwoul \times 1}^{ul} \\
		\mathbf{0}_{(n-2m) \times 1} \\
		\mathbf{b}_{\ctwolu \times 1}^{lu} \\
		\mathbf{0}_{(m-\ctwolu) \times 1}
	\end{array} \rsqb \oplus
	\lsqb\begin{array}{l}
		\mathbf{a'}_{\ctwouu \times 1}^{uu} \\ 
		\mathbf{0}_{(m-\ctwouu)  \times 1} \\
		\mathbf{a}_{(n-2m) \times 1}^c \\
		\mathbf{a'}_{\ctwolu \times 1}^{lu} \\
		\mathbf{0}_{(m-\ctwolu) \times 1}
	\end{array} \rsqb \oplus   
	\lsqb\begin{array}{l}
		\mathbf{0}_{(n-\ctwoll) \times 1} \\ \\ \\ \\
		\mathbf{e}_{\ctwoll \times 1}^{ll}
	\end{array} \rsqb \oplus 
	\lsqb\begin{array}{l}
		\mathbf{0}_{(m-\ctwoul) \times 1} \\ \\ 
		\mathbf{e}_{\ctwoul \times 1}^{ul} \\ \\
		\mathbf{0}_{(n-m) \times 1}
	\end{array} \rsqb,\label{eq:sldic-veryhigh13}
\end{align}
where $\mathbf{b}^{ul} = [b_{\ctwoul}, b_{\ctwoul-1}, \ldots, b_{1}]^T$, $\mathbf{b}^{lu} = [b_{n}, b_{n-1}, \ldots, b_{n-\ctwolu+1}]^T$, $\mathbf{a'}^{uu} = [a_{n}, b_{n-1}, \ldots, b_{n-\ctwouu+1}]^T$, $\mathbf{a}^{c} = [ a_{n-m}, a_{n-m-1}, \ldots, a_{m+1}]^T$, $\mathbf{a'}^{lu} = [a_{m}, a_{m-1}, \ldots, a_{m-\ctwolu+1}]^T$, $\mathbf{e}^{ll} = [d_{\frac{\ctwoll}{2}}, e_{\frac{\ctwoll}{2}}, \ldots, d_1, e_1]^T$ if $\ctwoll$ is even,  $\mathbf{e}^{ll} = [e_{\left\lceil\frac{\ctwoll}{2}\right\rceil}, d_{\left\lceil\frac{\ctwoll}{2}\right\rceil-1}, \ldots, d_1, e_1]^T$ if $\ctwoll$ is odd, $\mathbf{e}^{ul} = [e_{\frac{\ctwoul}{2}}, d_{\frac{\ctwoll}{2}}, \ldots, e_1, d_1]^T$ if $\ctwoul$ is even and $\mathbf{e}^{ll} = [d_{\left\lceil\frac{\ctwoul}{2}\right\rceil}, e_{\left\lceil\frac{\ctwoul}{2}\right\rceil-1}, \ldots, e_1, d_1]^T$ if  $\ctwoul$ is odd.

In the second time slot, the encoding scheme is reversed for transmitter $1$ and $2$. The proposed scheme achieves the following secrecy rate
\begin{align}
	R_s = n-2m + \frac{1}{2}\lsqb \coneul + 2\coneuu + \ctwouu + \conelu + \ctwoul\rsqb. \label{eq:sldic-veryhigh14}
\end{align}
\subsection{Proof of Theorem \ref{th:theorem_GSIC_outer1}}\label{sec:appendouter1}
Using Fano's inequality, the rate of user $1$ is upper bounded as
\begin{align}
& NR_1 \nonumber \\
& \leq I(W_1;\ybold_1^N) + N\epsilon_N, \nonumber \\
& \stackrel{(a)}{\leq} h(\ybold_1^N) - h(\ybold_1^N|W_1, \xbold_1^N) + N\epsilon_N, \nonumber \\
& \stackrel{(b)}{\leq} h(\ybold_1^N) - h(h_c\xbold_2^N + \zbold_1^N|\coopsignal, W_1, \xbold_1^N) + N\epsilon_N, \nonumber \\
& \stackrel{(c)}{=} h(\ybold_1^N) - h(h_c\xbold_2^N + \zbold_1^N|\coopsignal) + N\epsilon_N, \nonumber \\
& \stackrel{(d)}{=} h(\ybold_1^N) - h(h_c\xbold_2^N + \widetilde{\zbold}_1^N|\coopsignal) + N\epsilon_N, \nonumber \\
& \text{or } h(\sboldtilde_2^N|\coopsignal)  \leq  h(\ybold_1^N) - NR_1 + N\epsilon_N, \nonumber \\
& \qquad \qquad \quad \text{ where } \sboldtilde_2^N \triangleq h_c\xbold_2^N + \widetilde{\zbold}_1^N, \label{eq:appendouter1}
\end{align}
where (a) and (b) follow by using the fact that the entropy cannot increase by additional conditioning; (c) follows by using the 
relation in (\ref{eq:thouter0}), and (d) is obtained using the fact that the secrecy capacity region of an 
IC with confidential messages is invariant under any joint channel noise distribution $P(\zbold_1, \zbold_2)$ 
that leads to the same marginal distributions $P(\zbold_1)$ and $P(\zbold_2)$ \cite{he-CISS-2009}. Although this 
invariance property is stated for GIC in \cite{he-CISS-2009}, it is not difficult to see that this property holds for 
the GIC with limited-rate transmitter cooperation also.

Adopting similar steps as was used to obtain (\ref{eq:appendouter1}), the following bound on the conditional 
entropy is obtained.
\begin{align}
h(\sboldtilde_1^N|\coopsignal) & \leq  h(\ybold_2^N) - NR_2 + N\epsilon_N, \nonumber \\
& \qquad \qquad \quad \text{ where } \sboldtilde_1^N \triangleq h_c\xbold_1^N + \widetilde{\zbold}_2^N, \label{eq:appendouter2}
\end{align}
The rate of user $1$ can also be bounded as
\begin{align}
& NR_1 \nonumber \\
& \leq I(W_1;\ybold_1^N) + N\epsilon, \nonumber \\
& \stackrel{(a)}{\leq} I(W_1;\ybold_1^N) -  I(W_1;\ybold_2^N) + N\epsilon', \nonumber \\
& \stackrel{(b)}{\leq} I(W_1;\ybold_1^N,\ybold_2^N) -  I(W_1;\ybold_2^N) + N\epsilon', \nonumber \\
& = I(W_1;\ybold_1^N|\ybold_2^N) + N\epsilon', \nonumber\\
& = h(\ybold_1^N|\ybold_2^N) - h(\ybold_1^N|\ybold_2^N, W_1) + N\epsilon', \nonumber \\
& = h(\ybold_1^N,\ybold_2^N) -h(\ybold_2^N)- h(\ybold_1^N|\ybold_2^N, W_1) + N\epsilon', \nonumber \\
& \stackrel{(c)}{=} h(\ybold_1^N,\ybold_2^N,\sboldtilde_1^N, \sboldtilde_2^N) - h(\sboldtilde_1^N, \sboldtilde_2^N|\ybold_1^N,\ybold_2^N) - h(\ybold_2^N)  -  h(\ybold_1^N|\ybold_2^N, W_1) + N\epsilon', \nonumber \\
& = h(\sboldtilde_1^N, \sboldtilde_2^N) + h(\ybold_1^N,\ybold_2^N|\sboldtilde_1^N, \sboldtilde_2^N) - 
h(\sboldtilde_1^N, \sboldtilde_2^N|\ybold_1^N,\ybold_2^N)  - h(\ybold_2^N)- h(\ybold_1^N|\ybold_2^N, W_1) + N\epsilon', \nonumber 
\\
& = I(\sboldtilde_1^N, \sboldtilde_2^N;\ybold_1^N,\ybold_2^N) + h(\ybold_1^N,\ybold_2^N|\sboldtilde_1^N, \sboldtilde_2^N) - h(\ybold_2^N) - h(\ybold_1^N|\ybold_2^N, W_1) + N\epsilon', \nonumber \\
& \stackrel{(d)}{\leq} I(\sboldtilde_1^N, \sboldtilde_2^N;\ybold_1^N,\ybold_2^N, \coopsignal) + h(\ybold_1^N,\ybold_2^N|\sboldtilde_1^N, \sboldtilde_2^N)  - h(\ybold_2^N) - h(\ybold_1^N|\ybold_2^N, W_1) + N\epsilon', \nonumber \\
& \leq H(\coopsignal) + I(\sboldtilde_1^N, \sboldtilde_2^N;\ybold_1^N,\ybold_2^N|\coopsignal) + h(\ybold_1^N,\ybold_2^N|\sboldtilde_1^N, \sboldtilde_2^N)  - h(\ybold_2^N) - h(\ybold_1^N|\ybold_2^N, W_1) + N\epsilon', \nonumber 
\end{align}
\begin{align}
& \leq H(\coopsignalone^N) + H(\coopsignaltwo^N) + h(\sboldtilde_1^N, \sboldtilde_2^N|\coopsignal) - h(\sboldtilde_1^N, \sboldtilde_2^N|\ybold_1^N,\ybold_2^N,\coopsignal) + h(\ybold_1^N,\ybold_2^N|\sboldtilde_1^N, \sboldtilde_2^N)  \nonumber \\
& \qquad - h(\ybold_2^N) - h(\ybold_1^N|\ybold_2^N, W_1) + N\epsilon', \nonumber \\ 
& \stackrel{(e)}{\leq} H(\coopsignalone^N) + H(\coopsignaltwo^N) + h(\sboldtilde_1^N|\coopsignal) + h(\sboldtilde_2^N|\coopsignal)  - h(\sboldtilde_1^N, \sboldtilde_2^N|\ybold_1^N,\ybold_2^N,\coopsignal,\xbold_1^N,\xbold_2^N)  \nonumber \\
& \qquad    + h(\ybold_1^N,\ybold_2^N|\sboldtilde_1^N, \sboldtilde_2^N) - h(\ybold_2^N) -h(\ybold_1^N|\ybold_2^N, W_1,\xbold_1^N,\xbold_2^N) + N\epsilon', \nonumber \\
& = H(\coopsignalone^N) + H(\coopsignaltwo^N) + h(\sboldtilde_1^N|\coopsignal) + h(\sboldtilde_2^N|\coopsignal)   - h(\widetilde{\zbold}_1^N,\widetilde{\zbold}_2^N)+ h(\ybold_1^N,\ybold_2^N|\sboldtilde_1^N, \sboldtilde_2^N)   \nonumber \\
& \qquad - h(\ybold_2^N) - h(\zbold_1^N) + N\epsilon', \label{eq:appendouter5}
\end{align}
where (a) is obtained using the secrecy constraint at receiver $2$; (b) is due to the genie providing $\ybold_2^N$ to receiver~$1$; (c) is obtained using the relation $h(\ybold_1^N,\ybold_2^N,\sboldtilde_1^N, \sboldtilde_2^N)  = h(\ybold_1^N,\ybold_2^N) + h(\sboldtilde_1^N, \sboldtilde_2^N|\ybold_1^N,\ybold_2^N)$; (d) is obtained using chain rule for mutual information, and (e) is obtained using the fact that removing conditioning cannot decrease the entropy and conditioning cannot increase the entropy.

Using (\ref{eq:appendouter1}) and (\ref{eq:appendouter2}), (\ref{eq:appendouter5}) becomes
\begin{align}
 N[2R_1 + R_2] & \leq H(\coopsignalone^N) + H(\coopsignaltwo^N)  + h(\ybold_1^N)  + h(\ybold_1^N,\ybold_2^N|\sboldtilde_1^N, \sboldtilde_2^N) - h(\widetilde{\zbold}_1^N) -h(\widetilde{\zbold}_2^N)  - h(\zbold_1^N) + N\epsilon'', \nonumber \\
 \text{or } R  & \leq \displaystyle\max_{0 \leq |\rho| \leq 1} \frac{1}{3}\lsqb 2C_G + 0.5\log\lb 1 + \SNRsum  + 2\rho\SNRprod\rb  + 0.5\log \text{det}\lb \Sigma_{\ybar|\sbar}\rb\rsqb. \label{eq:appendouter6}
\end{align}
In the above equation, $\rho$, $\text{det}(\cdot)$ and $\Sigma_{\ybar|\sbar}$ are as defined in the statement of the theorem. The second term in (\ref{eq:appendouter6}) is obtained using the fact that differential entropy is maximized by the Gaussian distribution for a given power constraint. Hence, the following holds.
\begin{align}
h(\ybold_1) \leq  0.5\log\lb 2\pi e \lb 1 + \SNRsum + 2\rho\SNRprod \rb\rb, \label{eq:appendouter7}
\end{align}
where $\text{SNR}$ and $\text{INR}$ are as defined in the statement of the theorem. The last term in (\ref{eq:appendouter6}) is obtained as follows.
\begin{align}
h(\ybold_1,\ybold_2|\sboldtilde_1, \sboldtilde_2) \leq 0.5\log\text{det}\lb 2\pi e \Sigma_{\ybar|\sbar}\rb, \label{eq:appendouter8}
\end{align}
where $\Sigma_{\ybar|\sbar} \triangleq \Sigma_{\ybar} - \Sigma_{\ybar,\sbar}\Sigma_{\sbar}^{-1}\Sigma_{\ybar,\sbar}^{T}$, $\Sigma_{\ybar} \triangleq E[\bar{\ybold}\bar{\ybold}^T]$, $\Sigma_{\ybar,\sbar} \triangleq E[\bar{\ybold}\bar{\sbold}^T]$, $\Sigma_{\sbar} \triangleq E[\bar{\sbold}\bar{\sbold}^T]$, $\ybar \triangleq [\ybold_1\:\:\ybold_2]^T$, and  $\sbar \triangleq [\sboldtilde_1\:\:\sboldtilde_2]^T$. The evaluation of these terms are given in the statement of the theorem. This completes the proof.
\subsection{Proof of Theorem \ref{th:theorem_GSIC_outer2}}\label{sec:appendouter2}
Using Fano's inequality, the rate of user~$1$ is upper bounded as
\begin{align}
& NR_1 \nonumber \\
& \leq I(W_1;\ybold_1^N) + N\epsilon_N, \nonumber \\
& \stackrel{(a)}{\leq} I(W_1;\ybold_1^N,\xbold_2^N) + N\epsilon_N, \nonumber \\
& = I(W_1;\xbold_2^N) + I(W_1;\ybold_1^N|\xbold_2^N) +  N\epsilon_N, \nonumber \\
& = I(W_1;\xbold_2^N) + h(\sbolddash_1^{N}|\xbold_2^N) - h(\sbolddash_1^{N}|\xbold_2^N,W_1) + N\epsilon_N, \text{ where } \sbolddash_1^N \triangleq h_d \xbold_1^N + \zbold_1^N, \nonumber \\
& = I(W_1;\xbold_2^N) + I(W_1;\sbolddash_1^N|\xbold_2^N) + N\epsilon_N, \nonumber  \\
& \stackrel{(b)}{=}  I(W_1;\xbold_2^N, \sbolddash_1^N) + N\epsilon_N, \nonumber \\
& \stackrel{(c)}{\leq}   I(W_1;\sbolddash_1^N) - I(W_1;\ybold_2^N) + I(W_1;\xbold_2^N|\sbolddash_1^N) +  N\epsilon_N', \nonumber \\
&  \stackrel{(d)}{\leq}   I(W_1;\sbolddash_1^N,\ybold_2^N) - I(W_1;\ybold_2^N) + I(W_1;\xbold_2^N|\sbolddash_1^N) +  N\epsilon_N', \nonumber \\
& = I(W_1;\sbolddash_1^N|\ybold_2^N) + I(W_1;\xbold_2^N|\sbolddash_1^N) +  N\epsilon_N', \nonumber \\
& \stackrel{(e)}{\leq}  I(W_1;\sbolddash_1^N|\ybold_2^N) + I(W_1;\xbold_2^N,\coopsignal|\sbolddash_1^N) +  N\epsilon_N', \nonumber \\
& = I(W_1;\sbolddash_1^N|\ybold_2^N) + I(W_1;\coopsignal|\sbolddash_1^N) + I(W_1;\xbold_2^N|\sbolddash_1^N, \coopsignal) +  N\epsilon_N', \nonumber \\
& \leq  I(W_1;\sbolddash_1^N|\ybold_2^N) + H(\coopsignal|\sbolddash_1^N) + h(\xbold_2^N|\sbolddash_1^N, \coopsignal) \nonumber \\
& \qquad \qquad - h(\xbold_2^N|\sbolddash_1^N, \coopsignal,W_1)  +  N\epsilon_N', \nonumber \\
& \stackrel{(f)}{\leq}  h(\sbolddash_1^N|\ybold_2^N) - h(\sbolddash_1^N|\ybold_2^N,W_1) + H(\coopsignal) +  N\epsilon_N', \nonumber \\
& \text{or } R_1  \leq \displaystyle\max_{0 \leq |\rho| \leq 1} \lsqb 2C_G + 0.5\log \Sigma_{\sbolddash|\ybold_2}\rsqb, \label{eq:appendoutertwo1}
\end{align}
where (a) is due to the genie providing $\xbold_2^N$ to receiver $1$; (b) is obtained using chain rule for mutual information; 
(c) is obtained using secrecy constraint at receiver~$2$; (d) is due to the genie providing $\ybold_2^N$ as 
side information to receiver~$1$, where $\xbold_2^N$ is eliminated, (e) is obtained using the relation 
$I(W_1;\xbold_2^N,\coopsignal|\sbolddash_1^N) = I(W_1;\xbold_2^N|\sbolddash_1^N) + I(W_1;\coopsignal|\sbolddash_1^N,\xbold_2^N)$ 
and (f) is obtained using the relation in (\ref{eq:thouter0}) and the fact that removing conditioning does not decrease the entropy. 
The last inequality is obtained using the fact that the differential entropy is maximized by the Gaussian distribution 
for a given power constraint. The term $\Sigma_{\sbolddash|\ybold_2}$ is evaluated as follows.
\begin{align}
\Sigma_{\sbolddash|\ybold_2} & = E[\sbolddash_1^2] - E[\sbolddash\ybold_2]^2E[\ybold_2^2]^{-1}  \nonumber \\
& = 1 + \frac{\text{SNR} + \text{SNR}^2(1-\rho^2)}{1 + \SNRsum + 2\rho\SNRprod}. \label{eq:appendoutertwo2}
\end{align}
This completes the proof.
\subsection{Proof of Theorem \ref{th:theorem_GSIC_outer3}}\label{sec:appendouter3}
Using Fano's inequality, the rate of user $1$ is upper bounded as
\begin{align}
	NR_1 & \leq I(W_1;\ybold_1^N) + N\epsilon_N, \nonumber \\
	& \stackrel{(a)}{\leq} I(W_1;\ybold_1^N)-I(W_1;\ybold_2^N) + N\epsilon_N', \nonumber \\
	& \stackrel{(b)}{\leq} I(W_1;\ybold_1^N,\ybold_2^N)-I(W_1;\ybold_2^N) + N\epsilon_N', \nonumber \\
	& = I(W_1;\ybold_1^N|\ybold_2^N)-I(W_1;\ybold_2^N) + N\epsilon_N', \nonumber \\
	& \leq  h(\ybold_1^N|\ybold_2^N) - h(\zbold_1^N) + N\epsilon_N', \nonumber \\
	\text{or } R_1 & \leq \displaystyle\max_{0 \leq |\rho| \leq 1} 0.5\log \Sigma_{\ybold_1|\ybold_2}, \label{eq:appendouterthree1}
\end{align}
where $\Sigma_{\ybold_1|\ybold_2}$ is obtained as given below.
\begin{align}
	\Sigma_{\ybold_1|\ybold_2} & = E[\ybold_1^2] - E[\ybold_1\ybold_2]^2 E[\ybold_2^2]^{-1}, \nonumber \\
	& = 1 + \SNRsum + 2\rho\SNRprod-\frac{(2\SNRprod + \rho(\SNRsum))^2}{1 + \SNRsum + 2\rho\SNRprod}. \label{eq:appendouterthree2}
\end{align}
Substituting the value of $\Sigma_{\ybold_1|\ybold_2}$ from (\ref{eq:appendouterthree2}) in (\ref{eq:appendouterthree1}), results in (\ref{eq:theorem_GSIC_outer3}) and this completes the proof.
\subsection{Proof of Theorem~\ref{th:theorem_ach_weakint}}\label{sec:appendweakintf}
The proof involves analyzing the error probability at the encoder and decoder along with equivocation computation. One of the novelties in obtaining the achievable scheme lies in the choice of the cooperative private auxiliary codewords $\ubold_1$, and $\ubold_2$ so that the codeword carrying the cooperative private part of the message $(w_{cpi})$ is canceled at the unintended receiver. This simultaneously eliminates interference and ensures secrecy of the cooperate private message. For ensuring secrecy of the non-cooperative private message, it is required to show that the weak secrecy constraint is satisfied at the receiver~$j$, i.e., $H(W_{pi}|\ybold_j^N) \geq N[R_{pi}-\epsilon_s]$. In the equivocation computation, the main novelty lies in choosing the value of the rate sacrificed in confusing the unintended receiver $(R_{pi}')$ and rate of the dummy message $(R_{di})$ so that the weak secrecy constraint is satisfied. 
\subsubsection{Analysis of the probability of error} In the following, the probability of error is analyzed at the encoder and decoder. \\
\textbf{Encoding error:} For the Marton's coding scheme to succeed, the transmitter $1$ and $2$ can find at least one jointly typical sequence pair $(\ubold_1^N,\ubold_2^N)$ provided following condition is satisfied.
\begin{align}
& R_{cp1} \leq \widetilde{R}_{cp1},  R_{cp2} \leq \widetilde{R}_{cp2}, \nonumber \\
& R_{cp1} \leq C,  R_{cp2} \leq C, \nonumber \\
& \widetilde{R}_{cp1} - R_{cp1} + \widetilde{R}_{cp2} - R_{cp2} \geq I(\ubold_1;\ubold_2). \label{eq:pfweakint1}
\end{align}
When $\ubold_1$ and $\ubold_2$ are chosen independent of each other, then $I(\ubold_1;\ubold_2) = 0$. Given $R_{cp1} \leq \widetilde{R}_{cp1}$ and  $R_{cp2} \leq \widetilde{R}_{cp2}$, the last condition in (\ref{eq:pfweakint1}) becomes redundant.\\
\textbf{Decoding error:} Define the following event
\begin{align}
E_{ijk} = \lcb \lb \ybold_1^N, \xbold_{p1}^N(i,j), \ubold_1^N(k) \rb \in T_{\epsilon}^{N}\rcb. \label{eq:pfweakint3}
\end{align}
Without loss of generality, assume that transmitter $1$ and $2$ sends $\xbold_1^N(1,1,1,1)$ and $\xbold_2^N(1,1,1,1,1,1)$ with $(\tildewcpone,\tildewcptwo) =(1,1)$, respectively. An error occurs if the transmitted and received codewords are not jointly typical or a wrong codeword is jointly typical with the received codewords. Then by the union of events bounds
\begin{align}
\lambda_{e}^{(n)}  = P\lb E_{111}^c \bigcup \cup_{i \neq 1, j \neq 1, k \neq 1}E_{ijk}\rb \leq P(E_{111}^c) + P( \cup_{i \neq 1, j \neq 1, k \neq 1}E_{ijk}).  \label{eq:pfweakint4}
\end{align}
From the joint AEP \cite{cover-infotheory-2012}, $P(E_{111}^c) \rightarrow 0$ as $N \rightarrow \infty$. When $i\neq 1, j \neq 1$, and $k = 1$, then
\begin{align}
\lambda_{ij1}  = \displaystyle\sum_{i \neq 1, j \neq 1} P(E_{ij1})& \leq 2^{N\lsqb R_{p1}+R_{p1}'\rsqb}\sumAEPweak P(\xbold_{p1}^N)P(\ubold_1^N)P(\ybold_1^N|\ubold_1^N), \nonumber \\
& \leq 2^{N\lsqb R_{p1}+R_{p1}' - I(\xbold_{p1};\ybold_1|\ubold_1) + 4\epsilon\rsqb}. \label{eq:pfweakint8}
\end{align}
Hence, $\lambda_{ij1} \rightarrow 0$ as $N \rightarrow \infty$, if
\begin{align}
R_{p1} + R_{p1}' \leq  I(\xbold_{p1};\ybold_1|\ubold_1). \label{eq:pfweakint9}
\end{align}
When the above condition is satisfied, also, the probability of error $\lambda_{1j1}$ and $\lambda_{i11}$ goes to zero as $N \rightarrow \infty$. When $k\neq 1$ and $(i,j)=(1,1)$, then
\begin{align}
\lambda_{11k} = \displaystyle\sum_{k \neq 1} P(E_{11k}) & \leq 2^{N\tildercpone}\sumAEPweak P(\xbold_{p1}^N)P(\ubold_1^N)P(\ybold_1^N|\xbold_{p1}^N), \nonumber \\
& \leq 2^{N\lsqb \tildercpone - I(\ubold_1;\ybold_1|\xbold_{p1}) + 4\epsilon\rsqb}. \label{eq:pfweakint10}
\end{align}
Hence, $\lambda_{11k} \rightarrow 0$ as $N \rightarrow \infty$, if 
\begin{align}
\tildercpone \leq I(\ubold_1;\ybold_1|\xbold_{p1}). \label{eq:pfweakint11}
\end{align}
When $i \neq 1$, $j \neq 1$, and $k \neq 1$, then
\begin{align}
\lambda_{ijk}  = \displaystyle\sum_{i \neq 1, j \neq 1, k \neq 1}P(E_{ijk}) & \leq 2^{N\lsqb \tildercpone + R_{p1} + R_{p1}'\rsqb} \sumAEPweak P(\xbold_{p1}^N)P(\ubold_1^N)P(\ybold_1^N), \nonumber \\
& \leq 2^{N\lsqb \tildercpone + R_{p1} + R_{p1}' - I(\ubold_1,\xbold_{p1};\ybold_1) + 4\epsilon\rsqb}. \label{eq:pfweakint12}
\end{align}
Hence, $\lambda_{ijk} \rightarrow 0$ as $N \rightarrow \infty$, if
\begin{align}
\tildercpone + R_{p1} + R_{p1}' \leq I(\ubold_1,\xbold_{p1};\ybold_1). \label{eq:pfweakint13}
\end{align}
The above condition also ensures that $\lambda_{i1k}$ and $\lambda_{1jk}$ go to zero as $N \rightarrow \infty$. Hence, $\lambda_{e}^{(n)}$ in (\ref{eq:pfweakint4}) goes to $0$ as $N \rightarrow \infty$, when the conditions in (\ref{eq:pfweakint9}), (\ref{eq:pfweakint11}) and (\ref{eq:pfweakint13}) are satisfied. 

Now, using the Fourier-Motzkin procedure \cite{gamal-netinfotheory-2011}, the achievable rate in Theorem \ref{th:theorem_ach_weakint} is obtained. The choice of $R_{p1}'$ is discussed in Section~\ref{sec:eqvocweakint}.\\

\textbf{Choice of $\ubold_1$ and $\ubold_2$}\\
The cooperative private auxiliary codewords $\ubold_i\:(i =1,2)$, are chosen such that the interference caused by the codeword $\ubold_j$ at receiver~$i$ $(i \neq j)$ is nulled out completely. This not only eliminates the interference caused by the cooperative private part, but also ensures secrecy of the cooperative private message. Choose $\underline{\xbold}_{cp}$, $\ubold_1$ and $\ubold_2$ to be jointly Gaussian, and such that
\begin{align}
&  \underline{\xbold}_{cp} = \mathbf{w}_{1z}\underline{v}_{1z} +  \mathbf{w}_{2z}\underline{v}_{2z}, \nonumber \\
& \ubold_1 = \lsqb h_d\quad h_c\rsqb\underline{v}_{1z}\mathbf{w}_{1z}, \text{ and } \ubold_2 = \lsqb h_c\quad h_d\rsqb\underline{v}_{2z}\mathbf{w}_{2z},  \label{eq:pfhighint54}
\end{align}
where $\underline{v}_{1z} \triangleq\lsqb h_d \quad\quad-h_c \rsqb^{T}$, $\underline{v}_{2z} \triangleq \lsqb -h_c \quad\quad h_d \rsqb^{T}$
and $\mathbf{w}_{1z}$ and $\mathbf{w}_{2z}$ are independent Gaussian with variance $\sigma_{1z}^2$ and $\sigma_{2z}^2$,
respectively. Recall that $\mathbf{w}_{iz}$ $(i=1,2)$ is the codeword representing the cooperative private message
$w_{cpi}$ of transmitter $i$. The choice of $\sigma_{1z}^2$ and $\sigma_{2z}^2$ is discussed in the proof of
Corollary \ref{cor:cor_ach_weakint} (Appendix \ref{sec:corappendweakintf}).

\textit{Remark:} For the extension to asymmetric setting, the choice of the auxiliary
codewords for the cooperative private message needs to be modified as follows:
\begin{align}
& \ubold_1 = \lsqb h_{11}\quad h_{12}\rsqb\underline{v}_{1z}\mathbf{w}_{1z} = (h_{11}h_{22} - h_{12}h_{21})\wbold_{1z}, \nonumber \\
\text{ and } & \ubold_2 =
\lsqb h_{21}\quad h_{22}\rsqb\underline{v}_{2z}\mathbf{w}_{2z} = (h_{11}h_{22} - h_{12}h_{21})\wbold_{2z}, \label{eq:rev-precoding}
\end{align}
where $\underline{v}_{1z} \triangleq\lsqb h_{22} \quad\quad-h_{21} \rsqb^{T}$,
$\underline{v}_{2z} \triangleq \lsqb -h_{12} \quad\quad h_{11} \rsqb^{T}$, $\mathbf{w}_{1z}$ and
$\mathbf{w}_{2z}$ are independent Gaussian distributed with variance $\sigma_{1z}^2$ and $\sigma_{2z}^2$,
respectively. The rest of the proof follows along similar lines as explained 
below.
\subsubsection{Equivocation computation}\label{sec:eqvocweakint} The equivocation at the receiver~$2$ is bounded as follows.
\begin{align}
H(W_1|\ybold_2^N) & = H(W_{p1}, W_{cp1}|\ybold_2^N), \nonumber \\
& = H(W_{p1}|\ybold_2^N) + H(W_{cp1}|\ybold_2^N,W_{p1}). \label{eq:eqvocweak1}
\end{align}
First consider the term $H(W_{cp1}|\ybold_2^N,W_{p1})$. The output at the receiver $2$ is
\begin{align}
y_2 = u_2 + h_d x_{d2} + h_c x_{p1} + z_2.  \label{eq:eqvocweak2}
\end{align}
As $\ubold_1$ and $\ubold_2$ are chosen to be independent of each other, i.e., $I(\ubold_1;\ubold_2) = 0$, and $\wcpone$ is chosen independent of $\wprvone$, the following holds:
\begin{align}
H(W_{cp1}|\ybold_2^N,W_{p1}) = H(W_{cp1}). \label{eqvocweak3}
\end{align}
Hence, it is only required to show the following:
\begin{align}
H(W_{p1}|\ybold_2^N) \geq N\lsqb R_{p1} - \epsilon_s\rsqb. \label{eq:eqvocweak4}
\end{align}
Consider the following:
\begin{align}
H(W_{p1}|\ybold_2^N) 
& \geq H(W_{p1}|\ybold_2^N,\xbold_{p2}^N,\ubold_2^N, W_{d2}''), \nonumber \\
& \stackrel{(a)}{=} H(W_{p1},\ybold_2^N|\xbold_{p2}^N,\ubold_2^N, W_{d2}'') - H(\ybold_2^N|\xbold_{p2}^N,\ubold_2^N, W_{d2}''), \nonumber 
\\
& \stackrel{(b)}{=} H(W_{p1},\ybold_2^N, \xbold_{p1}^N,\xbold_{d2}^N|\xbold_{p2}^N,\ubold_2^N, W_{d2}'') - H(\xbold_{p1}^N,\xbold_{d2}^N|\ybold_2^N, \ubold_2^N, \xbold_{p2}^N, W_{p1},W_{d2}'') \nonumber \\
& \qquad -  H(\ybold_2^N|\xbold_{p2}^N,\ubold_2^N, W_{d2}''), \nonumber \\
& = H(\xbold_{p1}^N,\xbold_{d2}^N|\xbold_{p2}^N,W_{d2}'')  + H(W_{p1},\ybold_2^N|\xbold_{p1}^N, \xbold_{d2}^N,\ubold_2^N,\xbold_{p2}^N,W_{d2}'') \nonumber \\
& \quad  - H(\xbold_{p1}^N,\xbold_{d2}^N|\ybold_2^N, \xbold_{p2}^N, W_{p1},W_{d2}'') -  H(\ybold_2^N|\xbold_{p2}^N,\ubold_2^N, W_{d2}''), \nonumber \\
& \geq H(\xbold_{p1}^N) + H(\xbold_{d2}^N|W_{d2}'') + H(\ybold_2^N|\xbold_{p1}^N, \xbold_{d2}^N,\ubold_2^N,\xbold_{p2}^N,W_{d2}'')  H(\ybold_2^N|\xbold_{p2}^N,\ubold_2^N) \nonumber \\
& \qquad - H(\xbold_{p1}^N,\xbold_{d2}^N|\ybold_2^N, \ubold_2^N, \xbold_{p2}^N, W_{p1},W_{d2}''), \nonumber 
\\
& \stackrel{(c)}{=} H(\xbold_{p1}^N) + H(\xbold_{d2}^N|W_{d2}'') + H(\ybold_2^N|\xbold_{p1}^N, \xbold_{d2}^N,\ubold_2^N,\xbold_{p2}^N)  -  H(\ybold_2^N|\xbold_{p2}^N,\ubold_2^N) \nonumber \\
& \qquad   -  H(\xbold_{p1}^N,\xbold_{d2}^N|\ybold_2^N, \ubold_2^N, \xbold_{p2}^N, W_{p1},W_{d2}''), \nonumber \\
& = N\lsqb R_{p1} + R_{p1}' + R_{d2}'\rsqb - I(\xbold_{p1}^N,\xbold_{d2}^N;\ybold_2^N|\ubold_2^N,\xbold_{p2}^N)  \nonumber \\
& \qquad - H(\xbold_{p1}^N,\xbold_{d2}^N|\ybold_2^N, \ubold_2^N, \xbold_{p2}^N, W_{p1},W_{d2}''), \label{eq:eqvocweak5}
\end{align}
where (a) and (b) are obtained using the relations in (\ref{eq:eqvocweak5a}) and (\ref{eq:eqvocweak5b}), respectively;  and (c) is obtained using the fact that $W_{d2}'' \rightarrow (\xbold_{p1}^N,\xbold_{d2}^N,\xbold_{p2}^N,\ubold_2^N) \rightarrow \ybold_2^N$ forms a Markov chain. This can be shown with the help of a functional dependency graph \cite{kramer-now-2008}.

\begin{align}
& H(W_{p1}, \ybold_2^N|\xbold_{p2}^N,\ubold_2^N,W_{d2}'') = H(\ybold_2^N|\xbold_{p2}^N,\ubold_2^N,W_{d2}'')  + H(W_{p1}|\ybold_2^N,\xbold_{p2}^N,\ubold_2^N,W_{d2}''), \label{eq:eqvocweak5a} \\
& H(W_{p1},\ybold_2^N,\xbold_{p1}^N,\xbold_{d2}^N|\xbold_{p2}^N,\ubold_2^N,W_{d2}'')  = H(W_{p1},\ybold_2^N|\xbold_{p2}^N,\ubold_2^N,W_{d2}'')   \nonumber \\
&  \qquad \qquad \qquad \qquad + H(\xbold_{p1}^N,\xbold_{d2}^N|\ybold_2^N,\xbold_{p2}^N,\ubold_2^N,W_{p1},W_{d2}''). \label{eq:eqvocweak5b}
\end{align}

Using Lemma~\ref{lm1} in Appendix~\ref{sec:usefullemma}, it can be shown that
\begin{align}
I(\xbold_{p1}^N,\xbold_{d2}^N;\ybold_2^N|\ubold_2^N,\xbold_{p2}^N) \leq N I(\xbold_{p1},\xbold_{d2};\ybold_2|\ubold_2,\xbold_{p2}) + N\epsilon'.  \label{eq:eqvocweak6}
\end{align}
Thus the remaining key step in showing that the condition in (\ref{eq:eqvocweak4}) is satisfied is to bound the last term in (\ref{eq:eqvocweak5}). To bound this term, consider the joint decoding of $W_{p1}'$ and $W_{d2}'$ at receiver~$2$ assuming that a genie has given $W_{p1}$ and $W_{d2}''$ as side information to receiver~$2$. For a given $W_{p1} = w_{p1}$ and $W_{d2}'' = w_{d2}''$, assume that $w_{p1}'$ and $w_{d2}'$ are sent by transmitters~$1$ and $2$, respectively and receiver~$2$ knows the sequence $\ybold_2^N = y_2^N$ and $\ubold_2^N = u_2^N$.  For a given $W_{p1} = w_{p1}$ and $W_{d2}''= w_{d2}''$, receiver~$2$ declares that $j$ and $l$ was sent if $(\xbold_{p1}^N(w_{p1},j), \xbold_{d2}^N(l,w_{d2}''),\ybold_2^N)$ is jointly typical and such $(j,l)$ exists and is unique. Otherwise, an error is declared. Now, define the following event
\begin{align}
& E_{jl}^1 =   \lcb (\xbold_{p1}^N(w_{p1},j), \xbold_{d2}^N(l,w_{d2}''),\ybold_2^N) \in T_{\epsilon}^{N}(P_{\xbold_{p1},\xbold_{d2},\ybold_2|\ubold_2,\xbold_{p2}})\rcb,\label{eq:eqvocwak7}
\end{align}
where $T_{\epsilon}^{N}(P_{X_{p1}X_{d2}Y_2|U_2X_{p2}})$ denotes, for given typical sequences $\ubold_2$ and $\xbold_{p2}$, the set of jointly typical sequences $\ybold_1, \xbold_{p1}$, and $\ubold_1$ with respect to $P_{X_{p1}X_{d2}Y_2|U_2X_{p2}}$. Without loss of generality, assume that $\xbold_{p1}^N(w_{p1},1)$ and $\xbold_{d2}^{N}(1,w_{d2}'')$ were sent. Then, by the union of events bound, the following is obtained:
\begin{align}
P_{e1}^{N} & = P\lb E_{11}^{1^c}\bigcup \cup_{j \neq 1, l \neq 1}E_{jl}^1\rb, \nonumber \\
& \leq P(E_{11}^{1^c}) + 2^{NR_{p1}'}2^{-N\lsqb I(\xbold_{p1};\ybold_{2}|\xbold_{d2},\ubold_2,\xbold_{p2}) - 3\epsilon\rsqb}  + 2^{NR_{d2}'}2^{-N\lsqb I(\xbold_{d2};\ybold_{2}|\xbold_{p1},\ubold_2,\xbold_{p2}) - 3\epsilon\rsqb} \nonumber \\
& \qquad + 2^{N\lb R_{p1}' + R_{d2}'\rb}2^{-N\lsqb I(\xbold_{p1},\xbold_{d2};\ybold_{2}|\ubold_2,\xbold_{p2}) - 3\epsilon\rsqb}.  \label{eq:eqvocweak8}
\end{align}
Hence, the probability of error $P_{e1}^{N}$ is arbitrarily small for large $N$, provided the following conditions are satisfied.
\begin{align}
& R_{P1}' \leq I(\xbold_{p1};\ybold_{2}|\xbold_{d2},\ubold_2,\xbold_{p2}), \nonumber \\
& R_{d2}' \leq I(\xbold_{d2};\ybold_{2}|\xbold_{p1},\ubold_2,\xbold_{p2}), \nonumber \\
& R_{p1}' + R_{d2}' \leq I(\xbold_{p1},\xbold_{d2};\ybold_{2}|\ubold_2,\xbold_{p2}).  \label{eq:eqvocweak9}
\end{align}
When the conditions in (\ref{eq:eqvocweak9}) are satisfied and for sufficiently large $N$, the following bound is obtained using Fano's inequality:
\begin{align}
& \frac{1}{N} H(\xbold_{p1}^N,\xbold_{d2}^N| \ybold_1^N, \ubold_2^N, \xbold_{p2}^N, W_{p1}=w_{p1},W_{d2}''=w_{d2}'')  \leq \frac{1}{N}[1 + P_{e1}^{N} \log 2^{N[R_{p1}' + R_{d2}']}] \leq \delta_1.\label{eq:eqvocweak10}
\end{align}
Using the above, the last term in (\ref{eq:eqvocweak5}) is bounded as follows:
\begin{align}
 & H(\xbold_{p1}^N,\xbold_{d2}^N|\ybold_{2}^N,\ubold_{2}^N,\xbold_{p2}^N, W_{p1},W_{d2}'')  \nonumber \\
 & = \displaystyle\sum_{w_{p1},w_{d2}^{''}}\!\! P(w_{p1},w_{d2}'')  H(\xbold_{p1}^N,\xbold_{d2}^N|\ybold_2^N, \ubold_2^N,\xbold_{p2}^N, W_{p1}=w_{p1},W_{d2}''=w_{d2}''), \nonumber \\
& \leq N\delta_1.  \label{eq:eqvocweak11}
\end{align}
Using (\ref{eq:eqvocweak6}) and (\ref{eq:eqvocweak11}), (\ref{eq:eqvocweak5}) becomes
\begin{align}
& H(W_{p1}|\ybold_2^N)  \geq N\lsqb R_{p1} + R_{p1}' + R_{d2}' - I(\xbold_{p1},\xbold_{d2};\ybold_2|\ubold_2,\xbold_{p2})-\epsilon_1 \rsqb, \label{eq:eqvocweak12}
\end{align}
$\text{where } \epsilon_1 = \epsilon' + \delta_1$. By choosing $R_{p1}' + R_{d2}' = I(\xbold_{p1},\xbold_{d2};\ybold_2|\ubold_2,\xbold_{p2})-\epsilon_{11}$, (\ref{eq:eqvocweak12}) becomes
\begin{align}
H(W_{p1}|\ybold_2^N) \geq N\lsqb R_{p1} - \epsilon_s \rsqb, \text{ where } \epsilon_s = \epsilon_1 + \epsilon_{11}. \label{eq:eqvocweak13}
\end{align}
Hence, by choosing $R_{p1}' = I(\xbold_{p1};\ybold_2|\xbold_{p2},\ubold_2)-\epsilon_{11}'$ and $R_{d2}' = I(\xbold_{d2};\ybold_2|\xbold_{p1},\xbold_{p2},\ubold_2)-\epsilon_{11}''$, secrecy is ensured for the non-cooperative private message of transmitter~$1$, and also, the achievability condition in (\ref{eq:eqvocweak9}) is satisfied.

For receiver~$1$, also, it is only required to show that the non-cooperative private message of transmitter~$2$ remains secure. To bound the equivocation at receiver~$1$, consider the following:
\begin{align}
& H(W_{p2}|\ybold_1^N) \geq H(W_{p2}|\ybold_1^N,\xbold_{p1}^N,\ubold_1^N, W_{d2}'). \label{eq:eqvocweak15}
\end{align}
Then, by following similar steps as used in case of receiver~$2$, it can be shown that the choice of $R_{p2}' = I(\xbold_{p2};\ybold_1|\xbold_{p1},\ubold_1)-\epsilon_2'$ and $R_{d2}'' = I(\xbold_{d2};\ybold_1|\xbold_{p1},\xbold_{p2},\ubold_1)-\epsilon_2''$, ensures secrecy of the non-cooperative private message of transmitter~$2$. This completes the proof.

The following lemma is useful in bounding the mutual information in the proof of Theorem~\ref{th:theorem_ach_weakint}.
\subsection{A useful Lemma}\label{sec:usefullemma}
\begin{lemma}\label{lm1}
\begin{align}
I(\xbold_{p1}^N,\xbold_{d2}^N;\ybold_2^N|\xbold_{p2}^N,\ubold_2^N) \leq N \lsqb I(\xbold_{p1},\xbold_{d2};\ybold_2|\xbold_{p2},\ubold_2) + \epsilon_3\rsqb, \label{eq:lmref1}
\end{align}
where $\epsilon_3$ is small for sufficiently large $N$.
\end{lemma}
\begin{proof}
Let $T_{\epsilon}^{(N)}(\probdislmone)$ denote the set of typical sequences $(\typseqlmone)$ with respect to $P(x_{p1}, x_{p2}, x_{d2}, u_2,y_2)$. Define the following indicator random variable.
\begin{align}
\mathbb{\psi}(\typseqlmone) = \left\{\begin{array}{l l}
1 \qquad \qquad \quad\mbox{$(\typseqlmone) \notin T_{\epsilon}^N(\probdislmone)$} \\ 	
0 \quad \qquad \qquad\: \mbox{otherwise}
\end{array}\right.\label{eq:lmoneproof1}
\end{align}
Now, $I(\xbold_{p1}^N,\xbold_{d2}^N;\ybold_2^N|\xbold_{p2}^N,\ubold_2^N)$ is bounded as follows.
\begin{align}
I(\xbold_{p1}^N,\xbold_{d2}^N;\ybold_2^N|\xbold_{p2}^N,\ubold_2^N) &  \leq I(\xbold_{p1}^N, \xbold_{d2}^N, \mathbb{\psi};\ybold_2^N|\xbold_{p2}^N,\ubold_2^N), \nonumber \\
& = I(\mathbb{\psi};\ybold_2^N|\xbold_{p2}^N,\ubold_2^N) + I(\xbold_{p1}^N,\xbold_{d2}^N;\ybold_2^N|\xbold_{p2}^N,\ubold_2^N,\mathbb{\psi}). \label{eq:lmoneproof2}
\end{align}
Consider the first term in (\ref{eq:lmoneproof2}).
\begin{align}
& I(\mathbb{\psi};\ybold_2^N|\xbold_{p2}^N,\ubold_2^N) \leq H(\mathbb{\psi}) \leq 1. \label{eq:lmoneproof3}
\end{align}
Consider the second term in (\ref{eq:lmoneproof2}).
\begin{align}
I(\xbold_{p1}^N,\xbold_{d2}^N;\ybold_2^N|\xbold_{p2}^N,\ubold_2^N,\mathbb{\psi}) = \displaystyle\sum_{j=0}^1 P(\mathbf{\psi}=j)I(\xbold_{p1}^N,\xbold_{d2}^N;\ybold_2^N|\xbold_{p2}^N,\ubold_2^N,\mathbb{\psi}=j). \label{eq:lmoneproof3a}
\end{align}
When $j=1$, then $(\typseqlmone) \notin T_{\epsilon}^{(N)}$ and the following bound is obtained.
\begin{align}
P(\mathbf{\psi}=1)I(\xbold_{p1}^N,\xbold_{d2}^N;\ybold_2^N|\xbold_{p2}^N,\ubold_2^N,\mathbb{\psi}=1) & \leq P\lcb (\typseqlmone) \notin T_{\epsilon}^{(N)}\rcb H(\ybold_2^N), \nonumber \\
& \leq N\epsilon_3\log|\mathcal{Y}_2|. \label{eq:lmoneproof4}
\end{align}
When $j=0$, then $(\typseqlmone) \in T_{\epsilon}^{(N)}$ and the following bound is obtained.
\begin{align}
& P(\mathbf{\psi}=0)I(\xbold_{p1}^N,\xbold_{d2}^N;\ybold_2^N|\xbold_{p2}^N,\ubold_2^N,\mathbb{\psi}=0) \nonumber \\
& \leq I(\xbold_{p1}^N,\xbold_{d2}^N;\ybold_2^N|\xbold_{p2}^N,\ubold_2^N,\mathbb{\psi}=0), \nonumber \\
& \leq \displaystyle\sum_{(\typseqlmone) \in T_{\epsilon}^{(N)}} P(\typseqlmone)\lsqb \log P(\xbold_{p1}^N,\xbold_{d2}^N,\ybold_2^N |\xbold_{p2}^N, \ubold_2^N)  \right. \nonumber \\
& \qquad \qquad \qquad \qquad \qquad \qquad \left.- \log P(\ybold_2^N|\xbold_{p2}^N,\ubold_2^N) - \log P(\xbold_{p1}^N,\xbold_{d2}^N|\xbold_{p2}^N, \ubold_2^N)\rsqb, \nonumber \\
& \leq N\lsqb H(\ybold_2|\xbold_{p2},\ubold_2) + H(\xbold_{p1},\xbold_{d2}|\xbold_{p2},\ubold_2) - H(\xbold_{p1},\xbold_{d2},\ybold_2|\xbold_{p2},\ubold_2) + 3\epsilon_3\rsqb, \nonumber \\
& = N\lsqb I(\xbold_{p1};\ybold_2|\xbold_{p2},\ubold_2) + 3\epsilon_2\rsqb. \label{eq:lmoneproof5}
\end{align}
From (\ref{eq:lmoneproof3})-(\ref{eq:lmoneproof5}), (\ref{eq:lmoneproof2}) is bounded as follows.
\begin{align}
I(\xbold_{p1}^N;\ybold_2^N|\xbold_{p2}^N,\ubold_2^N) & \leq N I(\xbold_{p1};\ybold_2|\xbold_{p2},\ubold_2) + N\epsilon_3', \label{eq:lmoneproof6}
\end{align}
where $\epsilon_3' = \epsilon_3'\log|\mathcal{Y}_2| + 3\epsilon_2 + \frac{1}{N}$. 
\end{proof}
\subsection{Proof of Corollary \ref{cor:cor_ach_weakint}}\label{sec:corappendweakintf}
In the first and second time slots, transmitters $1$ and $2$ send the following encoded messages:
\begin{align}
& \xbold_1(1) = \underline{\xbold}_{cp}[1](1) + \xbold_{p1}(1), 
\text{ and }  \xbold_2(1) = \underline{\xbold}_{cp}[2](1) + \xbold_{p2}(1) + \xbold_{d2}(1), \nonumber \\
& \xbold_1(2) = \underline{\xbold}_{cp}[1](2) + \xbold_{p1}(2)+ \xbold_{d1}(2), 
\text{ and } \xbold_2(2) = \underline{\xbold}_{cp}[2](2) + \xbold_{p2}(2), \label{eq:pfweakcor1}
\end{align}
where  $\underline{\xbold}_{cp}$ is as defined in (\ref{eq:pfhighint54}) and $\underline{\xbold}_{cp}[i](j)$ corresponds to the $i^{\text{th}}$ element of the vector at the $j^{\text{th}}$ time slot. In the following, the achievable secrecy rate and power allocation for different messages are discussed in the case of the first time slot. Hence, for simplicity, the time index is omitted. The mutual information terms given in Theorem \ref{th:theorem_ach_weakint} are evaluated as follows. From Theorem \ref{th:theorem_ach_weakint}, $R_{p1}'$ and $R_{p2}'$ are set as $ 0.5\log\lb 1 + \frac{h_c^2 P_{p1}}{1 + h_d^2 P_{d2}}\rb$ and $0.5\log(1 + \frac{h_c^2 P_{p2}}{1 + h_c^2 P_{d2}})$, respectively. The first equation in (\ref{eq:weakint5}) becomes
\begin{align}
R_1 & \leq 0.5\log\lb 1 + \frac{\sigma_u^2 + h_d^2 P_{p1}}{1 + h_{c}^2 P_{d2} + h_c^2 P_{p2}}\rb - R_{p1}'.  \label{eq:pfweakcor2}
\end{align}
The second equation in (\ref{eq:weakint5}) becomes
\begin{align}
R_1 & \leq  0.5\log\lb 1 + \frac{h_d^2P_{p1}}{1 + h_c^2P_{d2}+h_c^2P_{p2}}\rb   + \min\lcb C_G, 0.5\log\lb 1 + \frac{\sigma_u^2}{1 + h_c^2P_{d2} + h_c^2 P_{p2}}\rb \rcb-R_{p1}'.  \label{eq:pfweakcor3}
\end{align}
The achievable rate for user $2$ becomes
\begin{align}
R_2 & \leq 0.5\log\lb 1 + \frac{\sigma_u^2 + h_d^2 P_{p2}}{1 + h_{d}^2 P_{d2} + h_c^2 P_{p1}}\rb - R_{p2}', \label{eq:pfweakcor4a} \\
R_2 & \leq 0.5\log\lb 1 + \frac{h_d^2P_{p2}}{1 + h_d^2P_{d2}+h_c^2P_{p1}}\rb  + \min\lcb C_G, 0.5\log\lb 1 + \frac{\sigma_u^2}{1 + h_d^2P_{d2} + h_c^2 P_{p1}}\rb \rcb-R_{p2}'.\label{eq:pfweakcor4}
\end{align}
The encoded message at transmitters $1$ and $2$ are
\begin{align}
& \xbold_1 = h_d\mathbf{w}_{1z} - h_c\mathbf{w}_{2z} + \mathbf{x}_{p1}, \text{ and } \nonumber \\
&  \xbold_2 = h_d\mathbf{w}_{2z} - h_c\mathbf{w}_{1z} + \mathbf{x}_{p2} + \mathbf{x}_{d2}.  \label{eq:pfweakcor5}
\end{align}
To simplify the power allocation, the variance of $\mathbf{w}_{1z}$ and $\mathbf{w}_{2z}$ are chosen to be the same, i.e, $\sigma_{1z}^2=\sigma_{2z}^2 = \sigma_z^2$. In order to satisfy the power constraint at the transmitters, the following conditions need to be satisfied.
\begin{align}
& (h_d^2 + h_c^2)\sigma_z^2 + P_{p1} \!\leq \! P_1 \text{ and } (h_d^2 + h_c^2)\sigma_z^2 + P_{p2} + P_{d2} \!\leq\! P_2, \label{eq:pfweakcor6}
\end{align}
where $P_i = \beta_i P \: (i=1,2)$, $0 \leq \beta_i \leq 1$ and $P$ is the maximum power available at either transmitter. The power for the non-cooperative private message, cooperative private message and dummy message are chosen as follows:
\begin{align}
& \sigma_z^2 = \frac{\theta_1}{\theta_1 + \theta_2}\frac{P_1}{h_d^2 + h_c^2}, P_{p1} = \frac{\theta_2}{\theta_1 + \theta_2}P_1, \nonumber \\
& P_{p2} = \frac{\eta_1}{\eta_1 + \eta_2}P', P_{d2} = \frac{\eta_2}{\eta_1 + \eta_2}P', \text{ and } \nonumber \\
& P' = (P_2 - (h_d^2 + h_c^2)\sigma_z^2)^{+}. \label{eq:pfweakcor7}
\end{align}
where $(\theta_i,\eta_i) \in [0,1]$. The parameters $\theta_i$ and $\eta_i$ act as power splitting parameters for transmitters $1$ and $2$, respectively. The parameter $\beta_i$ acts as a power control parameter. Hence, $\theta_i$, $\eta_i$ and $\beta_i$ are chosen such that the rates in (\ref{eq:pfweakcor2})-(\ref{eq:pfweakcor4}) are maximized, and the minimum of (\ref{eq:pfweakcor2}) and (\ref{eq:pfweakcor3}) gives the achievable secrecy rate for transmitter~$1$ i.e., $R_1^*(1)$; and the minimum of (\ref{eq:pfweakcor4a}) and (\ref{eq:pfweakcor4}) give the achievable secrecy rate for the transmitter~$2$ i.e., $R_2^*(1)$. This completes the proof.
\subsection{Proof of Theorem \ref{th:theorem_ach_highint}}\label{sec:appendhighintf}
In contrast to the achievable scheme for the weak/moderate interference regime, the dummy message sent by one of the users $i$ is required to be decodable at the receiver~$j$ $(j \neq i)$. Intuitively, since the cross links are stronger than the direct links, stochastic encoding alone is not sufficient to ensure secrecy of the non-cooperative private message. Hence, the dummy message sent by transmitter~$i$ acts as a self-jamming signal, preventing receiver~$i$ from decoding the message from the other transmitter~$j \neq i$. At the same time, ensuring that the dummy message is decodable at receiver~$j$ enables receiver~$j$ to cancel the interference caused by the dummy message, allowing it to decode its own message. Thus, although the cross-links are strong, receiver~$i$ is unable to decode the message from transmitter~$j$ because of the jamming signal; and this helps user~$j$ achieve a better rate. In the next time slot, user~$i$ can achieve a better rate by exchanging the roles of users~$i$ and $j$. The
proof involves analyzing the error probability at the encoder and the decoder along with equivocation computation. The conditions for the encoding error to go to zero and the choice of $\ubold_1$ and $\ubold_2$ remain the same as in the proof of Theorem~\ref{th:theorem_ach_weakint}. The details of the probability of error analysis are as follows. 
The conditions for encoding error to go to zero and the choice of $\ubold_1$ and $\ubold_2$ remains same as in the proof of Theorem \ref{th:theorem_ach_weakint}. Hence, in the following, error at the decoder is analyzed.\\

\textbf{Decoding error:} Define the following event
\begin{align}
E_{ijkl} = \lcb \lb \ybold_1^N, \xbold_{p1}^N(i,j), \ubold_1^N(k),  \xbold_{d2}^N(l)\rb \in T_{\epsilon}^{N}\rcb. \label{eq:pfhighint3}
\end{align}
Without loss of generality, assume that transmitter $1$ and $2$ sends $\xbold_1^N(1,1,1,1)$ and $\xbold_2^N(1,1,1,1)$ with $(\tildewcpone,\tildewcptwo) =(1,1)$, respectively. An error occurs if the transmitted and received codewords are not jointly typical or a wrong codeword is jointly typical with the received codewords. Then by the union of events bounds
\begin{align}
\lambda_{e}^{(n)} & = P\lb E_{1111}^c \bigcup \cup_{i \neq 1, j \neq 1, k \neq 1, l \neq 1}E_{ijkl}\rb  \leq P(E_{1111}^c) + P( \cup_{i \neq 1, j \neq 1, k \neq 1, l \neq 1}E_{ijkl}).  \label{eq:pfhighint4}
\end{align}
From the joint AEP \cite{cover-infotheory-2012}, $P(E_{1111}^c) \rightarrow 0$ as $N \rightarrow \infty$. 

When $i\neq 1$, $j\neq 1$, and $(k,l) = (1,1)$, then 
\begin{align}
\lambda_{ij11} & = \displaystyle\sum_{i\neq 1, j\neq 1} P(E_{ij11}), \nonumber \\
& \leq 2^{N\lsqb R_{p1}+R_{p1}'\rsqb}\sumAEP P(\xbold_{p1}^N)P(\ubold_1^N)P(\xbold_{d2}^N)P(\ybold_1^N|\ubold_1^N, \xbold_{d2}^N), \nonumber \\
& \leq 2^{N\lsqb R_{p1}+R_{p1}' - I(\xbold_{p1};\ybold_1|\ubold_1,\xbold_{d2}) + 5\epsilon\rsqb}. \label{eq:pfhighint8}
\end{align}
Hence, $\lambda_{ij11} \rightarrow 0$ as $N \rightarrow \infty$, if
\begin{align}
R_{p1} + R_{p1}' \leq  I(\xbold_{p1};\ybold_1|\ubold_1,\xbold_{d2}). \label{eq:pfhighint9}
\end{align}
Hence, if the condition in (\ref{eq:pfhighint9}) is satisfied, then the probability of error $\lambda_{i111}$ and $\lambda_{1j11}$, also goes to zero. When $k\neq 1$ and $(i,j,l) = (1,1,1)$, then
\begin{align}
\lambda_{11k1} & = \displaystyle\sum_{k \neq 1 } P(E_{11k1}), \nonumber \\
& \leq 2^{N\tildercpone}\sumAEP P(\xbold_{p1}^N)P(\ubold_1^N)P(\xbold_{d2}^N)P(\ybold_1^N|\xbold_{p1}^N, \xbold_{d2}^N), \nonumber \\
& \leq 2^{N\lsqb \tildercpone - I(\ubold_1;\ybold_1|\xbold_{p1},\xbold_{d2}) + 5\epsilon\rsqb}. \label{eq:pfhighint10}
\end{align}
Hence, $\lambda_{11k1} \rightarrow 0$ as $N \rightarrow \infty$, if 
\begin{align}
\tildercpone \leq I(\ubold_1;\ybold_1|\xbold_{p1},\xbold_{d2}). \label{eq:pfhighint11}
\end{align}
When $l\neq 1$ and $(i,j,k)=(1,1,1)$, then
\begin{align}
\lambda_{111l} & = \displaystyle\sum_{l \neq 1 }P(E_{111l}), \nonumber \\
& \leq 2^{NR_{d2}}\sumAEP P(\xbold_{p1}^N)P(\ubold_1^N)P(\xbold_{d2}^N)P(\ybold_1^N|\xbold_{p1}^N, \ubold_{1}^N), \nonumber \\
& \leq 2^{N\lsqb R_{d2} - I(\xbold_{d2};\ybold_1|\xbold_{p1},\ubold_1) + 5\epsilon\rsqb}. \label{eq:pfhighint12}
\end{align}
Hence, $\lambda_{111l} \rightarrow 0$ as $N \rightarrow \infty$, if
\begin{align}
R_{d2} \leq I(\xbold_{d2};\ybold_1|\xbold_{p1},\ubold_{1}). \label{eq:pfhighint13}
\end{align}
When $i \neq 1,j \neq 1,k \neq 1$, and $l=1$, then
\begin{align}
\lambda_{ijk1} & = \displaystyle\sum_{i \neq 1,j \neq 1,k \neq 1}P(E_{ijk1}), \nonumber \\
& \leq 2^{N\lsqb \tildercpone + R_{p1} + R_{p1}'\rsqb} \sumAEP P(\xbold_{p1}^N)P(\ubold_1^N)P(\xbold_{d2}^N)P(\ybold_1^N|\xbold_{d2}^N), \nonumber \\
& \leq 2^{N\lsqb \tildercpone + R_{p1} + R_{p1}' - I(\ubold_1,\xbold_{p1};\ybold_1|\xbold_{d2}) + 5\epsilon\rsqb}. \label{eq:pfhighint18}
\end{align}
Hence, $\lambda_{ijk1} \rightarrow 0$ as $N \rightarrow \infty$, if
\begin{align}
\tildercpone + R_{p1} + R_{p1}' \leq I(\ubold_1,\xbold_{p1};\ybold_1|\xbold_{d2}). \label{eq:pfhighint19}
\end{align}
Hence, if the condition in (\ref{eq:pfhighint19}) is satisfied, then the probability of error $\lambda_{i1k1}$ and $\lambda_{1jk1}$, also goes to zero. When $k\neq 1,l \neq 1$, and $(i,j) = (1,1)$, then
\begin{align}
\lambda_{11kl} & = \displaystyle\sum_{k\neq 1,l \neq 1}P(E_{11kl}), \nonumber \\
& \leq 2^{N\lsqb \tildercpone + R_{d2}\rsqb} \sumAEP P(\xbold_{p1}^N)P(\ubold_1^N)P(\xbold_{d2}^N)P(\ybold_1^N|\xbold_{p1}^N), \nonumber \\
& \leq 2^{N\lsqb \tildercpone + R_{d2} - I(\ubold_1,\xbold_{d2};\ybold_1|\xbold_{p1}) + 5\epsilon\rsqb}. \label{eq:pfhighint20}
\end{align}
Hence, $\lambda_{11kl} \rightarrow 0$ as $N \rightarrow \infty$, if
\begin{align}
\tildercpone + R_{d2} \leq I(\ubold_1,\xbold_{d2};\ybold_1|\xbold_{p1}). \label{eq:pfhighint21}
\end{align}
When $i\neq 1,j \neq 1,l \neq 1$, and $k=1$, then
\begin{align}
\lambda_{ij1l} & = \displaystyle\sum_{i\neq 1,j \neq 1,l \neq 1}P(E_{ij1l}), \nonumber \\
& \leq 2^{N\lsqb  R_{p1} + R_{p1}' + R_{d2}\rsqb} \sumAEP P(\xbold_{p1}^N)P(\ubold_1^N)P(\xbold_{d2}^N)P(\ybold_1^N|\ubold_{1}^N), \nonumber \\
& \leq 2^{N\lsqb  R_{p1} + R_{p1}' + R_{d2} - I(\xbold_{p1},\xbold_{d2};\ybold_1|\ubold_{1}) + 5\epsilon\rsqb}. \label{eq:pfhighint26}
\end{align}
Hence, $\lambda_{ij1l} \rightarrow 0$ as $N \rightarrow \infty$, if
\begin{align}
R_{p1} + R_{p1}' + R_{d2} \leq I(\xbold_{p1},\xbold_{d2};\ybold_1|\ubold_{1}). \label{eq:pfhighint27}
\end{align}
Hence, if the condition in (\ref{eq:pfhighint27}) is satisfied, then the probability of error $\lambda_{i11l}$ and $\lambda_{1j1l}$, also goes to zero. 

When $i \neq 1,j \neq 1,k \neq 1, \text{ and } l \neq 1$, then
\begin{align}
\lambda_{1jkl} & = \displaystyle\sum_{i \neq 1,j \neq 1,k \neq 1, l \neq 1}P(E_{ijkl}), \nonumber \\
& \leq 2^{N\lsqb  R_{p1} + R_{p1}' + \tildercpone + R_{d2}\rsqb} \sumAEP P(\xbold_{p1}^N)P(\ubold_1^N)P(\xbold_{d2}^N)P(\ybold_1^N), \nonumber \\
& \leq 2^{N\lsqb  R_{p1} + R_{p1}' + \tildercpone + R_{d2} - I(\ubold_1,\xbold_{p1},\xbold_{d2};\ybold_1) + 5\epsilon\rsqb}. \label{eq:pfhighint32}
\end{align}
Hence, $\lambda_{ijkl} \rightarrow 0$ as $N \rightarrow \infty$, if
\begin{align}
R_{p1} + R_{p1}' + \tildercpone + R_{d2} \leq I(\ubold_1,\xbold_{p1},\xbold_{d2};\ybold_1). \label{eq:pfhighint33}
\end{align}
Hence, if the condition in (\ref{eq:pfhighint33}) is satisfied, then the probability of error in $\lambda_{i1kl}$ and $\lambda_{1jkl}$, also goes to zero. 

In a similar way, it can be shown that the probability of decoding error at receiver $2$ goes to zero if following condition is satisfied.
\begin{align}
\tildercptwo \leq I(\ubold_2;\ybold_2). \label{eq:pfhighint37}
\end{align}
Hence, the encoding and decoding error go to $0$ as $N \rightarrow \infty$, if (\ref{eq:pfhighint9}), (\ref{eq:pfhighint11}), (\ref{eq:pfhighint13}), (\ref{eq:pfhighint19}), (\ref{eq:pfhighint21}), (\ref{eq:pfhighint27}) and  (\ref{eq:pfhighint33}) and  (\ref{eq:pfhighint37}). Then, by applying Fourier-Motzkin procedure \cite{gamal-netinfotheory-2011} to these equations and the conditions for encoding error, the achievable rate in Theorem \ref{th:theorem_ach_highint} can be obtained.
\subsubsection{Equivocation computation}\label{sec:eqvochighint} The equivocation at receiver~$2$ is bounded as follows. As the non-intended cooperative private message is canceled completely at receiver~$2$, it suffices to show the following, as mentioned in Appendix~\ref{sec:eqvocweakint}.
\begin{align}
H(W_{p1}|\ybold_2^N) \geq N\lsqb R_{p1} - \epsilon_s\rsqb. \label{eq:eqvocrxtwo4}
\end{align}
Consider the following:
\begin{align}
 H(W_{p1}|\ybold_2^N) & \geq H(W_{p1}|\ybold_2^N,\ubold_2^N), \nonumber \\
& = H(W_{p1},\ybold_2^N|\ubold_2^N) - H(\ybold_2^N|\ubold_2^N), \nonumber \\
& \stackrel{(a)}{=} H(W_{p1},\ybold_2^N,\xbold_{p1}^N,\xbold_{d2}^N|\ubold_2^N) - H(\xbold_{p1}^N,\xbold_{d2}^N|W_{p1},\ybold_{2}^N,\ubold_{2}^N)  - H(\ybold_2^N|\ubold_2^N), \nonumber \\
& = H(\xbold_{p1}^N,\xbold_{d2}^N|\ubold_2^N) + H(W_{p1},\ybold_2^N|\xbold_{p1}^N,\xbold_{d2}^N,\ubold_2^N)- H(\ybold_2^N|\ubold_2^N)  \nonumber \\
&\qquad  - H(\xbold_{p1}^N,\xbold_{d2}^N|W_{p1},\ybold_{2}^N,\ubold_{2}^N), \nonumber 
\\
& \geq H(\xbold_{p1}^N,\xbold_{d2}^N|\ubold_2^N) + H(\ybold_2^N|\xbold_{p1}^N,\xbold_{d2}^N,\ubold_2^N) - H(\ybold_2^N|\ubold_2^N) \nonumber \\
& \qquad - H(\xbold_{p1}^N,\xbold_{d2}^N|W_{p1},\ybold_{2}^N,\ubold_{2}^N), \nonumber \\
& = R_{p1} + R_{p1}' + R_{d2} - I(\xbold_{p1}^N,\xbold_{d2}^N;\ybold_2^N|\ubold_2^N)  - H(\xbold_{p1}^N,\xbold_{d2}^N|W_{p1},\ybold_{2}^N,\ubold_{2}^N), \label{eq:eqvocrxtwo5}
\end{align}
where (a) is obtained using the relation: $H(W_{p1},\ybold_2^N,\xbold_{p1}^N,\xbold_{d2}^N|\ubold_2^N) = H(W_{p1},\ybold_2^N|\ubold_2^N) + $ \\$H(\xbold_{p1}^N,\xbold_{d2}^N|W_{p1},\ybold_{2}^N,\ubold_{2}^N)$.
The second term in (\ref{eq:eqvocrxtwo5}) is upper bounded as follows.
\begin{align}
I(\xbold_{p1}^N,\xbold_{d2}^N;\ybold_2^N|\ubold_2^N) \leq N I(\xbold_{p1},\xbold_{d2};\ybold_2|\ubold_2) + N\epsilon'. \label{eq:eqvocrxtwo7}
\end{align}
The above bound can be obtained by using similar steps as used in the proof of Lemma~\ref{lm1} in Appendix~\ref{sec:usefullemma}.

To bound the last term in (\ref{eq:eqvocrxtwo5}),  consider the joint decoding of $W_{p1}'$ and $W_{d2}$, assuming that the receiver $2$ is given $W_{p1}$ and $\ubold_2^N$ as side information. By following similar steps as in Appendix~\ref{sec:eqvocweakint}, it is possible to show that the probability of error is arbitrarily small for large $N$, provided the following conditions are satisfied:
\begin{align}
& R_{P1}' \leq I(\xbold_{p1};\ybold_{2}|\xbold_{d2},\ubold_2), \nonumber \\
& R_{d2} \leq I(\xbold_{d2};\ybold_{2}|\xbold_{p1},\ubold_2), \nonumber \\
& R_{p1}' + R_{d2} \leq I(\xbold_{p1},\xbold_{d2};\ybold_{2}|\ubold_2).  \label{eq:eqvocrxtwo11}
\end{align}
When the conditions in (\ref{eq:eqvocrxtwo11}) are satisfied and for sufficiently large $N$, the following bound is obtained using Fano's inequality:
\begin{align}
H(\xbold_{p1}^N,\xbold_{d2}^N|W_{p1} = w_{p1}, \ybold_1^N, \ubold_2^N) & \leq N\delta_2.  \label{eq:eqvocrxtwo12}
\end{align}
Finally, the last term in (\ref{eq:eqvocrxtwo5}) is bounded as follows.
\begin{align}
 H(\xbold_{p1}^N,\xbold_{d2}^N|W_{p1},\ybold_{2}^N,\ubold_{2}^N)  &= \displaystyle\sum_{w_{p1}}P(w_{p1})H(\xbold_{p1}^N,\xbold_{d2}^N|W_{p1}=w_{p1}, \ybold_1^N, \ubold_2^N), \nonumber \\
& \leq N\delta_2.  \label{eq:eqvocrxtwo13}
\end{align}
Using (\ref{eq:eqvocrxtwo7}) and (\ref{eq:eqvocrxtwo13}), (\ref{eq:eqvocrxtwo5}) becomes
\begin{align}
& H(W_{p1}|\ybold_2^N) \nonumber \\
&  \geq N\lsqb R_{p1} + R_{p1}' + R_{d2}-  I(\xbold_{p1},\xbold_{d2};\ybold_2|\ubold_2) - (\delta_2 +\epsilon')\rsqb, \label{eq:eqvocrxtwo14}
\end{align}
By choosing $R_{p1}' = I(\xbold_{p1};\ybold_{2}|\ubold_2)-\epsilon_2'$ and $R_{d2} =  I(\xbold_{d2};\ybold_2|\xbold_{p1},\ubold_2)-\epsilon_2''$ secrecy of the non-cooperative private part is ensured. Thus,
\begin{align}
 H(W_{p1}|\ybold_2^N) \geq N \lsqb R_{p1} - \epsilon_2\rsqb.  \label{eq:eqvocrxtwo15}
\end{align}
This completes the proof.
\subsection{Proof of Corollary \ref{cor:cor_ach_highint}}\label{sec:corappendhighintf}
In the first and second time slots, transmitters $1$ and $2$ send the following encoded messages.
\begin{align}
& \xbold_1(1) = \underline{\xbold}_{cp}[1](1) + \xbold_{p1}(1), \text{ and } \xbold_2(1) = \underline{\xbold}_{cp}[2](1) + \xbold_{d2}(1), \nonumber \\
& \xbold_1(2) = \underline{\xbold}_{cp}[1](2) + \xbold_{d1}(2), \text{ and } \xbold_2(2) = \underline{\xbold}_{cp}[2](2) + \xbold_{p2}(2). \label{eq:pfcor1}
\end{align}
In the following, the achievable secrecy rate and power allocation for different messages are discussed in case of first time slot. In the following, the time index is omitted. The mutual information given in Theorem \ref{th:theorem_ach_highint} are evaluated as follows. From Theorem \ref{th:theorem_ach_highint}, $R_{p1}'$ and $R_{d2}$ are set as: $ 0.5\log\lb 1 + \frac{h_c^2 P_{p1}}{1 + h_d^2 P_{d2}}\rb$ and $0.5\log(1 + h_d^2 P_{d2})$, respectively. Consider the first equation in (\ref{eq:highint5}).
\begin{align}
R_1 & =   \min \lsqb 0.5\log(1 + \sigma_{u1}^2  + h_d^2 P_{p1}), 0.5\log(1 + h_d^2 P_{p1}) + \min\lcb 0.5\log(1 + \sigma_{u1}^2), C\rcb\rsqb - R_{p1}',  \label{eq:pfcor2}
\end{align}
The second equation in (\ref{eq:highint5}) reduces to following.
\begin{align}
R_1 & \leq \lsqb 0.5\log(1 + \sigma_{u1}^2 + h_d^2 P_{p1} + h_c^2 P_{d2}), 0.5\log(1 + \sigma_{u1}^2 + h_c^2 P_{d2}) + \min\lcb 0.5\log(1 + \sigma_{u1}^2),C\rcb, \right. \nonumber \\
&\qquad \qquad \qquad  \left. 0.5\log(1 + h_d^2 P_{p1}) + 0.5\log(1 + \sigma_{u1}^2 + h_c^2 P_{d2})\rsqb - (R_{p1}' + R_{d2}). \label{eq:pfcor3}
\end{align}
The third equation in (\ref{eq:highint5}) reduces to following.
\begin{align}
R_1 & \leq 0.5\log(1 + h_d^2 P_{p1} + h_c^2 P_{d2}) + 0.5\log(1 + \sigma_{u1}^2 + h_c^2 P_{d2}) - (R_{p1}' + 2R_{d2}). \label{eq:pfcor4}
\end{align}
Consider the fourth equation in (\ref{eq:highint5}).
\begin{align}
R_2 & = \min\lcb 0.5\log\lb 1 + \frac{\sigma_{u2}^2}{1 + h_d^2 P_{d2} + h_c^2 P_{p1}}\rb,C\rcb. \label{eq:pfcor5}
\end{align}
As $R_{d2} = 0.5\log(1 + h_d^2 P_{d2}) < 0.5\log(1 + h_c^2 P_{d2})$, $R_{d2}$ satisfies the last inequality in (\ref{eq:highint5}). Now, consider the power allocation for the private cooperative message, non-cooperative private message and dummy message as shown below. The encoded messages at transmitters $1$ and $2$ are:
\begin{align}
& \xbold_1 = h_d\mathbf{w}_{1z} - h_c\mathbf{w}_{2z} + \mathbf{x}_{p1}, \text{ and } \xbold_2 = h_d\mathbf{w}_{2z} - h_c\mathbf{w}_{1z} + \mathbf{x}_{d2}.  \label{eq:pfcor6}
\end{align}
To simplify the power allocation, the variance of $\mathbf{w}_{1z}$ and $\mathbf{w}_{2z}$ are chosen to be same. In order to satisfy the power constraint, the following conditions need to be satisfied.
\begin{align}
& (h_d^2 + h_c^2)\sigma_z^2 + P_{p1} \leq P_1, \text{ and } (h_d^2 + h_c^2)\sigma_z^2 + P_{d2} \leq P_2, \label{eq:pfcor7}
\end{align}
where $P_i = \beta_i P \: (i=1,2)$ and $0 \leq \beta_i \leq 1$. The power for the non-cooperative private message, cooperative private message and dummy message are chosen as follows:
\begin{align}
& \sigma_z^2 = \frac{\theta_1}{\theta_1 + \theta_2}\frac{P_1}{h_d^2 + h_c^2}, P_{p1} = \frac{\theta_2}{\theta_1 + \theta_2}P_1, \text{ and } P_{d2} = (P_2 - (h_d^2 + h_c^2)\sigma_z^2)^+. \label{eq:pfcor8}
\end{align}
where $\theta_i \in [0,1]$. The parameter $\theta_i$ and $\beta_i$ are the power splitting and power control parameter, respectively. Hence, $\theta_i$ and $\beta_i$ are chosen such that the rate in (\ref{eq:pfcor2})-(\ref{eq:pfcor5}) are maximized and the minimum of (\ref{eq:pfcor2})-(\ref{eq:pfcor4}) gives the achievable secrecy rate for the transmitter $1$ i.e., $R_1^*(1)$ and (\ref{eq:pfcor5}) gives the achievable secrecy rate for the transmitter $2$. i.e., $R_2^*(1)$. In a similar way, the achievable secrecy rate $R_1^*(2)$ and $R_1^*(2)$ can be determined in the second time slot. This completes the proof.
\bibliographystyle{IEEEtran}
\bibliography{IEEEabrv,secrecy_work_journal_ver}
\end{document}